%% file: main.tex
\documentclass[12pt]{article}

\usepackage[margin=1in]{geometry}

\usepackage[onehalfspacing]{setspace}

\expandafter\def\expandafter\normalsize\expandafter{%
    \normalsize%
    \setlength\abovedisplayskip{8pt}%
    \setlength\belowdisplayskip{8pt}%
    \setlength\abovedisplayshortskip{-8pt}%
    \setlength\belowdisplayshortskip{2pt}%
}

\usepackage[nodayofweek,level]{datetime}
\usepackage{graphicx}
\usepackage{lscape}
\usepackage{pdflscape}
\usepackage{afterpage}
\usepackage{pdfpages}
\usepackage{float}
\usepackage{booktabs}
\usepackage[utf8]{inputenc}
\usepackage{indentfirst}
\usepackage{arydshln}
\usepackage{arydshln}
\usepackage{tikz}
\usetikzlibrary{decorations.pathmorphing}
\usepackage{lipsum}
\usepackage{pgffor}
\usepackage{dsfont}
\usepackage{amsfonts}
\usepackage[fleqn]{amsmath} %
\usepackage{bbm}
\allowdisplaybreaks
\usepackage{mathtools}
\newcommand\independent{\protect\mathpalette{\protect\independenT}{\perp}}
\def\independenT#1#2{\mathrel{\rlap{$#1#2$}\mkern2mu{#1#2}}}
\DeclareMathOperator*{\argmax}{\arg\!\max}

\usepackage{xfrac}
\usepackage{amssymb}
\usepackage{bigints}
\usepackage{graphicx}
\usepackage{url}				%
\usepackage{caption}
\usepackage{subcaption} %
\usepackage{enumitem}
\usepackage{relsize,exscale}
\usepackage{multirow}
\usepackage{natbib}
\usepackage[ruled,vlined]{algorithm2e}

\setcitestyle{authoryear, open={(},close={)}}
\usepackage{authblk}
\usepackage[normalem]{ulem}
\usepackage{xcolor} %
\usepackage{adjustbox}
\usepackage{tikz}
\usepackage{pgfplots}
\pgfplotsset{compat=1.18}
\usetikzlibrary{calc}
\usepackage{float}
\usepackage{datetime}
\usepackage{listings}
\usepackage[bottom]{footmisc}

\newdateformat{monthyeardate}{%
	\monthname[\THEMONTH] \THEYEAR}

\usepackage{titlesec}
\titleformat*{\section}{\large\bfseries}
\titleformat*{\subsection}{\large\bfseries}
\usepackage{booktabs}
\usepackage{rotating}


\newcounter{parentnumber}
\usepackage{theorem}
\newtheorem{theorem}{Theorem}

\newtheorem{assumption}{Assumption}

\newtheorem{corollary}{Corollary}

\newtheorem{definition}{Definition}

\newtheorem{Lemma}{Lemma}

\newenvironment{proof}[1][Proof]{%
  \noindent\textbf{#1.}%
}{%
  \ \rule{0.5em}{0.5em}%
  \vspace{0.5em} %
}
\usepackage{caption}
\usepackage{colortbl}
\usepackage{xcolor}

\definecolor{lightblue}{RGB}{200, 220, 255}      %
\definecolor{lightred}{RGB}{255, 200, 200}       %
\definecolor{lightgreen}{RGB}{200, 255, 200}    %
\definecolor{lightpurple}{RGB}{220, 200, 255}

\usepackage[pagebackref]{hyperref} %
\hypersetup{
	colorlinks=true, breaklinks=true, bookmarks=true,bookmarksnumbered,
	urlcolor=blue, linkcolor=blue, citecolor=blue, %
	pdftitle={}, %
	pdfauthor={\textcopyright}, %
	pdfsubject={}, %
	pdfkeywords={}, %
	pdfcreator={pdfLaTeX}, %
	pdfproducer={LaTeX with hyperref and ClassicThesis} %
}

\usepackage{cleveref}

\renewcommand*{\backref}[1]{}
\renewcommand*{\backrefalt}[4]

  \author{Santiago Acerenza,  Julian Martinez-Iriarte, \\ Alejandro S\'{a}nchez-Becerra, Pietro Emilio Spini\thanks{\textit{Acerenza}: Universidad ORT,  acerenza@ort.edu.uy. \textit{Martinez-Iriarte}: UC Santa Cruz, jmart425@ucsc.edu. \textit{Spini}: University of Bristol, pietro.spini@bristol.ac.uk. \textit{S\'{a}nchez-Becerra}: Emory University, alejandro.sanchez.becerra@emory.edu. Manuel Hidalgo Sola provided excellent research assistance. We thank Abhishek Ananth, Mengsi Gao, Leonard Goff, Florian Gunsilius, Sukjin Han, Maximilian Kasy, Natalia Lazzati, Michael Leung, Arthur Lewbel, Jessie Li, Sokbae ``Simon'' Lee, Yuya Sasaki, Alex Torgovitsky, and seminar participants at Encounters in Econometric Theory at Oxford 2024, University of Bristol, LAMES 2024, Southern Economic Association 2024, UC Santa Cruz, Emory Data and Decision Sciences, Georgia Econometrics Workshop 2025, EAYE 2025, IAAE 2025, CEU, CEME 2025, USC, UT Knoxville, 
 for the helpful discussion on the earlier version of this paper titled \textit{``Bounds on treatment effect parameters with spillovers, unrestricted choice and strategic interactions: a computational approach."}  } }

\begin{document}
	\setstretch{1}
	\title{{\LARGE 
  Sharp bounds for within-household \\ encouragement designs with interference \\

 }}

	\maketitle

	\newsavebox{\tablebox} \newlength{\tableboxwidth}

\begin{center}

First Draft: May 31st, 2024 \\

\

		\large{\textbf{Abstract}}
	\end{center}

Experiments with spillovers create incentives for strategic behavior: when one household member's treatment can affect another, take-up decisions become interdependent.  We propose an instrumental-variables framework grounded in game theory, in which take-up is a Nash equilibrium among a small number of agents, and all variables are discrete. Under minimal assumptions, without specifying how equilibria are selected, we derive sharp non-parametric bounds for direct, indirect, and offer-mediated effects. Game-theoretic restrictions such as supermodularity, symmetry, and dominance can be layered on via a tractable linear program, and we characterize when each does or does not tighten the identified set. We illustrate the usefulness of our approach to study household experiments by reanalyzing a banking intervention in Kenya. 

	\



	\newpage

	\doublespacing

\linespread{1.5}
\selectfont

\input{main_text.tex}

\linespread{0.5}

\renewcommand{\refname}{References}

\bibliographystyle{apalike}

\bibliography{references_all}

\linespread{1.5}
\input{appendix_all.tex}

\end{document}

%% file: main_text.tex
\section{Introduction}\label{Sintro}

Policy experiments in which household heads or their spouses are offered treatment are widespread in the social sciences. Many such studies estimate the impacts of gender-targeted conditional cash transfers to parents \citep{armand2020effect,benhassine2015turning,haushofer2016short}, gender-targeted family planning \citep{d2022women,mccarthy2019intimate}, banking vouchers \citep{dupas2019effect}, and voter turnout encouragement \citep{foos2017all,nickerson2008voting,sinclair2012detecting}. Other household experiments, particularly in educational contexts, target siblings \citep{dizon2014parents,rogers2017intervening}. In these examples, understanding the household decision process has important policy implications \citep{vermeulen2002collective}.

These household experiments can exhibit two forms of interference. First, treating one spouse can have a direct spillover on the other's outcome. Second, \textit{individual} treatment offers can shape the household's \textit{collective} incentives to decide who takes up treatment. In this ``two-stage'' interference setting, treatment effects may be mediated by strategic take-up decisions, posing new identification challenges. Most of the policy evaluation literature has side-stepped the problem by assuming that agents are not strategic in order to identify complier-type effects \citep{Blackwell2017,DITRAGLIA20231589,Vazquez-Bare2020,ryu2024local,kormos2023} or weighted averages of local effects \citep{KangImbens2016,ImaiJiangMalani2021}. Nevertheless, the presumption that agents do not act in anticipation of spillovers is controversial and harder to justify among household members who interact frequently with each other. In contrast, in settings like firm entry and technology adoption, the IO literature routinely models game-theoretic payoffs \citep{Lazzati2015,HanBalat2023}, and aims to partially identify them \citep{molinari2008identification,ciliberto2009market}.

We combine insights from both literatures in our household experiment setting. We propose a simple and transparent approach to partially identify treatment effect parameters, including various direct and indirect average effects and offer-mediated effects, that is robust to arbitrary outcome and treatment interference within the household. The key idea is to capture all relevant heterogeneous preferences and strategic behaviors in a sufficiently rich space of latent types $\mathcal{T}$. While the distribution over such types is unknown, it is constrained by support restrictions and by the joint distribution of observables, including outcomes, take-up, and instruments, which are assumed to have discrete support. Computationally, we obtain bounds on treatment effects via a tractable linear program à la \cite{balke1997bounds}. Our theoretical results apply to any finite number of agents. In our empirical application from \citet{dupas2019effect}, we focus on households with two agents: spouses.

To model strategic behavior in the potential outcomes framework, we conceptualize treatment take-up as a within-household game in which the policymaker's offers shift the players' incentives. We introduce a novel notion of take-up best-response functions (BRFs): an individual's treatment can depend on the entire vector of household offers and the treatment decisions of others. Because games can have multiple Nash equilibria, we complete the model by introducing a class of potential equilibrium-selection rules.\footnote{\cite{Lazzati2025} develop a test for Nash equilibrium play in a game with monotone best responses.} Modeling best responses allows us to define counterfactual estimands for (i) ``fixed-allocation'' effects, such as average direct and indirect effects, where the policymaker enforces treatment allocations, and (ii) ``offer-mediated'' effects, where we study the impact of incentivizing another player to take up treatment. The framework also explains what classic instrumental-variable estimands capture when compliance is strategic. In Section \ref{sec:waldstatistic}, we decompose the Wald statistic and show when it admits a meaningful causal interpretation.

In our framework, many economic assumptions about households can be recast as restrictions on the set $\mathcal{M}$ of admissible distributions over the latent types. We show how researchers can encode these restrictions as optional linear constraints in a single linear program and compute the corresponding identified set $\Theta(\mathcal{M})$ automatically for any combination of restrictions.\footnote{For example, in Appendix \ref{sec:ijm} we impose the restrictions of \citet{ImaiJiangMalani2021} in our empirical application.} The same program delivers bounds when no restrictions are imposed. To assess identifying power, we call a restriction $\mathcal{M}'\subseteq\mathcal{M}$ \textit{irrelevant} if $\Theta(\mathcal{M}')=\Theta(\mathcal{M})$, that is, if it does not tighten the identified set for the parameter $\theta$.

Imposing ``strong strategic neutrality'' (SSN), under which individual take-up decisions are unresponsive to others' take-up, is often criticized as implausible, yet the usual alternative is to specify a full structural model of payoffs and equilibrium selection. Our framework sits between these extremes. We find that, for ``fixed-allocation'' estimands, SSN alone cannot tighten the identified set relative to the model without it.\footnote{Irrelevance does not mean that households are non-strategic. It means that whatever the data reveal about fixed-allocation parameters could equally have come from households satisfying SSN.} In Appendices \ref{sec:irrelevance_sequentialplay} and \ref{sec:supermodularity}, we show that restrictions weaker than SSN inherit this irrelevance by a sandwich argument. These include restrictions on the order of play, networked best responses, supermodularity, and submodularity, which have been historically important in the spillovers and peer effects literatures \citep{graham2023scenario,molinari2008identification,Lazzati2015,HanBalat2023}. We also prove two other irrelevance results unrelated to SSN. In Appendix \ref{sec:det_vs_stoc}, we show that restricting attention to deterministic rather than stochastic equilibrium-selection rules is irrelevant for more general estimands, including offer-mediated ones.\footnote{Absent further structure, even if outcomes, take-up, and instruments are discrete, there is a continuum of stochastic equilibrium-selection rules. In contrast, deterministic rules are finite, allowing us to recast the partial-identification problem as a computationally tractable linear program.} In Appendix \ref{sec:symmetry_refutability}, we consider two notions of best-response symmetry, equivariant and individual, and show that individual symmetry is irrelevant for fixed-allocation estimands. 

Our theoretical findings have practical implications. For ``fixed-allocation'' estimands, the combination of deterministic rules and SSN is irrelevant: the problem reduces to a canonical partial identification problem with endogenous treatment and exogenous instruments \citep{Manski1990,balke1994counterfactual,balke1997bounds,mourifie2015sharp,chesher2017generalized,Laffer2019,Gunsilius2020,kitagawa2021identification,HanYang2024,bai2024identifying,Swanson2018}. For these estimands, researchers can apply familiar tools robustly even under strategic behavior and arbitrary equilibrium selection.\footnote{Irrelevance of SSN requires a novel closure condition, which is trivially satisfied in the unrestricted model but can fail when other restrictions are imposed together with SSN.} However, whether a restriction is irrelevant depends on the parameter of interest.\footnote{For example, the joint distribution of best responses, analyzed by \cite{kline2012bounds}, can be defined as a parameter of interest in our framework but is not covered by the SSN irrelevance result in Theorem \ref{thm:equivalence_dominance_multiagent_distr}.} In particular, the irrelevance of SSN does not extend to offer-mediated estimands, for which it can tighten bounds, as we show in our Kenya banking application in Section \ref{sec:indirect_offer_effects}.

\textbf{Related Literature.}
Our computational approach builds on \citet{balke1994counterfactual,balke1997bounds}, who introduce probabilistic response functions as latent variables encoding the structural relationship between observables. In their setting, where there is no interference, each latent type specifies all individual potential treatments and outcomes, and take-up behavior reduces to four principal strata: always-takers, compliers, never-takers, and defiers. Strategic interaction substantially enriches this representation. In our setting, individual latent types must encode treatment responses for every relevant combination of household offers and other household members' treatment decisions.\footnote{An individual latent type contains the individual's best-response functions, one for each game induced by the possible values of the offers. The joint distribution of best-response functions can itself be an object of partial identification, as in \citet{kline2012bounds}, where no instruments are present.} Because the household is the unit of analysis, a household latent type combines the types of all its members. This representation is agnostic about how households are formed, allowing researchers to impose restrictions on latent-type sorting when these are supported by the design or institutional setting.

Our offer-mediated effects also connect the paper to the literature on causal mediation. Natural direct and indirect effects distinguish channels that operate directly through an individual's treatment from those that operate through an intermediate variable \citep{huber2019review}. In our household setting, the natural analogue of the mediator is a spouse's treatment decision. Unlike in standard mediation models, however, this mediator is not an externally assigned or passively realized intermediate outcome. Instead, it is a strategic choice that may depend on both household members' offers and be jointly determined with the focal individual's treatment through best-response behavior.

Existing mediation approaches obtain point identification by imposing substantial exogeneity and behavioral structure. Sequential ignorability, for example, restricts the joint relationship among treatment, the mediator, potential outcomes, and observed covariates \citep{imai2010general}. In our latent-type representation, an analogous condition would have to be imposed as a restriction on the joint distribution of household types rather than as a pointwise independence condition for fixed types. Other approaches use instruments for both treatment and the mediator and combine exclusion, monotonicity, exchangeability, and stability conditions \citep{rudolph2024using,frolich2017direct}. Translating those conditions exactly into our household game would require additional structure, because the mediator is itself another household member's strategic take-up decision. Our approach is complementary: rather than maintain those assumptions to obtain point identification, we allow arbitrary dependence among household potential outcomes and best responses and derive bounds under the restrictions encoded in the latent-type space. Mediation assumptions could still be incorporated as additional distributional restrictions when they are substantively credible.

Extending the \citet{balke1994counterfactual,balke1997bounds} formulation to strategic response behavior also requires confronting the possibility of multiple equilibria. With multiple equilibria, the mapping from individual response types to observable probabilities is not unique, which is a central feature of model incompleteness \citep{Jovanovic1989,tamer2003incomplete,honore2006bounds,Chesher2010,hahn2010bounds,chen2012bounds,kitagawa2021identification,HanBalat2023}. Rather than ignoring multiplicity or specifying a particular selection mechanism, we augment household latent types with equilibrium-selection rules. An arbitrary selection mechanism can then be represented through the distribution over augmented types and enters as part of the latent distribution in the linear program. This permits the analysis to remain agnostic about how households select among multiple equilibria.

Our paper also complements \citet{bai2024identifying}, who study bounds on treatment-effect parameters with multivalued outcomes, treatments, and instruments. Both papers generalize \citet{balke1994counterfactual,balke1997bounds} by representing partial identification as a linear program in which equality restrictions map latent types into observable behavior and additional restrictions discipline the admissible distribution of latent types. The two papers differ in emphasis. While \citet{bai2024identifying} work under SUTVA, our analysis focuses explicitly on within-group interference and constructs a latent-type space rich enough to represent strategic best responses, multiple equilibria, household symmetry, and restrictions on joint household behavior. Our main theoretical interest is in the identifying power and empirical content of these restrictions.

There is also a rich literature obtaining point-identified treatment effects for selected subpopulations by imposing restrictions on take-up behavior \citep{Blackwell2017,Vazquez-Bare2017,Vazquez-Bare2020,DITRAGLIA20231589,kormos2023,ImaiJiangMalani2021}. Identifying local effects with instrumental variables often relies on monotonicity assumptions \citep{blandhol2022tsls}, which become increasingly demanding with interference, when several instruments affect household incentives, and when take-up decisions are strategic. The number and interpretation of the corresponding latent subpopulations also become substantially more complicated.\footnote{In the original LATE framework of \citet{Imbens1994}, there are four latent treatment-response types: always-takers, compliers, never-takers, and defiers. Uniform monotonicity rules out the defier type. For a discussion of continuous and multivalued instruments, see \citet{Mogstad2021}.} Consequently, a point-identified local effect may refer to a narrow or difficult-to-interpret subpopulation. Our complementary approach focuses on population-average direct, indirect, and offer-mediated effects, while allowing these parameters to remain partially identified.

Finally, modeling strategic take-up connects our paper to the econometrics of games; see \citet{DePaula2013} for a review. \citet{Lazzati2015} combines individualized offer responses, strategic interaction, supermodularity, and treatment monotonicity to identify parameters governing best responses and equilibrium behavior. Closely related, \citet{HanBalat2023} study average and local treatment effects with endogenous equilibrium entry under restrictions that exclude direct effects of other firms' instruments on entry and impose strategic substitutability and a common direction of treatment response across players. These restrictions can be imposed as optional constraints on our latent-type space rather than as maintained features of the model.

Similarly, the joint distribution of best-response functions studied by \citet{kline2012bounds} can be written as a parameter of interest in our framework: it is linear in the distribution of latent types, so the identified sets of Section \ref{sec:set_identification} and the same linear program apply (see Appendix \ref{app:Kline_Tamer_2012}). Our irrelevance results do not cover such parameters, because they depend on best responses, including off-equilibrium ones, rather than only on potential outcomes. Restrictions on strategic interaction can be informative about them even when they are irrelevant for average direct and indirect effects.

The paper is organized as follows. In Section \ref{Sframework}, we introduce notation, the econometric framework, and the latent-type representation. In Section \ref{sec:ident_estim_infer}, we discuss partial identification, estimation, and inference. In Section \ref{sec:uselevel_restrictions}, we prove the main results concerning ``strong strategic neutrality'' (SSN). Section \ref{sec:waldstatistic} decomposes the Wald statistic. Section \ref{Sempirical} applies our methods to the setting studied by \citet{dupas2019effect}, and Section \ref{Sconclu} concludes. Additional results are collected in the Appendix.

\section{ Econometric Framework}\label{Sframework}
We start by introducing an example illustrating the type of issues that our methodology helps address. In a field experiment in Kenya, \cite{dupas2019effect} evaluated the impact of expanding access to bank accounts on entrepreneurial activity, for a region with relatively low access to formal banking, especially among women. Their study focused mostly\footnote{Female-headed households are also studied. We analyze them in Appendix \ref{appendix:female_headed}.} on dual-headed, heterosexual households, who did not have a bank account at baseline. Individual spouses were randomized to receive a nontransferable voucher waiving most savings‑account fees. Each spouse could then opt to open a bank account in their name. In practice, about 69\% of offered men and 65\% of offered women did, while no spouse who was not offered a voucher opened an account.\footnote{Spouses who were offered the voucher could also opt to enroll for a joint bank account: in the data, less than 5\% of households choose this option.} At endline the researchers measured whether each household member was receiving income from their own enterprise. In this context it is reasonable to expect spousal program take-up decisions to be coordinated to some degree, but modeling this econometrically presents some technical challenges.

\subsection{A ``Game-theoretic'' Causal Framework}

To tackle the analysis of treatment effects in studies like \cite{dupas2019effect}, we propose a new game-theoretic causal approach. We study experiments where individuals are encouraged to participate in an intervention. We assume that the researcher samples many i.i.d. clusters $g \in \{1,\ldots,G\}$ with $N$ members and observes outcomes, treatment, and instruments $(Y_{ig},D_{ig},Z_{ig})$, which have discrete support $(\mathcal{Y},\mathcal{D},\mathcal{Z})$.\footnote{Assuming that outcomes have discrete support is not essential, but simplifies the presentation of ideas.} In the case of \cite{dupas2019effect}, there are two spouses so that $N=2$. For a household member $i \in \{\text{``male'',``female''}\}$ we let $Z_{ig} \in \{0,1\}$ indicate whether they received the voucher, and $D_{ig} \in \{0,1\}$ whether they opened a bank account. Because randomization is at the individual level, within dual‑headed heterosexual households there are four treatment assignments: male‑only, female‑only, both, or neither.  We allow non-compliance: treatment take-up $(D_{ig},D_{-i,g})$ may differ from the offer status. More generally, $(Z_{-i,g},D_{-i,g})$ are vectors for $N > 2$.\footnote{In the language of network theory, $i$ is used to denote the ``ego'' (the focal individual which can be the male or female spouse) and $\{-i\}$ is the alter (the remaining spouse). The role of the ego and alter can be interchanged depending on whose outcomes we are analyzing. }

In \cite{dupas2019effect}, one member's access to the formal banking sector can have consequences for the economic activity of their spouse. In the exposure-mapping tradition \citep{aronow2017estimating}, we capture spillovers by defining a potential outcome as $Y_{ig}(d,d',z,z')$ and assuming it is not directly affected by offers, only by treatment status.\footnote{While Assumption \ref{AsExclu} is common in the literature on non-compliance, our framework could also be adapted to other settings. For example, in experiments with full compliance but non-random sample selection (selection into paid work, equipment usage, etc.), such as in \cite{Lee2009} and \cite{estrada2024spillover}, the offers $(z,z')$ coincide with the treatments, $D_{ig}$ indicates each person's extensive margin (selection), the outcome is the intensive margin (observed only when $D_{ig}=1$), and Assumption \ref{AsExclu} is replaced by $Y_{ig}(d,d',z,z') = d\, Y_{ig}(z,z')$.}

\begin{assumption}[Outcome Exclusion Restriction] \text{ }
\label{AsExclu}  $Y_{ig}(d,d',z,z')=Y_{ig}(d,d')$.
	\end{assumption}

If the intervention is expected to generate spillovers, it is natural for agents to take them into account in their decision-making process for enrolling. To describe strategic take-up decisions, we conceptualize the policymaker's offers $(z,z')$ as inducing a game of complete information within each household. Let $U_{ig}(d,d',z,z')$ be some individual-specific utility ranking (possibly dependent on $z,z'$) over actions $d,d'$. For $i=1,2$ we define:
\begin{equation}
D_{ig}(z,z',d') := \argmax_{d \in \mathcal{D}} \ U_{ig}(d,d',z,z').
\label{eq:defn_bestresponse}
\end{equation}

We assume that the maximizer is unique almost surely. Uniqueness of the maximizer is a generic property of $U_{ig}$ in many settings, such as in threshold-crossing models with continuously distributed shocks. For fixed values of its arguments, $D_{ig}(z,z',d')$ is a random variable because the underlying utility function $U_{ig}$ varies across the randomly sampled population. We let $D_{ig}(z,z',d')$ capture individual $i$'s ``potential'' best-response function (BRF) to their spouse's take-up decision $d'$ when playing a game with incentives $(z,z')$.\footnote{Partial identification of the joint distribution of $D(\cdot)$ appears in \citet{kline2012bounds}. Our work is related to a small literature in causal inference that defines potential take-up responses from structural payoff specifications \citep{Lazzati2015, HanBalat2023}.} Individuals with distinct $U(\cdot)$ may nonetheless have the same $D(\cdot)$. Because two households with the same set of best-response functions are behaviorally equivalent, %
we analyze $D_{ig}(z,z',d')$ as the primitive objects of our analysis.\footnote{\cite{tewolde2024game} and references therein discuss transformations of structural payoff games that preserve the best-response set. Two games are in the same equivalence class  precisely if they differ only via a best-response preserving transformation. In this sense, for the purpose of modeling heterogeneity, our latent type framework occupies an intermediate position between a full structural model, which specifies payoff functions, and a reduced-form approach.} %

\subsection{Parameters of Interest $(\theta)$}

\label{subsec:parameterofinterest}

Our goal is to obtain bounds on several commonly analyzed treatment effect parameters, which we generically denote by $\theta$. For the two-agent case, some examples of these parameters for unconditional direct and indirect effects \citep{hudgens2008toward} are
\begin{align}
    \theta&=\mathbb{E}[Y_{ig}(1,0)] - \mathbb{E}[Y_{ig}(0,0)] \qquad \text{(Average Direct Effects)} \label{eq:direct_effect} \\
    \theta&= \mathbb{E}[Y_{ig}(0,1)] - \mathbb{E}[Y_{ig}(0,0)] \qquad \text{(Average Indirect Effects).} \label{eq:indirect_effect}
\end{align}
For example, in the \cite{dupas2019effect} setting we say that men experience a direct effect if they are more likely to receive income from entrepreneurial activity when they directly open a bank account themselves than if they did not, in scenarios where their spouse's actions are held fixed, e.g. $d' = 0$. Similarly, men experience an indirect effect when their outcome changes as a result of the female spouse opening a bank account while men's treatment status is fixed, e.g. $d = 0$. The quantities in \eqref{eq:direct_effect} and \eqref{eq:indirect_effect} differ from local quantities for subpopulations that have been computed elsewhere in the literature \citep{Vazquez-Bare2020,DITRAGLIA20231589,kormos2023,ryu2024local}, in that we consider effects for the whole population. Under non-compliance, \eqref{eq:direct_effect} and \eqref{eq:indirect_effect} are not point-identified unless more restrictive assumptions are imposed, but bounds are feasible.

Our framework also provides a natural way to encode mediation-type effects \citep{huber2019review}, in an expanded setting that allows best-response behavior, involving quantities like 
\begin{equation}
\mathbb{E}[Y_{ig}(d,D_{-i,g}(z',z,d))] \qquad \text{or} \qquad \mathbb{E}[Y_{ig}(D_{ig}(z,z',d),d)]   
\label{eq:building_block_mediation_1}
\end{equation}
These objects differ from fixed-allocation effects because the partner's treatment status is generated by a best-response function. For example, evaluating $Y_{ig}(d,D_{-i,g}(z',z,d))$ requires both individual $i$'s potential-outcome schedule and the partner's best-response function.
For example, in Section \ref{sec:indirect_offer_effects}, we consider two
offer-mediated effects on men in \cite{dupas2019effect}: the indirect
female and male offer effects, respectively,
\begin{align}
\theta &= \mathbb{E}[Y_{mg}(0,D_{fg}(1,0,0))-Y_{mg}(0,D_{fg}(0,0,0))]
\label{eq:mediated_effect_1}\\
\theta &= \mathbb{E}[Y_{mg}(0,D_{fg}(0,1,0))-Y_{mg}(0,D_{fg}(0,0,0))].
\label{eq:mediated_effect_2}
\end{align}
Both estimands decompose the role of incentives on outcomes. They set male take-up to a fixed value, no bank account for men, and isolate channels operating through female account opening. The first estimand considers the indirect effect that female vouchers have on male entrepreneurial activity through incentivizing female account opening, going from $Z_f = 0$ to $Z_f = 1$.\footnote{The intent-to-treat effect $\mathbb{E}[Y_{ig} \mid Z_{ig} = 0,Z_{-i,g}=1] - \mathbb{E}[Y_{ig} \mid Z_{ig} = 0,Z_{-i,g}=0]$ captures an estimand with a similar flavor but actually captures the composite effect of $(Z_{ig},Z_{-i,g})$ on both male \textit{and} female take-up, without isolating the female take-up channel. By contrast, \eqref{eq:mediated_effect_1} isolates the effect of incentivizing female spouses alone.} The second estimand evaluates the indirect effect that male vouchers have on male entrepreneurial activity through the female spouse's best response, going from $Z_m = 0$ to $Z_m = 1$ while male take-up is fixed at zero. We call this effect second-order because it arises from a feedback channel: the male offer may affect the female spouse's take-up incentives, which may in turn affect the male spouse's outcome, even though male take-up itself is held fixed.\footnote{It is reasonable to ask whether these effects can ever be point-identified. One special case where this occurs is when men never open a bank account regardless of incentives and females perfectly comply with the relevant offer configurations. While in reality this exact type of configuration may not hold exactly, we will show that we can still construct robust bounds for any level of compliance.}  Finally, estimands involving \eqref{eq:building_block_mediation_1} can be used to answer policy-targeting questions \citep{Kitagawa2018,manski2004statistical,ananth2020optimal}. For example, we can analyze the effect of opening a bank account by default for men, while letting their female spouses naturally best-respond, by comparing the estimand in \eqref{eq:building_block_mediation_1} with $d = 1$ to the corresponding estimand with $d = 0$.\footnote{A technical complication in defining policy-targeting parameters in this setting is that researchers have to take a stand on reference values of $z$ and $z'$ even when $d$ is fixed. An alternative is to impose restrictions that discipline the alter's best-response function as a function of offers.}

\subsection{Principal Strata ($\vec{S}$) with Strategic Non-compliance}

Our notion of ``potential best responses'' is closely related to the ``complier'', ``never-taker'', ``always-taker'', and ``defier'' categorization used to analyze experiments with binary instruments and without spillovers \citep{Imbens1994}. In this classical categorization of ``principal strata'', a complier is a person that would accept the treatment if offered and not enroll if not offered. Individuals are classified into the remaining strata based on how they would react to different values of the instrument. Because each individual either receives the offer or does not, their principal strata membership encodes their counterfactual behavior under different scenarios. Table \ref{tab:illustration_individual_type} illustrates why our notion of ``potential best responses'' generalizes the classical categorization of counterfactual scenarios.

We consider an example with a couple Abel and Brenda, inspired by \cite{dupas2019effect}. Table \ref{tab:illustration_individual_type} shows the results of a hypothetical questionnaire filled-in by Abel with yes or no responses (encoded as 1 and 0). Suppose that many such individuals received this questionnaire regarding their spouse and answered truthfully, then two individuals with the same set of responses in the last column would belong to the same principal stratum and have the same ``type''.

\begin{table}[h]
\centering
\footnotesize
    \begin{tabular}{lllcc}
        \hline
        Who is offered? & $\begin{array}{l} \text{Does Brenda} \\ \text{open an account?} \end{array}$ & Question & $\begin{array}{l} \text{Abel's} \\ \text{Response}  \end{array}$ & $S_{ig} $\\
        \hline
        Nobody & No & Does Abel open an account? & $D_{ig}(0,0,0)$ & 0 \\
     Nobody & Yes & Does Abel open an account? & $D_{ig}(0,0,1)$ & 0 \\
    Only Brenda & No & Does Abel open an account? & $D_{ig}(0,1,0)$ & 0  \\
    Only Brenda & Yes & Does Abel open an account? & $D_{ig}(0,1,1)$ & 0 \\       
     Only Abel & No & Does Abel open an account? & $D_{ig}(1,0,0)$ & 1 \\
     Only Abel & Yes & Does Abel open an account? & $D_{ig}(1,0,1)$ & 1\\     
    Both & No & Does Abel open an account? & $D_{ig}(1,1,0)$ & 1\\
     Both & Yes & Does Abel open an account? & $D_{ig}(1,1,1)$ & 1 \\     
    - & No & Does Abel get income without account? & $Y_{ig}(0,0)$ & 0\\
     - & Yes & Does Abel get income without account? & $Y_{ig}(0,1)$ & 0 \\   
         - & No & Does Abel get income with account? & $Y_{ig}(1,0)$ & 1  \\
     - & Yes & Does Abel get income with account? & $Y_{ig}(1,1)$  & 1  \\
     \hline
    \end{tabular}
    \caption{Hypothetical latent-type questionnaire}
    \label{tab:illustration_individual_type}
\end{table}
In Table \ref{tab:illustration_individual_type} we follow the \cite{balke1994counterfactual} tradition of stacking take-up and outcome types into a single vector, and innovate by augmenting it to the strategic compliance setting.  For convenience we can stack all types into a vector $\vec{S}_g = (S_{1g},S_{2g},\ldots,S_{Ng}).$ If we knew $\vec{S}_g$, we would have all the information to construct the treatment effect parameters from the previous section. For example, the mediated effects in \eqref{eq:mediated_effect_1} and \eqref{eq:mediated_effect_2} can be obtained by recursively substituting different components of $\vec{S}_{g}$. Assuming that Abel and Brenda both had the same type depicted in Table \ref{tab:illustration_individual_type}, then $Y_{ig}(0,D_{-i,g}(0,1,0)) = 0$.

We argue that many restrictions on potential outcomes and take-up imposed in the
literature can be expressed as restrictions on $\vec{S}$. Conversely, allowing $\vec{S}$ to take any value in its full support yields an essentially unrestricted model.\footnote{Because everything is binary, $S_{ig}$ is supported on $\{0,1\}^{12}$. For two agents as in \cite{dupas2019effect} the unrestricted model has many types, but it is still tractable numerically. We discuss ways of reducing dimensionality without impacting the length of the bounds in later sections.} In the empirical application section, we show that we can compute reasonable bounds on treatment effect parameters even in the completely unrestricted model. We next discuss a few examples from the literature.

\textbf{Case 1 (No Interference).}  We nest the classical setup in which $D_{ig}(z,z',d')$ only depends on an agent's own offer $(z)$ and $Y_{ig}(d,d')$ only depends on their own treatment $(d)$ \citep{Imbens1994,balke1997bounds}.  For instance, an agent who has a constant outcome regardless of treatment status, $Y_{ig}(d,d') = 0$, and who complies with their offer $D_{ig}(z,z',d') = z$, is encoded in the notation of Table \ref{tab:illustration_individual_type}, as a type vector $S_{ig} = (0,0,0,0,1,1,1,1, 0,0,0,0)$. Without interference, some of the entries of $S_{ig}$ in Table \ref{tab:illustration_individual_type} are duplicated because Abel is not responsive to Brenda's offer nor her take-up decision. %
    
\textbf{Case 2 (Interference without Strategic Take-up).} We allow for types in which $Y_{ig}(d,d')$ is unrestricted, where $D_{ig}(z,z',d')$ only depends on the set of household's offers $(z,z')$ but not the action of their spouse $(d')$ \citep{DITRAGLIA20231589,ryu2024local,Vazquez-Bare2020, kormos2023}. For instance, this includes types of agents whose outcome is an indicator of both being treated, $Y_{ig}(d,d') = d \times d'$, and who only take-up if both are offered,\footnote{\citet{Vazquez-Bare2020}, among others, refers to such a type as a group-complier.} $D_{ig}(z,z',d') = z \times z'$. This is encoded, in the notation of Table \ref{tab:illustration_individual_type}, as a type vector $S_{ig} = (0,0,0,0,0,0,1,1,0,0,0,1)$.

\textbf{Case 3 (Interference with Strategic Take-up).}  We allow for cases in which spouses strategically coordinate their take-up decisions (as would be typical in games). For instance, an agent who takes an action only if both are treated $Y_{ig}(d,d') = d \times d'$, and who mimics their spouses' take-up $D_{ig}(z,z',d') = d'$ is encoded as $S_{ig} = (0,1,0,1,0,1,0,1,0,0,0,1)$.

\section{Identification, Estimation, and Inference}
\label{sec:ident_estim_infer}
   Table \ref{tab:illustration_individual_type} illustrates why $S_{ig}$ provides a complete specification of potential outcomes and best-responses. To streamline exposition, we drop the $\{ig\}$ subscript and introduce
   \begin{equation}
    Y(d,d';s) \qquad  \text{and} \qquad D(z,z',d';s),
   \end{equation}
   to denote responses for type $S_{ig} = s$.   Given a realization of types $\vec{s}$ (aggregating the types of all individuals) and group offers $\vec{z}$, we define the set of pure-strategy Nash equilibria in take-up decisions $\vec{d}$.\footnote{Our framework can be adapted to other fixed-point notions producing a finite number of equilibria.} In the two agent case, this takes the form\footnote{Generalizations to $N > 2$ are straightforward, but require additional notation tracking individual indices.}
	\begin{align}
    \text{Nash}(\vec{s},\vec{z}) := \left\{ \vec{d} \in \{0,1\}^2: \quad \begin{array}{c} D(z_1,z_2,d_2;s_1) = d_1 \\ D(z_2,z_1,d_1;s_2) = d_2 \end{array} \right\}.
    \label{eq:nash}
	\end{align} 
    A set of types $\vec{s}$ induces multiple equilibria if $|\text{Nash}(\vec{s},\vec{z})| > 1$ for a given offer vector. This feature can lead to model incompleteness: one distribution of latent types may be compatible with more than one distribution of the observables. To complete the model, we associate to each type an equilibrium selection function $E(\vec{z},\vec{d})$ which gives the probability of choosing $\vec{d} = (d_1,d_2)$ given instruments $\vec{z}$ and we require $\sum_{\vec{d}}E(\vec{z},\vec{d}) = 1$. 
    \par 
    \textbf{Case 3 (continued):} If the types $s_1$ and $s_2$ describe agents who mimic the take-up behavior of their spouse, then the only two pure-strategy Nash equilibria 
that satisfy the definition in \eqref{eq:nash} are contained in the set $\text{Nash}(\vec{s},\vec{z}) = \{ (0,0), (1,1)\}$; either both take-up or neither of them does. In practice, agents must find some deterministic or stochastic rule to choose an equilibrium. One example of such rules is
\begin{eqnarray*}
 e(\vec{z},\vec{d})=\begin{cases}
     0, &\text{ if } d_1 = 0, d_2 = 0, \\
          0, &\text{ if } d_1 = 0, d_2 = 1, \\
      0 ,& \text{ if } d_1 = 1, d_2 = 0, \\
     1,& \text{ if } d_1  = 1, d_2 = 1.
 \end{cases}
\end{eqnarray*}
In this example, agents choose the same rule for all values $\vec{z}$. More generally, we allow the equilibrium selection itself to be shifted by the instruments and to be correlated with potential outcomes, \textit{i.e.}, the equilibrium selection is game-dependent. Similar to Table \ref{tab:illustration_individual_type}, we encode the equilibrium selection rule stacking all the values for different $(\vec{z},\vec{d})$. We combine all functions above into a single, group-specific, game type $T_g = (\vec{S}_g,E_g) \in \mathcal{T}$.
We let $T$ be random, with distribution $\mu \in \Delta(\mathcal{T})$, where $\Delta(\mathcal{T})$ is the space of possible distributions over $\mathcal{T}$. We also let $t=(\vec s,e)$ denote a realization of $T$. We require instruments to be exogenous to the types.
    \begin{assumption}[Independence]\label{AsIndep} $\vec{Z}_g \independent T_g$.
	\end{assumption}    
    Assumption \ref{AsIndep} is standard in instrumental variable settings. In many experiments or quasi-experiments, including our empirical application, Assumption \ref{AsIndep} holds automatically because offers are fully randomized, and hence not systematically assigned with knowledge of the latent types. To accommodate blocked or stratified experiments, Assumption \ref{AsIndep} can be relaxed to conditional independence given covariates and/or network information.    
    
The starting collection of types $\mathcal{T}$ is very rich. We start by imposing baseline assumptions, such as consistency with Nash equilibrium play, discussed below. These assumptions suffice to compute feasible bounds on treatment effect parameters. In Section \ref{SecEncode} we also explore the impact of imposing, either individually or in combination, additional restrictions from the literature, such as dominance, symmetry, monotonicity etc.

	\begin{assumption}[Consistency with Nash]\label{AsNash} For a type $t = (\vec{s},e) \in \mathcal{T}$ we have $Nash(\vec{s},\vec{z}) \ne \varnothing$. Moreover: $
			e(\vec{z},\vec{d}) > 0  \implies \vec{d} \in Nash(\vec{s},\vec{z})$.
	\end{assumption}
    Assumption \ref{AsNash} requires the equilibrium selection rules to be supported over the set of Nash equilibria that $\vec{s} = (s_1,s_2)$ could produce. If such a set is empty, type $t$ is ruled out: Assumption \ref{AsNash} guarantees existence of a pure-strategy Nash equilibrium with probability 1. Analogous assumptions appear also in \citet{Tamer2010}, \citet{kline2012bounds} and \citet{HanBalat2023}. %
For computational reasons it is convenient to work with sets $\mathcal{T}$ of finite dimension. We impose a restriction that only includes deterministic equilibrium selection rules which, in combination with the discreteness of $s$, provides a general discrete space $\mathcal{T}$.
    \begin{assumption}[Deterministic Equilibria Selection] 
    The sets $(\mathcal{Y},\mathcal{Z},\mathcal{D})$ are finite.  If $(\vec{s},e) \in \mathcal{T}$, then $e(\vec{z},\vec{d}) \in \{0,1\}$ for all $(\vec{z},\vec{d})$.
    \label{assump:deterministic_equilibria}
    \end{assumption}
    In some cases, it may also be natural to consider stochastic equilibrium selection rules, e.g. $e(\vec{z},\vec{d}) = 0.5$. In Section \ref{sec:det_vs_stoc} we formally show that generalizing to stochastic rules has no impact on the identified set of parameters of interest, which we describe below, as long as such rules are contained in the convex hull of deterministic rules in $\mathcal{T}$. 

\subsection{Set Identification}
\label{sec:set_identification}

By analogy to Table \ref{tab:illustration_individual_type}, the group type $T_g = (\vec{S}_g,E_g)$ can be interpreted as a very long vector encoding the household's collective plan of action under different scenarios. Under Assumptions \ref{AsExclu}-\ref{assump:deterministic_equilibria}, once we fix the household type $T_g = (\vec{S}_g,E_g)$ and the offer schedule $\vec{Z}_g$, outcomes and take-up profiles can be computed directly. We can define a ``reduced-form'' rule that maps each offer vector into the corresponding selected equilibrium action profile.
\begin{equation}
    \vec{D}^{e}(\vec{z}) 
    =\sum_{\vec{d} \in \mathcal{D}^N}\vec{d} \ e(\vec{z},\vec{d}), \qquad \vec{D}_g = \vec{D}^{E_g}(\vec{Z}_g).
    \label{eq:ec_selected_actions}
\end{equation}

\textbf{Mapping Types to Observable Restrictions}. Our identification result relies on the following equivalence 
\begin{align}
    \label{eq:identityobservational}
    \begin{split}
    &(\vec{Y}_g,\vec{D}_g) = 
    (\vec{y},\vec{d}) \iff  h(T_g,\vec{Z}_g) = (\vec{y},\vec{d}), \quad h(t,\vec{z}) := (\vec{Y}(\vec{D}^{e}(\vec{z});\vec{s}),\vec{D}^{e}(\vec{z})),
    \end{split}
\end{align}
where $t = (\vec{s},e) \in \mathcal{T}$ and $\vec{Y}(\vec{d};\vec{s})= [ Y(d_{i},d_{-i};s_i) ]_{i = 1,\ldots,N}$ denotes the vector of potential outcomes across individuals under allocation $\vec{d}$ and type vector $\vec{s}$. Equation \eqref{eq:identityobservational} replaces $\vec{d}$ with the selected equilibrium take-up profile for each type. The importance of specifying an equilibrium selection rule $E_g$ as part of each type is that it resolves any ambiguity in behavior. Using the $h(\cdot)$ defined in \eqref{eq:identityobservational}, we can link observed behavior to types via 
\begin{align}
&\mathbb{P}( (\vec{Y}_g,\vec{D}_g) = (\vec{y},\vec{d}) \mid \vec{Z}_g = \vec{z}) = \mathbb{P}( h(T_g,\vec{z}) = (\vec{y},\vec{d}) \mid \vec{Z}_g = \vec{z}) \nonumber \\ %
&= \mathbb{P}( h(T_g,\vec{z}) = (\vec{y},\vec{d})) = \sum_{t \in \mathcal{T}} \mathbbm{1}\{ h(t,\vec{z}) = (\vec{y},\vec{d})\} \mathbb{P}(T_g = t).
\label{eq:map_oberved_types}
\end{align}
By Assumptions \ref{AsExclu}-\ref{AsNash}, the probability in \eqref{eq:map_oberved_types} depends on the distribution of $T_g$ alone. Because $\mathcal{D}, \mathcal{Z}, \mathcal{Y}$ and the number of agents are all finite, under Assumption \ref{assump:deterministic_equilibria}, the type space is also finite. The indicator in the last line is known once we fix the type and offer vector. The distribution of types is unknown but is constrained by the data. $\mathbb{P}(T_g = t)$ is generally not point identified, as the number of possible types typically exceeds the number of observable restrictions.

Since all the variables are finite, we can enumerate the types $t_\ell$ and observed values $w_k = (\vec{y}_k,\vec{d}_k,\vec{z}_k)$ for $\ell \in \{ 1,\ldots,|\mathcal{T}|\}$ and $k \in \{ 1,\ldots, |\mathcal{W}|\}$ (where $\mathcal{W}$ is the support of $(\vec{Y}_g,\vec{D}_g,\vec{Z}_g))$. Algorithm \ref{alg:types} describes a simple, brute-force approach to enumerate and encode the types, closely matching the numerical approach from our custom Python package. Using the enumeration from Algorithm \ref{alg:types}, we can write \eqref{eq:map_oberved_types} as $p_k = \sum_{\ell = 1}^{|\mathcal{T}|} A_{k \ell} \mu_\ell$, where $A_{k \ell} = \mathbbm{1}\{h(t_\ell,\vec{z}_k) = (\vec{y}_k,\vec{d}_k)\}$, $p_k = \mathbb{P}(\vec{Y}_g = \vec{y}_k,\vec{D}_g = \vec{d}_k \mid \vec{Z}_g = \vec{z}_k)$, and $\mu_\ell = \mathbb{P}(T_g = t_\ell)$, which can be encoded in matrix form as\footnote{As highlighted by Algorithm \ref{alg:types}, the matrix $A$ only depends on the support of the observed variables and the number of agents. In practice, we efficiently pre-compute this matrix before starting our analysis. The matrix $A$ is $|\mathcal{W}|\times |\mathcal{T}|$, $\mu$ is $|\mathcal{T}| \times 1$, and $p$ is $|\mathcal{W}|\times 1$.} 
\begin{equation}
    A \mu = p.
    \label{eq:linearsystem_matrices}
\end{equation}
Each row $k$ of $A$ is associated with a possible realization of observable behavior and offers, $(\vec{y},\vec{d},\vec{z})$. %
Each column $\ell$ is associated with an augmented type $(\vec{s},e)$. We assign a one to entry $A_{k\ell}$ if type $t_\ell$ generates the observable behavior indexed by row $k$, and zero otherwise.\footnote{In the \cite{dupas2019effect} example, an entry $A_{k\ell}=1$ for a row in which neither spouse receives entrepreneurial income but both open a bank account when offered means that type $t_\ell$ generates exactly that outcome and take-up profile under the corresponding offer vector. } 

\textbf{Mapping Types to Estimands}. We can write the parameters of interest as
\begin{align}
    \theta = \sum_{t \in \mathcal{T}} \theta(t) \mathbb{P}(T_g = t),
    \label{eq:parameter_interest}
\end{align}
where $\theta(t)$ is the type-specific contribution to the estimand. For example, following the logic of Table \ref{tab:illustration_individual_type}, we can write $\theta(t) = Y(1,0;s_1) - Y(0,0;s_1)$ for the direct effect in \eqref{eq:direct_effect} for males in the \cite{dupas2019effect} example and $\theta(t) = Y(0,1;s_1) - Y(0,0;s_1)$ for the indirect effect. The offer-mediated estimands from \eqref{eq:mediated_effect_1} and \eqref{eq:mediated_effect_2} can also be expressed in terms of types once we fix the reference values of the offers. Depending on the estimand, $\theta(t)$ may depend on only a subset of the information contained in the group type or all of it. The identified set is 
\begin{align*}
\begin{split}
 \Theta(\Delta(\mathcal{T}),p) = \bigg\{ &\theta \in \mathbb{R} : \text{ exists } \mathbb{P} \in \Delta(\mathcal{T}) \text{ such that } \mathbb{P}(T_g = t) \text{ satisfies } \eqref{eq:map_oberved_types}  \\
 &\text{ and }\theta = \sum_{t \in \mathcal{T}} \theta(t) \mathbb{P}(T_g = t) 
 \bigg\}.
 \end{split}
\end{align*}
The bounds on treatment effect parameters obtained by $\Theta(\Delta(\mathcal{T}),p)$ are sharp because we leverage all the observed restrictions on the type space. When $\Theta(\Delta(\mathcal{T}),p)$ is non-empty, it is a closed interval, and thus fully described by its endpoints. This is because $\Delta(\mathcal{T})$ is a compact and convex space, and all the restrictions and estimands are linear.

\subsection{Estimation}

Algorithm \ref{alg:estimation} provides details on how to estimate the lower and upper bounds of the identified set. The algorithm starts by computing a sample analog of the observed probability vector, which we denote by $\widehat{p}$. Next, following the conventions in  \eqref{eq:linearsystem_matrices} and \eqref{eq:parameter_interest}, $	\theta = c^T\mu,$
where $c$ is a $|\mathcal{T}| \times 1$ vector whose entries encode $\{ \theta(t) \}_{t \in \mathcal{T}}$ and $\mu$ is a non-negative vector encoding the probabilities. The value depends on the specific estimand of interest. For instance, it could take-on different values for estimating direct or indirect effects.
 Once the key objects have been defined, the identified set can be obtained by solving two linear programs. The restriction $1^T \mu = \sum_{\ell =1}^{|\mathcal{T}|} \mu_\ell = 1$ ensures that probabilities add up to one.

\SetKwComment{Comment}{/* }{ */}
\begin{algorithm}[h]

\caption{Estimation of identified set}
\label{alg:estimation}

\KwData{$\{(Y_{ig},D_{ig},Z_{ig})\}_{i=1}^N$, i.i.d.\ for $g=1,\ldots,G$; finite support
$\mathcal{Y},\mathcal{D},\mathcal{Z}$.}

Compute $(\mathcal{W},A)$ from Algorithm \ref{alg:types} \\
$\widehat{p} \gets$ Empty $|\mathcal{W}| \times 1$ vector. Encodes estimated conditional probabilities.\\

\For{$k \in \{1,\ldots,|\mathcal{W}|\}$}{
    $(\vec{y},\vec{d},\vec{z}) \gets$ Extract $k^{th}$ element from $\mathcal{W}$ \\

    $G(\vec z) \gets \sum_{g=1}^G
    \mathbbm{1}\left\{
    \{Z_{ig}\}_{i=1}^N=\vec z
    \right\}$ \\

    $\widehat{p}_k \gets \dfrac{1}{G(\vec z)} \sum_{g=1}^G
    \mathbbm{1}\left\{
    \left(
    \{Y_{ig}\}_{i=1}^N,
    \{D_{ig}\}_{i=1}^N,
    \{Z_{ig}\}_{i=1}^N
    \right)
    =
    (\vec y,\vec d,\vec z)
    \right\}$

}
\For{$\ell \in \{1,\ldots,|\mathcal{T}|\}$}{
    $t = (\vec{s},e) \gets$ Extract $\ell^{th}$ element of $\mathcal{T}$ \\
    $c_{\ell} \gets \theta(t)$
}
Solve via constrained optimization:
$\left[ \underline{\theta}(\Delta(\mathcal{T}),\widehat{p}),
\overline{\theta}(\Delta(\mathcal{T}),\widehat{p})\right] :=
\left[
\min\limits_{\left\{
\mu \in \mathbb{R}_{+}^{|\mathcal{T}|}:
\begin{array}{c}
A\mu=\widehat{p},\\
1^T\mu=1
\end{array}
\right\}}
 c^T\mu,
\quad
\max\limits_{\left\{
\mu \in \mathbb{R}_{+}^{|\mathcal{T}|}:
\begin{array}{c}
A\mu=\widehat{p},\\
1^T\mu=1
\end{array}
\right\}}
 c^T\mu
\right]
$
\end{algorithm}
\subsection{Encoding Additional Restrictions}\label{SecEncode}

Algorithm \ref{alg:estimation} estimates bounds for the identified set under the minimal restrictions imposed by Assumptions \ref{AsExclu}-\ref{assump:deterministic_equilibria}. In our empirical example we show that these minimal restrictions suffice to get meaningful bounds. However, our framework can be readily extended to accommodate extra structure that may reduce the length of the bounds. For us, a model restriction is a class of probability distributions over $T_g$, which we denote by $\mathcal{M} \subseteq \Delta(\mathcal{T})$. We focus on restrictions that can be encoded as linear constraints,
\begin{equation}
    R\mu \le r, \qquad \mu \ge 0,
    \label{eq:linear_restrictions}
\end{equation}
for a matrix $R$ and vector $r$. Below we describe three categories of commonly used restrictions that can be encoded in the form of \eqref{eq:linear_restrictions}. In Table \ref{tab:Restriction table} in the Appendix, we give examples of a broader set of restrictions previously proposed in the literature.

\textbf{Example 1 (Type Subsetting).} For example, the models in \cite{KangImbens2016} and \cite{DITRAGLIA20231589} assume that best-responses only depend on an individual's own offer, i.e. $D_{ig}(z,z',d') = D_{ig}(z)$, while those in \cite{Vazquez-Bare2020} and \cite{ryu2024local} impose $D_{ig}(z,z',d') = D_{ig}(z,z')$, with certain monotonicity restrictions across different values of $z,z'$. Both of these restrictions rule out certain types in $\mathcal{T}$. They can also be evaluated numerically for each type vector $t \in \mathcal{T}$ by comparing individual entries, using labels analogous to Table \ref{tab:illustration_individual_type}. Encoding this class of restrictions using a modified version of Algorithm \ref{alg:types} is straightforward. We iterate over all feasible types to construct a vector $R$, and construct a one indicator for whether that type fits violates the restriction and zero otherwise. With $\mu \ge 0$, the condition $R\mu \le 0$ assigns positive mass only to types that satisfy the condition. The discarded types can also be ruled out by sub-setting certain columns of $A$ from $R$ in Algorithm \ref{alg:estimation}  (effectively constructing the type mapping from types to observables only for a subset $\mathcal{T}^* \subseteq \mathcal{T}$), adjusting the dimension of $\mu$ accordingly.\footnote{Hard restrictions can be encoded with by setting $r = 0$ or by dropping columns of $A$. The latter is more computationally efficient since it reduces the dimensionality of the problem.}

\textbf{Example 2 (Type Tolerating).} For instance, \cite{deChaisemartin2017} imposes an upper bound $r$ on the proportion of defiers (individuals who enroll when not offered and decline when offered). For example, we could let $R$ denote entries equal to one for ``defiers'' and zero for ``non-defiers'' analogous to the procedure described above and impose $R \mu \le r$.
There are many important restrictions with game-theoretic content such as best-response with dominant strategies $D_{ig}(z,z',d') = D_{ig}(z,z')$ and super-modularity $D_{ig}(z,z',1) \geq D_{ig}(z,z',0)$. There are also important restrictions on equilibria, such as restricting to types with unique equilibria or with a common equilibrium selection given the available set. Another important type of restrictions is symmetry in outcomes or best-responses \citep{ciliberto2009market,berry1992estimation,mazzeo2002product}, \textit{i.e.}, equivariant symmetry $D_{ig}(z,z',d') = D_{-i,g}(z,z',d')$, which arises from identical utility functions in \eqref{eq:defn_bestresponse}.\footnote{See Appendix \ref{sec:symmetry_refutability} for this and and an alternative notion of individual symmetry.} Any of these restrictions could be encoded as hard restrictions or softer ``tolerating'' restrictions.

\textbf{Example 3 (Cross-type Constraints).}  Restrictions such as monotone treatment selection (MTS), discussed in \cite{Manski2000} and \cite{manski2009more}, impose cross-equation restrictions across types. Although their original assumption is stated without spillovers, we can extend it to the more general case. For ease of exposition we write this condition with $(mg,fg)$ (male/female in group $g$) subscripts rather than $(ig,-ig)$, and write $\mathbb{P}(Y_{mg}(a,b) = 1 \mid D_{mg} = 1) \ge \mathbb{P}(Y_{mg}(a,b) = 1 \mid D_{mg} = 0)$ for each allocation $(a,b)\in\{0,1\}^2$.

In the \cite{dupas2019effect} example, this would mean male spouses who selected into opening a bank account would be more likely to receive entrepreneurial income than those who did not, for any of the possible counterfactual allocations $(a,b)$. This restriction differs from those we discussed previously in that it combines type information with ``natural'' selection decisions made by participants. In the next lemma we show how to express this restriction in matrix form for males. An analogous restriction can be made for females.
\begin{Lemma}[Encoding Male Monotone Treatment Selection]
\label{lem:mts_equivalence}
Suppose that the types in \(\mathcal{T}\) satisfy Assumptions \ref{AsExclu}-\ref{assump:deterministic_equilibria}, and suppose that \(0<\mathbb{P}(D_{mg}=1)<1\).
Let $\mathcal{D} = \{0,1\}$ and fix a counterfactual allocation \((a,b)\in \mathcal{D}^2\). For each type
\(\ell=1,\ldots,|\mathcal{T}|\), let \(s_{m\ell}\) be the component of $\vec{s}_\ell$ for the male spouse and let $D^{e_\ell}_m(\vec{z})$ be their
take-up decision according to \(e_\ell\)
under offer vector $\vec{z}$. Then the MTS inequality
$\mathbb{P}(Y_{mg}(a,b)=1\mid D_{mg}=1)
\ge
\mathbb{P}(Y_{mg}(a,b)=1\mid D_{mg}=0)$ is equivalent to $R^{a,b}\mu\le 0$, where each entry of the row vector \(R^{a,b}\) is
{\small
\begin{equation*}
R^{a,b}_{\ell}
=
\sum_{\vec{z}} \mathbb{P}(\vec Z_g=\vec{z})
\times 
\left[
\mathbbm{1}\{Y(a,b;s_{m\ell})=1,\;D^{e_\ell}_m(\vec{z})=0\}
-
\kappa_m
\mathbbm{1}\{Y(a,b;s_{m\ell})=1,\;D^{e_\ell}_m(\vec{z})=1\}
\right],
\end{equation*}
}
with $\kappa_m
=
\frac{\mathbb{P}(D_{mg}=0)}{\mathbb{P}(D_{mg}=1)}$. If MTS is imposed for every binary allocation, the four row vectors
\(\{R^{a,b}:(a,b)\in \mathcal{D}^2\}\) are stacked into a matrix \(R^{MTS}\),
and the restriction is $R^{MTS}\mu\le 0$.
\end{Lemma}

The restriction in Lemma \ref{lem:mts_equivalence} pools across offer vectors: because $D_{mg}=D^{e}_{m}(\vec Z_g)$, the populations conditioned on in the MTS inequality mix offer cells with weights that depend on the experimental design $\mathbb P(\vec Z_g=\vec z)$, so its content is design-dependent. Under one-sided noncompliance, for instance, it compares men who open an account with all other men, including those who were not offered a voucher. A design-invariant alternative imposes MTS within each offer cell,
$\mathbb{P}(Y_{mg}(a,b)=1\mid D_{mg}=1,\vec Z_g=\vec z)\ \ge\ \mathbb{P}(Y_{mg}(a,b)=1\mid D_{mg}=0,\vec Z_g=\vec z)$
for every $\vec z$ with $0<\mathbb P(D_{mg}=1\mid\vec Z_g=\vec z)<1$. Since $\vec Z_g$ is independent of $T_g$, this is a restriction on $\mu$ alone, and it is encoded similar to Lemma \ref{lem:mts_equivalence} with one row per cell, rather than pooling across offers. We follow \cite{manski2009more} in stating MTS for the population as a whole, and use the pooled version in Section \ref{sec:mts}.

Unlike the first two categories we considered, here the vector $R$ depends on unknown probabilities which need to be estimated in practice via $\widehat{R}$.  For any of the restrictions discussed above, we can replace the criteria for obtaining bounds in Algorithm \ref{alg:estimation} with
{\small
\begin{align*}
\underline\theta(\mathcal M,\widehat p)
&=\min_{\mu\in\mathcal F(\mathcal M,\widehat p)}c^T\mu,
\qquad
\overline\theta(\mathcal M,\widehat p)
=\max_{\mu\in\mathcal F(\mathcal M,\widehat p)}c^T\mu,\\
\mathcal F(\mathcal M,\widehat p)
&=\left\{\mu\in\mathbb R_+^{|\mathcal T|}:
A\mu=\widehat p,\; \mathbf 1^T\mu=1,\; \widehat{R}\mu\le r\right\}.
\end{align*}}

\subsection{Inference}

As noted by \cite{Hsieh2022Inference}, the value of a linear program is, in general, not differentiable as a function of its inputs. Therefore, sampling variation in the estimated ingredients of the optimization problem need not imply asymptotic normality of the estimated endpoints of the identified set. In particular, Theorem 3.1 in \cite{FangSantos2018} shows that, for a Hadamard directionally differentiable map, the standard bootstrap is consistent if and only if the map is fully differentiable. Since the value of a linear program is generally only directionally differentiable in its inputs, the standard bootstrap may fail. We follow \cite{goff2025inferencevaluelinearprogram}, who develop uniformly valid bootstrap tests of the solvability of systems of linear inequalities with known or estimated coefficients (see details in Appendix \ref{AppInference}).

\section{Is Strong Strategic Neutrality (SSN) irrelevant?}
\label{sec:uselevel_restrictions}

	In the previous section we obtained a basic linear programming representation of the identified set $\Theta(\mathcal{M},p)$. We now introduce two notions that help applied researchers discuss whether additional assumptions on agents' behavior can be useful for identification. 
	
	\begin{definition}
		\label{def:irrelevance}
		Consider a set of restrictions $\mathcal{M}^* \subseteq \mathcal{M} \subseteq \Delta(\mathcal{T})$. We say that (i) $\mathcal{M}^*$ is p-non-refutable for $\theta$ if $\Theta(\mathcal{M}^*,p) \neq \varnothing$, (ii) $\mathcal{M}^*$ is $(\mathcal{M},p)$-irrelevant for parameter $\theta$ if $\Theta(\mathcal{M}^*,p) = \Theta(\mathcal{M},p)$.
	\end{definition}
	
	According to Definition \ref{def:irrelevance} $(i)$, a restriction $\mathcal{M}^*$ is $p$-non-refutable if the observable distribution $p$ can be rationalized by a latent distribution that satisfies that restriction. Definition \ref{def:irrelevance} $(ii)$ instead says that, starting from $\mathcal{M}$, moving to $\mathcal{M}^*$ does not tighten the identified set for $\theta$. We explore simple properties of these definitions that will prove useful in later analyses.
	
	\begin{theorem}[Nested Restrictions]
		\label{prop:chain}	
			(i) Let $\mathcal{M}$ be p-non-refutable. If $\mathcal{M^*}$ is $(\mathcal{M},p)$-irrelevant then it is also p-non-refutable.
			(ii) Let $\mathcal{M^*}$ be $(\mathcal{M},p)$-irrelevant for $\theta$. For any $\mathcal{M}'$ such that $\mathcal{M}^* \subseteq \mathcal{M'} \subseteq \mathcal{M}$, then $\mathcal{M'}$ is also $(\mathcal{M},p)$-irrelevant. 
			(iii) Let $\mathcal{M'} \subseteq \mathcal{M}$. If $\mathcal{M'}$ is p-non-refutable, $\mathcal{M}$ is p-non-refutable.
	\end{theorem}

    These definitions will allow us to apply ``transitive properties'' to different types of restrictions. One such restriction, common in the literature, assumes that agents do not respond to others' take-up \citep{DITRAGLIA20231589,ImaiJiangMalani2021,KangImbens2016,ryu2024local}.\footnote{It is natural to expect dominant strategies in games with additively separable payoffs, e.g., $U_i=D_i(1+\varepsilon_i)+\sum_{j\ne i}D_j$, where $\varepsilon_i$ is a random utility shock.} In our framework, these assumptions can be encoded as restrictions over $\mathcal{T}$.
	\begin{definition}[Strong Strategic Neutrality]
		A type $(\vec{s},e) \in \mathcal{T}$ is comprised of agents with dominant best-responses if $D(z,z',d';s_i) $ does not depend on $d'$. A subset of $\mathcal{T}$ is strongly strategically neutral if all the types are comprised of agents with dominant best responses.
	\end{definition}
	\begin{definition}[Measure over Dominant Types]
		Let $\mathcal{M}$ be any set of distributions over types. We define $\mathcal{M}_{dominant} := \Delta(\mathcal{T}_{dominant}) \cap \mathcal{M}$ to be the subset of distributions where all types have dominant best-responses.
	\end{definition}
	
Does strong strategic neutrality (SSN) tighten the identified set? We show that if we impose a minimal ``closure'' condition on the set $\mathcal{T}$ and focus on fixed-allocation estimands, then SSN is irrelevant in the sense defined above. To understand what the closure condition means, let $D^{e}_i(\vec z)$ be the ``reduced-form'' take-up of each agent as a function of offers, defined as the $i^{th}$ component of \eqref{eq:ec_selected_actions}. Several different types can produce the same choices encoded in $D^{e}_i(\vec z)$; this function can be rationalized by the type $(\vec{s},e)$ but also by other types $(\vec{s}^{\,\varphi},e)$. For example, the reduced form can be rationalized by a set of agents with dominant best-responses who only make take-up choices based on offers and have $D(z_i,z_{-i},d_{-i};s_i^{\,\varphi}) = D_i^e(\vec{z})$. If, in addition, the types $\vec{s}$ and $\vec{s}^{\,\varphi}$ produce the same potential outcomes, then $(\vec{s}^{\,\varphi},e)$ is observationally equivalent to $(\vec{s},e)$ and we call it its ``neutral counterpart''.  We formalize this requirement through a notion of closure under dominant strategies.

\begin{assumption}[Closure under Dominant Best-responses]
\label{assump:closure_dominant_multiagent} 
For $t=(\vec s,e)\in\mathcal T$, define its neutral counterpart
$\varphi(t):=(\vec{s}^{\,\varphi},e)$ by $D(z_i,\vec z_{-i},\vec d_{-i};s_i^{\,\varphi}) := D^{e}_i(\vec z)$ and 
$Y(\vec d;s_i^{\,\varphi}) := Y(\vec d;s_i)$ for each $i$, all $(\vec z,\vec d_{-i})$, and all $\vec d$. A subset $\mathcal T^*\subseteq\mathcal T$ is closed under dominant best responses
if $\varphi(\mathcal T^*)\subseteq\mathcal T^*$.
\end{assumption}

	The superset of types $\mathcal{T}$ and the empty set both trivially satisfy the closure condition in Assumption \ref{assump:closure_dominant_multiagent}. This could only be violated with further restrictions on the set of types. Assumption \ref{assump:closure_dominant_multiagent} holds when the type space is sufficiently rich to include agents with dominant strategies that can mimic the same equilibrium as strategic types, and have the same potential outcomes. Considering sets of types $\mathcal{T}$ that are closed under dominant strategies makes sense if the researcher wishes to be agnostic about non-strategic behavior.

	\begin{theorem}[Irrelevance of Strong Strategic Neutrality]
		\label{thm:equivalence_dominance_multiagent_distr}
		Consider estimands taking the form $\theta = \mathbb{E}[\theta(T_g)]$, where $\theta(t)$ only depends on $t$ through the potential outcomes profile $\{Y(\cdot;s_i)\}_{i=1}^N$. Suppose that Assumptions \ref{AsExclu}--\ref{assump:deterministic_equilibria} hold and let $\mathcal{T}^* \subseteq \mathcal{T}$ and $\mathcal{M} = \Delta(\mathcal{T}^*)$. If the set $\mathcal{T}^*$ is closed under dominant best responses (Assumption \ref{assump:closure_dominant_multiagent}), then for all $p$,
		$\Theta(\mathcal{M},p) = \Theta(\mathcal{M}_{dominant},p)$.
	\end{theorem} 
    
    Theorem \ref{thm:equivalence_dominance_multiagent_distr} applies to the average direct and indirect effects introduced in \eqref{eq:direct_effect} and \eqref{eq:indirect_effect}, and more generally to any fixed-allocation parameter that only depends on types through the potential outcomes. Under mild conditions, we show that restricting attention to types with dominant strategies does not tighten the identified set. However, the theorem does not apply to the offer-mediated effects in \eqref{eq:mediated_effect_1} and \eqref{eq:mediated_effect_2} where the take-up response enters $\theta(t)$ directly, particularly for off-equilibrium responses.
	
	The interpretation of Theorem \ref{thm:equivalence_dominance_multiagent_distr} is subtle. It does not imply that strategic interactions are unimportant in practice, but rather that—for fixed allocation parameters—a partial identification approach based solely on dominant types is already sufficiently general. A key practical implication is that average treatment effects can be estimated using a more computationally efficient linear program defined over a lower-dimensional type space $\mathcal{T}_{dominant} \subseteq \mathcal{T}$ while still maintaining robustness to strategic behavior. 

    \textbf{Extensions:} In the Appendix, we show that other restrictions are also irrelevant for fixed-allocation parameters. In Section \ref{subsec:ssn_equilibriuminvariance}, we extend Theorem \ref{thm:equivalence_dominance_multiagent_distr} to prove irrelevance of SSN under restrictions of the form $R\mu \le r$ satisfying an additional equilibrium invariance condition, which includes monotone treatment selection with two agents. In Section \ref{sec:irrelevance_sequentialplay}, we show that a broad set of restrictions on sequential play and network best-responses are irrelevant. In Section \ref{sec:supermodularity} we show that certain notions of sub/supermodularity are also irrelevant, building on the insights from Theorem \ref{thm:equivalence_dominance_multiagent_distr}. In Section \ref{sec:symmetry_refutability}, we evaluate two types of symmetry, individual and equivariant, and show that individual symmetry is also irrelevant.

\section{What does the Wald Statistic identify?}
\label{sec:waldstatistic}

Our framework allows us to characterize the quantities identified by traditional instrumental variable methods in terms of our general latent types. Following the convention in \eqref{eq:ec_selected_actions}, once we know how a household resolves equilibrium multiplicity, we can determine its ``reduced form'' take-up choices as a function of offers, which in a two-agent setting takes the form $\big(D_{1g}^{E_g}(z_1,z_2),D_{2g}^{E_g}(z_1,z_2)\big)
    =\sum_{d_1,d_2}(d_1,d_2)E_g(z_1,z_2,d_1,d_2)$. For a given partner offer $z_2$, we say that individual 1 is an \emph{equilibrium complier} if they take up if and only if they are offered $C_{1g}(z_2)= \left\{ D_{1g}^{E_g}(0,z_2) = 0,\quad D_{1g}^{E_g}(1,z_2) = 1 \right\}$.
In settings with spillovers and strategic take-up, the interpretation of an equilibrium complier is less straightforward because it includes at least two types of individuals: (i) those who find it beneficial to enroll after receiving an offer (in terms of higher utility) and (ii) those who enroll as a result of an equilibrium response to their spouse's take-up, even if their payoffs are lower than in the counterfactual equilibrium without the offer. Thus, with strategic compliance, compliers need not be individuals who perceive themselves as better off after receiving a treatment offer, as in \cite{Heckman2005}. 

\begin{Lemma}[Interpretation of Wald Statistic]
\label{lem:ec_wald_decomposition}
Suppose that Assumptions \ref{AsExclu}-\ref{assump:deterministic_equilibria} hold and that $N_g = 2$. Fix $z_2$ such that $\mathbb{P}((Z_{1g},Z_{2g})=(z_1,z_2))>0$ for $z_1\in\{0,1\}$. Suppose potential outcomes are integrable, $D_{1g}^{E_g}(1,z_2)\geq D_{1g}^{E_g}(0,z_2)$ almost surely, and $\mathbb{P}(C_{1g}(z_2))>0$. Then
{\small
\begin{align}
&\frac{\mathbb E[Y_{1g}\mid Z_{1g}=1,Z_{2g}=z_2]
       -\mathbb E[Y_{1g}\mid Z_{1g}=0,Z_{2g}=z_2]}
      {\mathbb E[D_{1g}\mid Z_{1g}=1,Z_{2g}=z_2]
       -\mathbb E[D_{1g}\mid Z_{1g}=0,Z_{2g}=z_2]}
\nonumber\\
={}&\mathbb E\big[
Y_{1g}(1,D_{2g}^{E_g}(1,z_2))
-Y_{1g}(0,D_{2g}^{E_g}(0,z_2))
\mid C_{1g}(z_2)\big]
\nonumber\\
&+\frac{\mathbb P(C_{1g}(z_2)^c)}
        {\mathbb P(C_{1g}(z_2))}
\mathbb E\big[
Y_{1g}(D_{1g}^{E_g}(0,z_2),D_{2g}^{E_g}(1,z_2))
-Y_{1g}(D_{1g}^{E_g}(0,z_2),D_{2g}^{E_g}(0,z_2))
\mid C_{1g}(z_2)^c\big]
\label{eq:ec_wald_decomposition}
\end{align}
}
where $C_{1g}(z_2)^c$ denotes the noncomplier event.
\end{Lemma}
Lemma \ref{lem:ec_wald_decomposition} shows that the Wald statistic can be decomposed into a weighted sum of offer-mediated effects for both equilibrium compliers and non-compliers. It does not identify a well-defined quantity without further restrictions. For example, let $20\%$ of groups have dominant strategies $D_{1g}(z_1,z_2,d_2)=z_1$, $D_{2g}(z_2,z_1,d_1)=0$, and $Y_{1g}(d_1,d_2)=0$. Let the remaining $80\%$ have $D_{1g}(z_1,z_2,d_2)=0$, $D_{2g}(z_2,z_1,d_1)=z_1$, and $Y_{1g}(d_1,d_2)=d_2$. Increasing the offer to individual $1$ gives an own first stage of $0.2$ and an outcome effect of $0.8$. Wald equals four although all direct effects on $Y_{1g}$ are zero. Alternatively, if each agent responds only to their own offer, Wald identifies a direct effect among compliers at the partner's unchanged treatment. These conclusions maintain equilibrium monotonicity, which does not follow merely from monotone best responses holding partner actions fixed.\footnote{Assuming monotonicity, one-sided compliance, individualized offer response, and strong strategic neutrality does not always facilitate the interpretation of effects in instrumental variables models. \cite{DITRAGLIA20231589} decompose the two stage least squares estimands under these conditions using individual offers and group saturations as instrumental variables, and find that coefficients identify direct and indirect effects for distinct populations. Strategic interactions and first-stage interference exacerbate these issues.}

The first term of \eqref{eq:ec_wald_decomposition} combines own and partner treatment changes, so even in isolation it would not generally identify a direct effect. The second quantity vanishes under any of the following conditions: (i) equilibrium non-compliers do not experience spillover effects, i.e., $Y_{1g}(d_1,d_2)=Y_{1g}(d_1,d_2')$; (ii) the spouse's equilibrium take-up does not respond to individual 1's offer, $D_{2g}^{E_g}(0,z_2)=D_{2g}^{E_g}(1,z_2)$; or (iii) there is full equilibrium compliance, $\mathbb{P}(C_{1g}(z_2))=1$. It may also vanish when potential-outcome and take-up responses offset each other.

\newcolumntype{C}[1]{>{\centering\arraybackslash}p{#1}}

\section{Empirical Illustration}\label{Sempirical}

\cite{dupas2019effect} evaluate the impact of expanding access to bank accounts on various outcomes in rural Kenya, a region with relatively low access to formal banking, especially among women. We focus primarily on how account ownership influences the likelihood of entrepreneurial activity but we consider other outcomes below. In their experiment, individuals without a bank account are randomized to receive a nontransferable voucher waiving most savings-account fees. Because randomization is at the individual level, within dual-headed heterosexual households there are four treatment assignments: male-only, female-only, both, or neither. Each individual spouse can then opt to open a bank account in their name.\footnote{Spouses who were offered the voucher could also opt to enroll for a joint bank account: in the data, less than 5\% of households choose this option. %
}

In our setting, the researcher observes only each spouse's choices, not any private agreements. Such a setting with partially observable individual decisions may align more naturally with a non-cooperative framework than with the workhorse collective bargaining model for household finance and consumption \citep{lundberg1993separate, lundberg1994noncooperative}. Here, a non-cooperative model is without loss of generality: canonical cooperative outcomes can be implemented as equilibria of a suitably constructed game. In particular, a unitary household can be represented via a potential game where the maximizers of the common welfare are guaranteed to be pure-strategy Nash equilibria \citep{monderer1996potential}.\footnote{Through this embedding, unitary households are identified with latent types whose joint BRFs are derived from a potential game. Thus, being a ``unitary household'' can be cast as a restriction on the latent space of types. When $N=2$ and $\mathcal{D} = \{0,1\}$, this restriction coincides with imposing the existence of Pure Strategy Nash Equilibrium (a.s.). In Appendix \ref{sec:unitary_households} we show that when $\mathcal{D}$ includes at least three actions, one can show that the PSNE assumption is strictly weaker than being a ``unitary household''.} By contrast, non-cooperative equilibria may be Pareto-dominated, placing them outside of the scope of many cooperative/collective models that, by construction, assume costless communication and Pareto efficiency. Empirically, communication frictions \citep{delavande2025information} and non-cooperative elements in household decisions \citep{ashraf2009spousal} can be first order.
Whether these features are plausible in a given institutional context is an empirical question; therefore, we adopt a model that does not impose them \textit{a priori}.

If the researcher is interested in identifying particular average or policy targeting effects, the experiment's setup presents some complications due to non-compliance. This may arise from both heterogeneity in net benefits from opening a bank account, or strategic interaction between the spouses. \citet{dupas2019effect} explore related issues in their analysis of first stage estimates (Section 3). They report a test which they interpret as rejecting the unitary household model. As we detail in Appendix \ref{sec:unitary_households}, this rejection may also stem from asymmetric responses to who receives the voucher rather than failure of the unitary model itself.

Exogenous variation in the partner's cost of opening bank accounts can affect the individual choice of savings. The voucher instrument could be acting as a monetary incentive for the recipient and as an information treatment for the recipient's spouse. Finally, social and gender norms regarding household finances may further complicate the household's choices. In this context, hoping to point-identify average effects seems to rely on significant restrictions on the households' behavior. As a result, we consider a partial identification alternative with less restrictive assumptions.

We let $Z_{mg}$, $Z_{fg}$, $D_{mg}$, $D_{fg}$ be the indicators for, respectively, the male and female spouse receiving the voucher and choosing to open a bank account. $Y_{mg},Y_{fg}$ range over a small collection of outcomes, including indicators for each spouse owning a business and indicators for the household sending or receiving financial transfers from friends or next-of-kin. The empirical analysis has three steps. First, we report conventional 2SLS estimates, which summarize the main patterns in the data and illustrate why standard interpretations are delicate in the presence of strategic take-up and multiple equilibria. Second, we use the joint distribution of treatment take-up and outcomes across voucher assignments to construct bounds on average direct and indirect effects for entrepreneurial activity, as defined in \eqref{eq:direct_effect} and \eqref{eq:indirect_effect}. Third, we analyze a variety of indirect offer effects, as defined in \eqref{eq:mediated_effect_1} and \eqref{eq:mediated_effect_2}.\footnote{In Appendix \ref{appendix:female_headed}, we analyze single, female-headed households, where there are no spillovers.}

\subsection{2SLS and Wald Analysis}

   Tables \ref{tab:2sls_regression_results}-\ref{tab:2sls_by_round} show two 2SLS regressions: $Y_{ig} = \beta_{0} + \beta_{1} D_{ig} + \beta_{2} D_{-ig} + \beta_{3} D_{ig}\times D_{-ig} + U_{ig},$ where $i \in \{m,f\}$ and the instruments are $Z_{ig}$, $Z_{-i,g}$, and $Z_{ig}\times Z_{-ig}$.
We run this for men and women separately for the own enterprise outcome. The combination of latent potential outcomes, best responses and equilibrium selections results in a large number of latent types, which makes the interpretation of these regressions complicated. However, Table \ref{tab:2sls_regression_results} shows a positive association of male bank-account ownership with own-enterprise income and a negative interaction between spouses' account ownership.
Consider the own-offer Wald ratio,
holding the partner's offer fixed introduced in Lemma \ref{lem:ec_wald_decomposition}. For the male spouse, this is
\begin{align}
	W_{m}(z_{f}) \ = \
	\frac{\mathbb{E}[Y_{mg} \mid (Z_{mg},Z_{fg}) = (1,z_{f})]
	     - \mathbb{E}[Y_{mg} \mid (Z_{mg},Z_{fg}) = (0,z_{f})]}
	     {\mathbb{E}[D_{mg} \mid (Z_{mg},Z_{fg}) = (1,z_{f})]
	     - \mathbb{E}[D_{mg} \mid (Z_{mg},Z_{fg}) = (0,z_{f})]},
	\label{eq:wald_conditional}
\end{align}
for $z_{f} \in \{ 0, 1\}$. $W_{f}(z_{m})$ is defined analogously. Both are ratios of two intention-to-treat contrasts computed in the subsample with a given partner offer, using only the four assignment cells of Table \ref{tab:observable_distribution}. Table
\ref{tab:wald_conditional} reports the sample analogues using round 6 data.

\begin{table}[h]
	\centering
    \scriptsize
	\begin{tabular}{lcccccc}
    \hline
		& (1) & (2) & (3) & (4) & (5) & (6) \\
        &  Own & Own &  Receive & Receive & Send & Send  \\
        &  Enterprise & Enterprise & Children/Sibling & Friends/Parents &  Children/Sibling & Friends/Parents \\
		&  (Men) &  (Women) &  (All) & (All) & (All) & (All) \\  \hline
		&  &  &  &  &  &  \\
		$D_{ig}$ & 0.162** & -0.012 & -0.170** & 0.048 & 0.009 & 0.108 \\
		& (0.071) & (0.075) & (0.066) & (0.069) & (0.056) & (0.068) \\
		$D_{-ig}$ & 0.081 & -0.001 & -0.095 & 0.003 & 0.027 & 0.079 \\
		& (0.071) & (0.073) & (0.067) & (0.072) & (0.057) & (0.069) \\
		$D_{mg} \times D_{fg}$ & -0.207* & -0.040 & 0.179 & -0.007 & 0.001 & -0.106 \\
		& (0.124) & (0.125) & (0.113) & (0.116) & (0.098) & (0.119) \\
		Constant & 0.207*** & 0.290*** & 0.710*** & 0.379*** & 0.310*** & 0.400*** \\
		& (0.032) & (0.037) & (0.033) & (0.036) & (0.026) & (0.033) \\
		&  &  &  &  &  &  \\
		Observations & 1,679 & 1,679 & 1,679 & 1,679 & 1,679 & 1,679 \\
        \hline
		\multicolumn{7}{c}{ } \\
		\multicolumn{7}{c}{ *** p$<$0.01, ** p$<$0.05, * p$<$0.1} \\
	\end{tabular}
	\caption{ \footnotesize 2SLS estimates for \cite{dupas2019effect}. The first two columns display the coefficients for the binary outcomes $Y_{mg}$ and $Y_{fg}$ respectively, for whether an individual received income from their own enterprise. Columns 3-6 display the coefficients for binary outcomes related to the presence (or absence) of transfers with either (i) children and siblings, and (ii) friends, neighbors, and parents. For columns (1) and (2), $D_{ig}$ represents $D_{fg}$ or $D_{mg}$ ($i$'s own account opening) while $D_{-i,g}$ their partner's; for columns (3) to (6), $D_{ig}=D_{mg}$ and $D_{-i,g}=D_{fg}$. Clustered standard errors at the household level in parentheses. Sample: dual-headed households in the \citet{dupas2019effect} analysis sample, survey rounds 3--6, with one observation per household-round. The outcomes in columns (3)--(6) are household-level transfer indicators that take the same value for both spouses. }
	\label{tab:2sls_regression_results}
\end{table}

\normalsize
\begin{table}[h]
\scriptsize
\centering
\begin{tabular}{lcccccccc} \hline
 & (1) & (2) & (3) & (4) & (5) & (6) & (7) & (8) \\
 & Men & Women & Men & Women & Men & Women & Men & Women \\
 & Round 3 & Round 3 & Round 4 & Round 4 & Round 5 & Round 5 & Round 6 & Round 6 \\ \hline
 & & & & & & & & \\

$D_{ig}$
 & 0.0955 & -0.0938 & 0.1503 & -0.0134 & 0.1990* & -0.0033 & 0.2071** & 0.0647 \\
 & (0.0980) & (0.1074) & (0.0984) & (0.1101) & (0.1027) & (0.1030) & (0.0959) & (0.1075) \\

$D_{-i,g}$
 & 0.1181 & -0.0671 & 0.0409 & 0.0789 & 0.0742 & -0.0476 & 0.0829 & 0.0253 \\
 & (0.1025) & (0.1023) & (0.1005) & (0.1060) & (0.0948) & (0.1043) & (0.0963) & (0.1015) \\

$D_{mg} \times D_{fg}$
 & -0.1836 & -0.0023 & -0.2199 & -0.0656 & -0.2185 & -0.0581 & -0.2080 & -0.0241 \\
 & (0.1708) & (0.1721) & (0.1686) & (0.1812) & (0.1725) & (0.1728) & (0.1663) & (0.1760) \\

Constant
 & 0.233*** & 0.315*** & 0.211*** & 0.296*** & 0.203*** & 0.284*** & 0.181*** & 0.264*** \\
 & (0.050) & (0.055) & (0.049) & (0.054) & (0.047) & (0.053) & (0.046) & (0.052) \\

 & & & & & & & & \\
Observations
 & 437 & 437 & 417 & 417 & 414 & 414 & 411 & 411 \\
\hline
\multicolumn{9}{c}{ } \\
\multicolumn{9}{c}{ *** p$<$0.01, ** p$<$0.05, * p$<$0.1 }
\end{tabular}
\caption{\footnotesize 2SLS estimates for whether an individual received income from their own enterprise (Own Enterprise) by round. Heteroskedastic-robust standard errors in parentheses. The dependent variable is an indicator for the individual earning income from a business they own. Sample: dual-headed households, one observation per household-round; endogenous regressors and instruments as in Table \ref{tab:2sls_regression_results}. Each column is a separate regression run on the round indicated, with 437, 417, 414 and 411 households in rounds 3--6.}
\label{tab:2sls_by_round}
\end{table}
\normalsize

\begin{table}[H]
\footnotesize
    \centering
    \begin{tabular}{lccccccc}
        \hline
        Statistic & $n(1)$ & $n(0)$ & Numerator & Denominator & Wald & S.E. & 95\% CI \\
        \hline
        $W_{m}(z_{f}=0)$ & 96  & 72 & 0.142 & 0.688 & 0.207 & 0.096 & $[\ \ 0.019,\ 0.395]$ \\
        $W_{m}(z_{f}=1)$ & 145 & 98 & 0.034 & 0.697 & 0.049 & 0.082 & $[-0.111,\ 0.209]$ \\
        $W_{f}(z_{m}=0)$ & 98  & 72 & 0.042 & 0.653 & 0.065 & 0.108 & $[-0.146,\ 0.276]$ \\
        $W_{f}(z_{m}=1)$ & 145 & 96 & 0.029 & 0.641 & 0.045 & 0.094 & $[-0.138,\ 0.229]$ \\
        \hline
    \end{tabular}
    \caption{\footnotesize Conditional Wald ratios \eqref{eq:wald_conditional} for
    round 6 ($G=411$), with the own-enterprise outcome and the sample of Table
    \ref{tab:bounds_ade_ase}. $n(1)$ and $n(0)$ are the numbers of households
    with and without the own offer, at the partner offer indicated. The
    numerator is the outcome comparison and the denominator the take-up comparison.}
    \label{tab:wald_conditional}
\end{table}

Three features are worth noting. First, every denominator equals the take-up
rate among those offered, because no spouse who is not offered a voucher opens
an account: the one-sided noncompliance of Table
\ref{tab:observable_distribution} makes $\mathbb{E}[D_{mg} \mid Z_{mg}=0,z_{f}]=0$
in both cells. Second, $W_{m}(z_{f}=0) = 0.207$ coincides with the round-6 2SLS coefficient on $D_{mg}$ for men in Table \ref{tab:2sls_by_round}. This equality follows from one-sided noncompliance: no spouse opens an account without a voucher, so the $Z=(0,0)$ cell pins down the constant and, in the $Z=(1,0)$ cell, $D_{fg}$ and $D_{mg}D_{fg}$ are identically zero, leaving only the coefficient on $D_{mg}$.

Third, the point estimates differ with the partner's offer: $0.207$ when the wife is not offered a voucher against $0.049$ when she is, with nearly identical denominators. Absent any interference, the two ratios would coincide in the population, since randomization makes the male complier populations identical across the two subsamples. A gap can arise from spillovers in outcomes, but also from the wife's offer changing which men comply. In our sample, however, the difference of $0.158$ has a standard error of about $0.126$ ($p \approx 0.21$), so the data do not distinguish these channels from sampling noise. What the comparison does illustrate is that neither ratio need equal $\text{ADE}_{0,m}$: $W_{m}(z_{f}=1) = 0.049$ lies below the estimated identified set $[0.07,0.38]$ of Table \ref{tab:bounds_ade_ase}, although its confidence interval overlaps it.

\subsection{Average Direct/Indirect Effects Analysis} \label{sec:fixed_allocation}
Table \ref{tab:observable_distribution} reports the joint distribution of outcome and take-up, conditional on voucher assignment, denoted by $\hat{p}$ in Algorithm \ref{alg:estimation}. Note that, while the coefficients of the regressions reported above can be recovered from the marginal distributions in Table \ref{tab:observable_distribution}, those regressions alone do not recover the joint potential values of $Y_{mg},Y_{fg}$. The table also shows how patterns of noncompliance relate to the outcome distribution which is a key ingredient for bounding the direct and indirect effects, as it captures the interaction between take-up and outcomes.
\begin{table}[htbp]
\centering
\scriptsize
\begin{tabular}{cc}
\textbf{$Z_g$ = (0,0)} & \textbf{$Z_g$ = (0,1)} \\

\begin{tabular}{c|C{0.8cm}C{0.8cm}C{0.8cm}C{0.8cm}|C{0.8cm}}
\toprule
$D_g \backslash Y_g$ & (0,0)    & (0,1)    & (1,0)    & (1,1)    &  \\
\midrule
(0,0)               & 0.611 & 0.208 & 0.125 & 0.056 & 1 \\
(0,1)               & 0 & 0 & 0 & 0 & 0 \\
(1,0)               & 0 & 0 & 0 & 0 & 0 \\
(1,1)               & 0 & 0 & 0 & 0 & 0 \\
\midrule
                 & 0.611 & 0.208 & 0.125 & 0.056 & 1 \\
\bottomrule
\end{tabular}
&
\begin{tabular}{c|C{0.8cm}C{0.8cm}C{0.8cm}C{0.8cm}|C{0.8cm}}
\toprule
$D_g\backslash Y_g$ & (0,0)    & (0,1)    & (1,0)    & (1,1)    &  \\
\midrule
(0,0)               & 0.184 & 0.071 & 0.051 & 0.041 & 0.347 \\
(0,1)               & 0.357 & 0.153 & 0.102 & 0.041 & 0.653 \\
(1,0)               & 0 & 0 & 0 & 0 & 0 \\
(1,1)              & 0 & 0 & 0 & 0 & 0 \\
\midrule
                 & 0.541 & 0.224 & 0.153 & 0.082 & 1 \\
\bottomrule
\end{tabular}
\\[2ex]

\textbf{$Z_g$ = (1,0)} & \textbf{$Z_g$ = (1,1)} \\

\begin{tabular}{c|C{0.8cm}C{0.8cm}C{0.8cm}C{0.8cm}|C{0.8cm}}
\toprule
$D_g \backslash Y_g $ & (0,0)    & (0,1)    & (1,0)    & (1,1)    &  \\
\midrule
(0,0)               & 0.135 & 0.104 & 0.062 & 0.010 & 0.312 \\
(0,1)              & 0 & 0 & 0 & 0 & 0 \\
(1,0)               & 0.333 & 0.104 & 0.188 & 0.062 & 0.688 \\
(1,1)               & 0 & 0 & 0 & 0 & 0 \\
\midrule
                 & 0.469 & 0.208 & 0.250 & 0.073 & 1 \\
\bottomrule
\end{tabular}
&
\begin{tabular}{c|C{0.8cm}C{0.8cm}C{0.8cm}C{0.8cm}|C{0.8cm}}
\toprule
$D_g \backslash Y_g$ & (0,0)    & (0,1)    & (1,0)    & (1,1)    &  \\
\midrule
(0,0)               & 0.097 & 0.034 & 0.041 & 0.014 & 0.186 \\
(0,1)               & 0.055 & 0.028 & 0.028 & 0.007 & 0.117 \\
(1,0)               & 0.083 & 0.034 & 0.034 & 0.021 & 0.172 \\
(1,1)              & 0.290 & 0.110 & 0.062 & 0.062 & 0.524 \\
\midrule
                 & 0.524 & 0.207 & 0.166 & 0.103 & 1 \\
\bottomrule
\end{tabular}

\end{tabular}
\caption{\footnotesize Joint distribution of $D_g=(D_{mg},D_{fg})$ and $Y_g=(Y_{mg},Y_{fg})$ conditional on each instrument value $Z_g=(Z_{mg},Z_{fg})$. Round 6 ($n=411$). Sample: dual-headed households observed in round 6, one observation per household. $Y_{mg}$ and $Y_{fg}$ are indicators for the male and the female spouse earning income from a business they own; $D_{mg}$ and $D_{fg}$ indicate opening a bank account; $Z_{mg}$ and $Z_{fg}$ indicate receiving a fee-waiver voucher. Each panel conditions on one offer vector, so its cells sum to one; the final column reports take-up margins and the final row outcome margins. The rows of zeros record one-sided noncompliance: no spouse who was not offered a voucher opened an account, whatever the partner was offered.}
\label{tab:observable_distribution}
\end{table}
\normalsize
These probabilities are consistent with one-sided non-compliance. This may be because a necessary condition for participation in the trial is not having an active bank account. Hence, conditional on being in the sample, individuals may have high subjective costs of banking that are only changed by the fee waiver. Taking the account opening fee at the village bank as a cost benchmark\footnote{The voucher's offer covers both opening and maintenance fees. These are, respectively: \$3.75 opening fee and \$0 monthly maintenance fee (for the village bank) and \$0 opening fee and \$0.60 monthly maintenance fee (for the commercial bank).} the complete lack of account opening highlights high entry costs for individuals in the sample. Non-participation is stable over time: no individual in households who do not receive the voucher ($Z_g=(0,0)$) opens a bank account in the following years either. Most, but not all, individuals open an account after receiving the voucher ($Z_g=(0,1)$ or $Z_g = (1,0)$), suggesting that the voucher changes net incentives enough to overcome these baseline barriers through a combination of fee waivers, increased salience, and shifts in household dynamics. Take-up probabilities differ slightly by gender, maybe reflecting differences in the subjective costs of banking between men and women.

In terms of best responses, shielding one's income in a bank account may be optimal when a spouse does the same. In the data, this happens only when the cost is subsidized: no spouse without a voucher (regardless of their spouse's treatment assignment status) opens a bank account.\footnote{We note the possibility that when one spouse gains access, the other spouse could control the account. Coercion is unlikely to be important here, as we would see more than 5\% of individuals forcibly opting for a joint bank account when offered one.} When both spouses are assigned to treatment, the compliance patterns are supported on all treatment outcomes. This is consistent with two hypotheses: $i)$ even if there are complementarities in account opening, they are not sufficiently strong to overcome the negative sample selection for high costs from banking, $ii)$ the information treatment from the spouse receiving the voucher is weak.\footnote{The weakness of the information treatment may be because the receipt of vouchers is not directly observable by the other spouse and communication between spouses has frictions \citep{delavande2025information} or because the update regarding the cognitive cost of banking is limited.} The gender asymmetry of the responses may reflect different subjective costs due to gender-specific habits and/or work patterns.

To complement the reduced-form analysis without restricting the non-compliance patterns, we obtain bounds for direct effects and  indirect effects, as defined in \eqref{eq:direct_effect} and \eqref{eq:indirect_effect}:
\begin{align*}
	\text{ADE}_{0, m} &:= \mathbb E[Y_{mg} (1,0) - Y_{mg} (0,0)] \ \ 
        &\text{AIE}_{0, m} &:= \mathbb E[Y_{mg} (0,1) - Y_{mg} (0,0)] \\
        \text{ADE}_{0, f} &:= \mathbb E[Y_{fg} (1,0) - Y_{fg} (0,0)] \ \
	&\text{AIE}_{0, f} &:= \mathbb E[Y_{fg} (0,1) - Y_{fg} (0,0)]
\end{align*}
Table \ref{tab:bounds_ade_ase} reports estimated bounds for the four parameters under strong strategic neutrality, as this restriction reduces dimensionality without changing the identified set.\footnote{This was verified by running both linear programs, with and without any restrictions. The results are numerically identical, as predicted by Theorem \ref{thm:equivalence_dominance_multiagent_distr}.}

\begin{table}
\footnotesize
    \centering
    \begin{tabular}{lccc}
        \hline
        Parameter & Lower Bound & Upper Bound & 95\% CI \\
        \hline

        $\text{ADE}_{0,m}
        := \mathbb{E}[Y_{mg}(1,0) - Y_{mg}(0,0)]$
        & 0.07 & 0.38 & [-0.15,\ 0.71] \\[0.2cm]

        $\text{AIE}_{0,m}
        := \mathbb{E}[Y_{mg}(0,1) - Y_{mg}(0,0)]$
        & -0.04 & 0.31 & [-0.23,\ 0.59] \\[0.2cm]

        $\text{ADE}_{0,f}
        := \mathbb{E}[Y_{fg}(1,0) - Y_{fg}(0,0)]$
        & -0.07 & 0.28 & [-0.29,\ 0.61] \\[0.2cm]

        $\text{AIE}_{0,f}
        := \mathbb{E}[Y_{fg}(0,1) - Y_{fg}(0,0)]$

        & -0.10 & 0.22 & [-0.32,\ 0.55] \\

        \hline
    \end{tabular}
    \caption{\footnotesize Bounds and 95\% confidence intervals for direct and indirect effects for men and women in round 6. Outcome: an indicator for the individual earning income from a business they own. Sample: dual-headed households in round 6 ($G=411$), one observation per household. Bounds are the optimal values of the linear programs in Algorithm \ref{alg:estimation} over the latent type space, imposed under (SSN); this leaves the identified set unchanged while reducing the dimension of the program, as verified by running the program with and without the restriction. Confidence intervals are obtained by grid inversion and using the normalization $\widehat\Lambda$ defined after \eqref{eq:gm-statistic} in Appendix \ref{AppInference}, with $B=500$ bootstrap draws and tuning parameter $\sqrt{\log G}$.}
    \label{tab:bounds_ade_ase}
\end{table}

\subsection{Indirect Offer Effects}\label{sec:indirect_offer_effects}
In this section we consider the Average Indirect Female Offer Effect on Men (AIFM) and the Average Indirect Male Offer Effect on Men (AIMM), from equations \eqref{eq:mediated_effect_1} and \eqref{eq:mediated_effect_2}:
\begin{align*}
    \text{AIFM} &= \mathbb{E}[Y_{mg}(0,D_{fg}(1,0,0))] - \mathbb{E}[Y_{mg}(0,D_{fg}(0,0,0))],  \\
    \text{AIMM} &= \mathbb{E}[Y_{mg}(0,D_{fg}(0,1,0))] - \mathbb{E}[Y_{mg}(0,D_{fg}(0,0,0))].
\end{align*}
Under strong strategic neutrality (SSN), an individual's take-up response does not depend on their spouse's take-up decision. In the two-spouse notation, this means that for each $i \in \{m,f\}$ and each offer pair $(z_{ig},z_{-i,g})$, $D_{ig}(z_{ig},z_{-i,g},0) = D_{ig}(z_{ig},z_{-i,g},1) \equiv D_{ig}(z_{ig},z_{-i,g}).$ Under this restriction, the indirect offer effects can be written as
\begin{align*}
    \text{AIFM}^{SSN} &= \mathbb{E}[Y_{mg}(0,D_{fg}(1,0))] - \mathbb{E}[Y_{mg}(0,D_{fg}(0,0))],  \\
    \text{AIMM}^{SSN} &= \mathbb{E}[Y_{mg}(0,D_{fg}(0,1))] - \mathbb{E}[Y_{mg}(0,D_{fg}(0,0))].
\end{align*}

Table \ref{tab:observable_distribution} shows that women who are not offered a voucher do not open a bank account, regardless of whether their spouse is offered one. Under SSN this observed one-sided non-compliance implies $D_{fg}(0,0) = D_{fg}(0,1) = 0,$ and therefore $\text{AIMM}^{SSN}= \mathbb{E}[Y_{mg}(0,0)] - \mathbb{E}[Y_{mg}(0,0)]= 0.$ This is indeed verified when we run the linear program. Without imposing SSN, the linear program gives $\widehat{\text{AIMM}} \in [-0.108,0.398].$ The reason is that AIMM depends on $D_{fg}(0,1,0)$, the female best response when the male is offered the voucher but male take-up is fixed at zero. In the observed data for $Z_g=(1,0)$, female take-up is always zero, but this reveals $D_{fg}(0,1,0)$ only in households where the male spouse does not open an account. In the remaining households, about 69\%, he does open one, so the female spouse's observed take-up is her best response to $d_m=1$, $D_{fg}(0,1,1)=0$, while her response to $d_m=0$ is off the equilibrium path and unrestricted by the data. For example, a female spouse who opens an account only if her husband does not, $D_{fg}(0,1,d_m)=1-d_m$, is consistent with the observed data but has $D_{fg}(0,1,0)=1$. Once SSN is imposed, this off-equilibrium ambiguity disappears because $D_{fg}(0,1,0)=D_{fg}(0,1,1)=D_{fg}(0,1)=0$, and hence $\widehat{\text{AIMM}}^{SSN}=0$.

AIFM behaves differently. Its point identification comes from the compliance pattern, not from any
restriction on behavior: men who are not offered a voucher never open an account. When $Z_g=(0,1)$,
every household therefore plays an equilibrium with $d_m=0$, in which the female spouse's take-up is
her best response to $d_m=0$, $D_{fg}(1,0,0)$, and the observed male outcome is
$Y_{mg}(0,D_{fg}(1,0,0))$. Likewise, when $Z_g=(0,0)$ the observed male outcome is
$Y_{mg}(0,D_{fg}(0,0,0))$. Both terms of AIFM are thus means of observed male outcomes in a single
offer cell, for every distribution of types consistent with the data, and AIFM is point-identified
without any restriction on best responses. This is a feature of the data rather than of the estimand:
if some men opened an account without a voucher, AIFM would depend on female best responses to
actions their husbands do not take in equilibrium, and would be only partially identified. The
estimator of the first term is the fraction of households with $Y_{mg}=1$ in the $Z_g=(0,1)$ panel,
that is, the cells with $Y_g=(1,0)$ and $Y_g=(1,1)$, summed over all realized female take-up values;
the second is the corresponding fraction in the $Z_g=(0,0)$ panel. Table
\ref{tab:observable_distribution} implies
\begin{align*}
    \widehat{\mathbb{E}}[Y_{mg}(0,D_{fg}(1,0,0))]
    &= \widehat{\mathbb{P}}(Y_{mg}=1 \mid Z_{mg}=0,Z_{fg}=1) \\
    &= \widehat{\mathbb{P}}(Y_g=(1,0)\mid Z_g=(0,1))
    + \widehat{\mathbb{P}}(Y_g=(1,1)\mid Z_g=(0,1)) \\
    &= 0.153 + 0.082 = 0.235, \\
    \widehat{\mathbb{E}}[Y_{mg}(0,D_{fg}(0,0,0))]
    &= \widehat{\mathbb{P}}(Y_{mg}=1 \mid Z_{mg}=0,Z_{fg}=0) \\
    &= \widehat{\mathbb{P}}(Y_g=(1,0)\mid Z_g=(0,0))
    + \widehat{\mathbb{P}}(Y_g =(1,1)\mid Z_g=(0,0)) \\
    &= 0.125 + 0.056 = 0.181,
\end{align*}
so that $\widehat{\text{AIFM}} = 0.054$.

Imposing SSN cannot change this value. In principle SSN could narrow the identified set of an
offer-mediated effect, as it does for AIMM, but AIFM is already a single point, and SSN is
non-refutable. Under SSN, $D_{fg}(1,0,0)=D_{fg}(1,0)$, so $\text{AIFM}^{SSN}$ is the same estimand
restricted to SSN types, and $\text{AIFM}=\text{AIFM}^{SSN}=0.054$; the linear program under SSN
returns the bounds $[0.054,0.054]$. Unlike AIMM, AIFM does not collapse to zero under SSN, because the
female offer does change female take-up: in Table \ref{tab:observable_distribution}, female take-up
is positive when only women are offered the voucher. We therefore compute inference for AIFM under
SSN, which reduces the size of the linear program; since SSN is non-refutable and leaves the
identified value unchanged, the resulting confidence interval is also valid for AIFM without SSN. Table \ref{tab:bounds_indirect_offer} summarizes the estimated bounds.

\begin{table}
\footnotesize
    \centering
    \begin{tabular}{lccc}
        \hline
        Parameter & $\begin{array}{c} \text{Lower} \\ \text{Bound} \end{array}$ & $\begin{array}{c} \text{Upper} \\ \text{Bound} \end{array}$ & 95\% CI \\
        \hline

        $\text{AIFM} = \text{AIFM}^{SSN}$
        & 0.054 & 0.054 & $[-0.216,\ 0.224]$ \\[0.2cm]

        $\text{AIMM}$
        & $-0.108$ & 0.398 & $[-0.268,\ 0.582]$ \\

        \hline
    \end{tabular}
    \caption{\footnotesize Bounds and 95\% confidence intervals for indirect
    offer effects on men using round 6. AIFM is point-identified with or without SSN, because men who are not offered a voucher never open an account. AIMM is computed without restrictions and under SSN it equals zero. The normalization is $\widehat\Lambda$ as defined after \eqref{eq:gm-statistic}, $B=500$, and tuning parameter $\sqrt{\log G}$. Outcome: an indicator for the male spouse earning income from a business he owns. Sample: dual-headed households in round 6 ($G=411$), one observation per household. Both estimands hold male take-up fixed at $d_m=0$, so they isolate the effect of the offers that operates through female account opening.}
    \label{tab:bounds_indirect_offer}
\end{table}

\subsection{Testing the Behavioral Restrictions}\label{sec:solvability_tests}

The bounds in Section \ref{sec:fixed_allocation} impose SSN for computational reasons, as it greatly reduces the cardinality of the latent space without changing the results. In contrast, the indirect offer effects studied in Section \ref{sec:indirect_offer_effects} can differ depending on whether one imposes SSN.

Using the solvability test of \citet[Section 3.1]{goff2025inferencevaluelinearprogram}, we ask if restrictions are compatible with the data. The test statistic, denoted $\hat T_G$, is the optimal value of the dual linear program in \eqref{eq:gm-statistic} (Appendix \ref{AppInference}). It is non-negative and equals zero if and only if the restriction rationalizes $\hat p$ exactly.\footnote{It is unrelated to the latent household types $T_g$.} We reject when $\sqrt G\hat T_G$ exceeds the bootstrap critical value. We test sub/supermodularity, dominance (equivalently, SSN), best-response symmetry in the equivariant sense of Definition \ref{Def:Equivariant symmetry}, $D_{mg}(z,z',d') = D_{fg}(z,z',d')$ for all $(z,z',d')$, which is what identical utility functions imply, and individualized offer responses (IOR), $D_{ig}(z,z',d_{-i}) = D_{ig}(z,d_{-i})$, under which take-up does not respond to the partner's offer but may still respond to the partner's take-up. Combined with SSN, IOR gives the no-interference restriction $D_{ig}(z,z',d_{-i}) = D_{ig}(z)$ of \citet{DITRAGLIA20231589}.

Appendix \ref{sec:symmetry_refutability} introduces two notions of best-response symmetry, individual and equivariant, and shows that individual symmetry is non-refutable. We test equivariant symmetry, which is the notion implied by identical utility functions and is not covered by our non-refutability result. As Table \ref{tab:solvability_tests} shows, we cannot reject it at conventional significance levels (p-value $0.987$). We also consider two restrictions involving potential outcomes. Outcome symmetry requires both spouses to have the same outcome under every allocation,
that is, $Y_{mg}(d_1,d_2) = Y_{fg}(d_2,d_1)$ for all $(d_1,d_2) \in \{0,1\}^2$,
and full symmetry imposes best-response equivariant symmetry and outcome symmetry jointly. Both restrict each household rather than the distribution of households: they require, for instance, $Y_{mg}(0,0)=Y_{fg}(0,0)$. Because no spouse without a voucher opens an account, this implies $\Pr(Y_{mg}\neq Y_{fg}\mid Z_g=(0,0))=0$, whereas Table \ref{tab:observable_distribution} reports $0.208+0.125=0.333$. Table \ref{tab:solvability_tests} reports the results. The three strategic-interaction restrictions in the bottom panel, supermodularity, submodularity, and dominance (SSN), yield $\hat{T}_G = 0$. This follows from Theorem \ref{thm:equivalence_dominance_multiagent_distr} and Appendix \ref{sec:supermodularity}: dominance is non-refutable, and dominant best-response functions are contained in both the supermodular and submodular classes. 

\begin{table}
\footnotesize
    \centering
    \begin{tabular}{lrcccc}
        \hline
        Restriction & Types & $\sqrt{G}\hat{T}_G$ & CV (95\%) & $p$-value & Reject at 5\% \\
        \hline
        IOR                      & 69{,}632    & 0.168 & 2.209 & 1.000 & No  \\
        Best-response symmetry (equivariant)  & 184{,}320   & 0.990 & 1.938 & 0.987 & No  \\
        Outcome symmetry         & 1{,}048{,}576 & 4.339 & 2.981 & 0.004 & Yes \\
        Full symmetry (equivariant)           & 11{,}520    & 4.339 & 2.981 & 0.004 & Yes \\
        \hline
        Supermodularity          & 2{,}560{,}000 & 0.000 & 2.140 & 1.000 & No  \\
        Submodularity            & 2{,}560{,}000 & 0.000 & 2.140 & 1.000 & No  \\
        Dominance (SSN)          & 65{,}536    & 0.000 & 2.140 & 1.000 & No  \\
        \hline
    \end{tabular}
    \caption{\footnotesize Solvability tests of \citet[Section 3.1]{goff2025inferencevaluelinearprogram} for each behavioral restriction, using round 6 ($G=411$). $\hat T_G$ is the test statistic of \citet[eq.~(3.1)]{goff2025inferencevaluelinearprogram}. Best-response symmetry is the equivariant notion of Definition \ref{Def:Equivariant symmetry}, $D_{mg}(z,z',d')=D_{fg}(z,z',d')$, and full symmetry adds outcome symmetry. Critical values and $p$-values are from a nonparametric bootstrap with $B=1000$ draws and tuning parameter $ \sqrt{\log G} \approx 2.45$. Results are unchanged under a tuning parameter $\log G \approx 6.02$, for which the critical values range from 2.386 to 2.981 and the $p$-values are identical to those reported. Outcome and sample as in Table \ref{tab:bounds_ade_ase}: the own-enterprise indicator for dual-headed households in round 6 ($G=411$). ``Types'' counts augmented types, that is, columns of $A$: a pair of best-response and potential-outcome profiles together with a deterministic equilibrium-selection rule, out of $16{,}777{,}216$ in total. Restrictions are imposed on both spouses.}
    \label{tab:solvability_tests}
\end{table}

\subsection{Testing Monotone Treatment Selection}\label{sec:mts}

The restrictions in Table \ref{tab:solvability_tests} restrict the latent type space through known incidence or support restrictions. By contrast, monotone treatment selection (MTS) leaves the incidence matrix \(A\) fixed and introduces additional inequality restrictions of the form $R\mu\le 0$,
where $R$ depends on \(\kappa=\mathbb{P}(D_{ig}=0)/\mathbb{P}(D_{ig}=1)\) for $i \in \{m,f\}$. When \(\kappa\) is replaced by its sample analogue, the corresponding coefficients in the inequality matrix $\widehat{R}$ are estimated. Thus, testing MTS requires the estimated-coefficient-matrix procedure of \citet[Section 3.2]{goff2025inferencevaluelinearprogram}, with the weighting matrix ($\Omega$) estimated by bootstrap and the test statistic $\hat T_G$ is now that of \citet[eq.~(3.6)]{goff2025inferencevaluelinearprogram}. Table \ref{tab:mts_tests} reports the results. MTS is not rejected in any specification, but the two sides of the household behave differently. For women, the test statistic is $\hat T_G = 0$: the inequalities hold exactly at $\hat p$, and imposing them leaves every estimated bound in Tables \ref{tab:bounds_ade_ase} and \ref{tab:bounds_indirect_offer} unchanged.\footnote{Under MTS for women, the estimated bounds are $[0.069,0.382]$ for $\text{ADE}_{0,m}$, $[-0.038,0.309]$ for $\text{AIE}_{0,m}$, $[-0.070,0.277]$ for $\text{ADE}_{0,f}$, $[-0.097,0.215]$ for $\text{AIE}_{0,f}$, $0.054$ for AIFM and $[-0.108,0.398]$ for AIMM, identical to the unrestricted values in Tables \ref{tab:bounds_ade_ase} and \ref{tab:bounds_indirect_offer}. Adding strong strategic neutrality sets AIMM to zero, as without MTS.} For men, the test statistic is $\hat T_G = 0.0022 > 0$: the sample probabilities violate MTS slightly, so the linear program with the male MTS rows is infeasible at $\hat p$ and the estimated identified set is empty. The violation is within sampling error, so MTS for men is not refuted. The same holds when the inequalities for both spouses are imposed.
\begin{table}[!htbp]
\footnotesize
    \centering
    \begin{tabular}{lccccc}
        \hline
        Restriction & $\hat{T}_G$ & $\sqrt{G}\hat{T}_G$ & CV (95\%) & $p$-value & Reject at 5\% \\
        \hline
        MTS men                & 0.0022 & 0.045 & 2.377 & 1.000 & No \\
        MTS women              & 0.0000 & 0.000 & 2.443 & 1.000 & No \\
        MTS men $+$ women      & 0.0022 & 0.045 & 2.390 & 1.000 & No \\
        \hline
    \end{tabular}
    \caption{\footnotesize Estimated-coefficient-matrix solvability tests of \citet[Section 3.2]{goff2025inferencevaluelinearprogram} for monotone treatment selection, using round 6 ($G=411$). ``MTS men'' imposes the four inequalities for the male spouse, ``MTS women'' for the female spouse. $\hat{T}_G$ is the test statistic. Critical values and $p$-values are from a nonparametric bootstrap with $B=1000$ draws and tuning parameter $\sqrt{\log G}$; $\Omega$ is estimated by bootstrap. Outcome and sample as in Table \ref{tab:bounds_ade_ase}: the own-enterprise indicator for dual-headed households in round 6 ($G=411$). Inequalities are imposed on the spouse's potential outcomes, one for each allocation $(d_m,d_f)$, so ``MTS men'' and ``MTS women'' add four inequalities each and ``MTS men $+$ women'' adds eight.}
    \label{tab:mts_tests}
\end{table}

\section{Conclusion}\label{Sconclu}
We propose a linear programming formulation to obtain bounds on a variety of
treatment effect parameters in the presence of spillovers and strategic
interactions. The possibility of strategic interactions results in an incomplete
model, in the sense that the map from latent probabilities to observed
probabilities is not unique. Different restrictions that the literature has
commonly used can be incorporated into the linear program as a way of
restricting the support of latent probabilities. In this way, their joint or
separate identifying power can be gauged against their impact on the width of
the identified set.

A conventional analysis of \citet{dupas2019effect} on the round-6 sample would report a direct effect of account ownership on men's entrepreneurial activity of $0.207$ with a standard error of $0.096$ (Table \ref{tab:2sls_by_round}), a 95\% interval of $[0.02,0.40]$ that excludes zero. On the same sample, the identified set for $\text{ADE}_{0,m}$ is $[0.07,0.38]$ and its confidence interval is $[-0.15,0.71]$ (Table \ref{tab:bounds_ade_ase}). The estimated bounds are themselves positive and contain the 2SLS estimate, but their confidence interval includes zero: ruling out a null effect at conventional levels relies on the noncompliance assumptions behind the 2SLS interpretation. The own-offer Wald ratio for men is $0.207$ when the wife receives no voucher and $0.049$ when she does (Table \ref{tab:wald_conditional}). The difference is not statistically significant, but under strategic take-up neither ratio has a direct-effect interpretation (Lemma \ref{lem:ec_wald_decomposition}).
For indirect effects, under strong strategic neutrality
$\text{AIMM}$ is exactly zero, because women who are not offered a voucher never
open an account; without that restriction the identified set is
$[-0.108,0.398]$ (Table \ref{tab:bounds_indirect_offer}) and not even the sign
is determined. %

Being able to test the restrictions separates two questions: whether the data refute an assumption, and whether the assumption
buys anything in terms of identifying power. Supermodularity, submodularity and dominance are each exactly compatible with these data (Table \ref{tab:solvability_tests}), yet none of them narrows the identified set for fixed-allocation effects: dominance by Theorem \ref{thm:equivalence_dominance_multiagent_distr}, and super- and submodularity because both are implied by dominance (Appendix \ref{sec:supermodularity}). The same restrictions can matter for offer-mediated effects, as the AIMM example above shows. Monotone treatment selection is likewise not rejected (Table \ref{tab:mts_tests}): for women it leaves the bounds unchanged, while for men its sample analogue is marginally infeasible. Best-response symmetry is not rejected ($p = 0.987$). Outcome symmetry is rejected ($p = 0.004$), but mechanically: it requires spouses to have identical outcomes when neither is treated, while in a third of untreated households exactly one spouse earns business income. Which assumptions are refutable, and which are
consequential, is therefore an output of the framework rather than something the
researcher must settle in advance.

%% file: appendix_all.tex
\clearpage
\begin{center}
\Large\bfseries Appendix
\end{center}

\setcounter{section}{0}

\renewcommand{\thesection}{\Alph{section}}
\renewcommand{\thesubsection}{\thesection.\arabic{subsection}}
\setcounter{algocf}{0}
\renewcommand{\thealgocf}{\thesection.\arabic{algocf}}

\setcounter{equation}{0} \renewcommand\theequation{A.\arabic{equation}}

\section{Constructing Type Matrix}

\SetKwComment{Comment}{/* }{ */}

\begin{algorithm}[!ht]
    \caption{Mapping types to observables}
    \label{alg:types}
\textbf{Input:} Researcher specifies the number of agents $N$ per group and the support $\{ \mathcal{Y},\mathcal{D},\mathcal{Z}\}$ of outcomes, take-up, and offers. No data used at this stage.

Let $|\cdot|$ denote the size of each set.

    $\begin{array}{l} \mathcal{W} = \text{list}\!\left(\mathcal{Y}^N \times \mathcal{D}^N \times \mathcal{Z}^N \right) \\
    \mathcal{S} = \text{list}\!\left(\mathcal{Y}^{|\mathcal{D}|^N}  \times \mathcal{D}^{|\mathcal{Z}|^N  |\mathcal{D}|^{N-1}}\right)  \\
    \mathcal{S}^N = \text{list}\!\left( \mathcal{S}_1 \times \ldots \times \mathcal{S}_N \right) \\
    \mathcal{Q} = \text{list}\!\left( \left(\mathcal{D}^N\right)^{\mathcal{Z}^N} \right)
    \end{array} 
    \begin{array}{l} \text{Compute permutations numerically and store} \\ \text{as lists of vectors with labeled coordinates. For} \\  \text{example, if $\mathcal{Y} = \mathcal{D} = \mathcal{Z} = \{0,1\}$ and $N = 2$ } \\ \text{then $\mathcal{W}$ lists all $(y_1,y_2,d_1,d_2,z_1,z_2) \in \{0,1\}^6$ } \\
    \text{$\mathcal{S}$ all elements in $\{0,1\}^{12}$ (Table \ref{tab:illustration_individual_type} labels), and } \\
    \text{$\mathcal{Q}$ all $4^4=256$ maps assigning a profile} \\
    \text{$\vec{D}^e(\vec{z}) \in \{0,1\}^2$ to each $\vec{z} = (z_1,z_2) \in \{0,1\}^2$.}
    \end{array}$

    $\mathcal{T} = \varnothing $

    \For{$\vec{s} \in \mathcal{S}^N, \vec{D}^e \in \mathcal{Q}$ }{

        $count = 0$
        \For{$\vec{z} \in \mathcal{Z}^N$ }{
        \text{ }
        \If{ $\vec{D}^e(\vec{z}) \in \text{Nash}(\vec{s},\vec{z})$ }{
            $count = count + 1 \gets \begin{array}{l} \text{Evaluate whether $\vec{D}^e(\vec{z})$ is a} \\ \text{Nash equilibrium given types $\vec{s}$ and offers $\vec{z}$.}  \end{array} $
        }
        }

        \If{ $count = |\mathcal{Z}^N|$ }{
            $\mathcal{T} = \mathcal{T} \bigcup \{ (\vec{s},e) \} \gets \begin{array}{l}  \text{Add types to set $\mathcal{T}$ only if $\vec{D}^e(\vec{z})$ is a Nash } \\
            \text{equilibrium  for all $\vec{z} \in \mathcal{Z}^N$. Store as lists of vectors} \\ \text{with labeled coordinates. Under Assumption \ref{assump:deterministic_equilibria},}  \\ \text{each $\vec{D}^e$ can be mapped to a corresponding $e$ rule, } \\
            \text{and vice versa.} \end{array}$
        }

        }

     $A \gets $Empty $|\mathcal{W}| \times |\mathcal{T}|$ matrix. %

    \For{ $k \in \{1,\ldots,|\mathcal{W}|\}$ }{
    $(\vec{y},\vec{d},\vec{z}) \gets$  Extract $k^{th}$ element of $\mathcal{W}$  \\
        \For{ $\ell \in \{1,\ldots,|\mathcal{T}|\}$}{

        \text{ }
         $(\vec{s},e) \gets $ Extract $\ell^{th}$ element of $\mathcal{T}$

        $ \begin{array}{l} \vec{d}^* = \vec{D}^{e}(\vec{z}) \\
        \vec{Y}(\vec{d}^*;\vec{s}) = \{Y( d_i^*,d_{-i}^* ; s_i)\}_{i=1}^N \\
        A_{k \ell} = \begin{cases} 1 & \text{ if } \left( \vec{Y}(\vec{d}^*;\vec{s}), \vec{d}^* \right)= (\vec{y},\vec{d}) \\
        0 & \text{otherwise}
        \end{cases} \\
        \end{array} \gets  \begin{array}{l} \text{Determine if type $(\vec{s},e)$} \\ \text{would report $(\vec{y},\vec{d})$} \\ \text{under offers $\vec{z}$.}
        \end{array}
        $

        }

    }

\end{algorithm}

  \section{Proofs of Main Results}\label{AdditionalProofs}  

\begin{proof}[Proof of Lemma \ref{lem:mts_equivalence}]
Fix a counterfactual allocation \((a,b)\in\{0,1\}^2\). By Bayes' rule and the positivity condition \(0<\mathbb{P}(D_{mg}=1)<1\), $\mathbb{P}(Y_{mg}(a,b)=1\mid D_{mg}=d)
=
\frac{\mathbb{P}(Y_{mg}(a,b)=1,D_{mg}=d)}{\mathbb{P}(D_{mg}=d)}$, for $d \in \{0,1\}$. Therefore, the MTS inequality can be rewritten as
\begin{align*}
&\mathbb{P}(Y_{mg}(a,b)=1\mid D_{mg}=0)
-
\mathbb{P}(Y_{mg}(a,b)=1\mid D_{mg}=1) \le 0 \\
&\iff \mathbb{P}(Y_{mg}(a,b)=1,D_{mg}=0)
-
\frac{\mathbb{P}(D_{mg}=0)}{\mathbb{P}(D_{mg}=1)}
\mathbb{P}(Y_{mg}(a,b)=1,D_{mg}=1)
\le 0.
\end{align*}
Let $\kappa_m=\frac{\mathbb{P}(D_{mg}=0)}{\mathbb{P}(D_{mg}=1)}$. For each type \(t_\ell=(\vec s_\ell,e_\ell)\), and offer vector $\vec{z}$, the potential outcome \(Y_{mg}(a,b)\) and the realized ``reduced form'' male take-up \(D_{m}^{e_\ell}(\vec{z})\) are deterministic,
\begin{align*}
&\mathbb{P}(Y_{mg}(a,b)=1,D_{mg}=d\mid T_g=t_\ell,\vec Z_g=\vec z) = \mathbbm{1}
\{Y(a,b;s_{m\ell})=1,\; D^{e_\ell}_m(\vec{z})=d\}.
\end{align*}
By the law of iterated expectations and Assumption \ref{AsIndep},
\begin{align*}
    &\mathbb{P}(Y_{mg}(a,b)=1,D_{mg}=d)
=
\sum_{\ell=1}^{|\mathcal T|}
\sum_{\vec{z}}
\mathbbm{1}
\{Y(a,b;s_{m\ell})=1,\; D^{e_\ell}_m(\vec{z})=d\}
\mathbb{P}(\vec Z_g=\vec{z})
\mathbb{P}(T_g=t_\ell).
\end{align*}
Substituting this expression for \(d=0\) and \(d=1\) into the MTS inequality gives the restriction $\sum_{\ell=1}^{|\mathcal T|}
R^{a,b}_\ell \mathbb{P}(T_g=t_\ell)
\le 0,$ where \(R^{a,b}_\ell\) is the row entry defined in Lemma \ref{lem:mts_equivalence}. Writing \(\mu_\ell=\mathbb{P}(T_g=t_\ell)\) gives $
R^{a,b}\mu\le 0.$
If MTS is imposed for all \((a,b)\in\{0,1\}^2\), the same argument applies allocation by allocation, and the corresponding rows are stacked into \(R^{MTS}\mu\le 0\).
\end{proof}

\begin{proof}[Proof of Lemma \ref{lem:ec_wald_decomposition}] Assumptions \ref{AsExclu}-\ref{assump:deterministic_equilibria} imply, for $z_1\in\{0,1\}$, \begin{align*} \mathbb E[Y_{1g}\mid (Z_{1g},Z_{2g})=(z_1,z_2)] &=\mathbb E\!\left[Y_{1g}(D_{1g}^{E_g}(z_1,z_2),D_{2g}^{E_g}(z_1,z_2))\right],\\ \mathbb E[D_{1g}\mid (Z_{1g},Z_{2g})=(z_1,z_2)] &=\mathbb E[D_{1g}^{E_g}(z_1,z_2)]. \end{align*} Binary take-up and monotonicity give $D_{1g}^{E_g}(1,z_2)-D_{1g}^{E_g}(0,z_2) =\mathbf{1}\{C_{1g}(z_2)\}$ so the denominator in \eqref{eq:ec_wald_decomposition} equals $\mathbb{P}(C_{1g}(z_2))$. By the law of iterated expectations, \begin{align*} & \mathbb{E}\!\left[Y_{1g}(D_{1g}^{E_g}(z_1,z_2),D_{2g}^{E_g}(z_1,z_2))\right] \\ &= \mathbb{P}(C_{1g}(z_2)) \mathbb{E}\!\left[Y_{1g}(D_{1g}^{E_g}(z_1,z_2),D_{2g}^{E_g}(z_1,z_2)) \mid C_{1g}(z_2) \right] \\ &\quad + \mathbb{P}(C_{1g}(z_2)^c) \mathbb{E}\!\left[Y_{1g}(D_{1g}^{E_g}(z_1,z_2),D_{2g}^{E_g}(z_1,z_2)) \mid C_{1g}(z_2)^c \right]. \end{align*} For equilibrium compliers, own take-up changes from zero to one; for equilibrium non-compliers, $D_{1g}^{E_g}(1,z_2)=D_{1g}^{E_g}(0,z_2)$. Substituting these identities into the respective contributions and dividing by $\mathbb{P}(C_{1g}(z_2))$ yields \eqref{eq:ec_wald_decomposition}. \end{proof}

\begin{proof}[Proof of Theorem \ref{prop:chain}] \text{ }\\
\ $(i)$ This is immediate: $\Theta(\mathcal{M}^*,p) = \Theta(\mathcal{M},p) \neq \varnothing$, first equality by irrelevance and second inequality by non-refutability.
\ $(ii)$ Because $\mathcal{M}^* \subseteq \mathcal{M}' \subseteq \mathcal{M}$ are a nested set of restrictions over the set of measures, then $\Theta(\mathcal{M}^*,p) \subseteq \Theta(\mathcal{M}',p) \subseteq \Theta(\mathcal{M},p)$. By assumption, $\Theta(\mathcal{M}^*,p) = \Theta(\mathcal{M},p)$, and hence $\Theta(\mathcal{M}',p) = \Theta(\mathcal{M},p)$ by a sandwich argument.
$(iii)$ Exactly as above,  $\Theta(\mathcal{M'},p) \subseteq \Theta(\mathcal{M},p)$ Then non-refutability of $\mathcal{M}'$ implies $\exists \theta \in \Theta(\mathcal{M'},p) \subseteq \Theta(\mathcal{M},p)$. Hence $\Theta(\mathcal{M},p) \neq \varnothing$ and $\mathcal{M}$ is non-refutable.
\end{proof}

    \begin{proof}[Proof of Theorem \ref{thm:equivalence_dominance_multiagent_distr}]
Since $\mathcal M_{dominant}=\Delta(\mathcal T_{dominant}\cap\mathcal T^*)\subseteq\mathcal M$,
$\Theta(\mathcal M_{dominant},p)\subseteq\Theta(\mathcal M,p)$.
Conversely, let $\theta\in\Theta(\mathcal M,p)$, attained by $\mu\in\Delta(\mathcal T^*)$ with $\theta = \sum_{t \in \mathcal{T}^*}\theta(t)\mu(t) $ and $\mathbb{P}( (\vec{Y}_g,\vec{D}_g) = (\vec{y},\vec{d}) \mid \vec{Z}_g = \vec{z}) = \sum_{t \in \mathcal{T}^*} \mathbbm{1}\{ h(t,\vec{z}) =  (\vec{y},\vec{d}) \} \mu(t)$ for all values in the observed support. Define the pushforward measure
$\mu^\varphi(t'):=\sum_{t:\varphi(t)=t'}\mu(t)$, 
which is a probability measure because $\sum_{t' \in \varphi(\mathcal{T}^*)}\mu^\varphi(t')=\sum_{t \in \mathcal{T}^*} \mu(t)=1$.

We start by analyzing the properties of the neutral counterpart. Each best response of $s_i^\varphi$ is constant in $\vec d_{-i}$, so
$\text{Nash}(\vec s^{\,\varphi},\vec z)=\{\vec D^{e}(\vec z)\}$ for every $\vec z$. It is non-empty and
$e$ places all its mass on it, so Assumption \ref{AsNash} holds. The type $\varphi(t)$ also inherits Assumptions \ref{AsExclu}
and \ref{assump:deterministic_equilibria} because the outcomes and $e$ are unchanged, so that $\varphi(t) \in \mathcal{T}$ (it is a valid type). Furthermore, since the best-responses are dominant, $\varphi(t) \in \mathcal{T}_{dominant} \subseteq \mathcal{T}$ and by the closure condition in Assumption \ref{assump:closure_dominant_multiagent}, 
$\mu^\varphi$ is supported on $\mathcal T_{dominant}\cap\mathcal T^*$. 

By \eqref{eq:identityobservational}, $h(t,\vec z)=(\vec Y(\vec D^{e}(\vec z);\vec s),\vec D^{e}(\vec z))$
depends on $\vec s$ only through outcome schedules, which $\varphi$ preserves. This means that $h(\varphi(t),\vec{z}) = h(t,\vec{z})$. We can show that the pushforward measure satisfies the observed restrictions, by substituting the definitions
\begin{align*}
&\mathbb{P}\big((\vec Y_g,\vec D_g)=(\vec y,\vec d)\mid\vec Z_g=\vec z\big) \\
 &=\sum_{t\in\mathcal T^*}\mathbbm{1}\{h(t,\vec z)=(\vec y,\vec d)\}\mu(t)
 =\sum_{t\in\mathcal T^*}\mathbbm{1}\{h(\varphi(t),\vec z)=(\vec y,\vec d)\}\mu(t)\\
&=\sum_{t'\in\varphi(\mathcal T^*)}\sum_{t:\varphi(t)=t'}\mathbbm{1}\{h(t',\vec z)=(\vec y,\vec d)\}\mu(t)
 =\sum_{t'\in\varphi(\mathcal T^*)}\mathbbm{1}\{h(t',\vec z)=(\vec y,\vec d)\}\mu^{\varphi}(t').
\end{align*}

Since $\theta(t)$ only depends on $t$ through the potential outcome schedules, then $\theta(t) = \theta(\varphi(t))$. Then using similar arguments as before,
\begin{equation*}
    \theta = \sum_{t \in \mathcal{T}^*}\theta(t)\mu(t) = \sum_{t \in \mathcal{T}^*}\theta(\varphi(t))\mu(t) = \sum_{t' \in \varphi(\mathcal{T}^*)} \theta(t')\mu^{\varphi}(t').
\end{equation*}
Hence $\mu^{\varphi} \in \mathcal{M}_{dominant}$ rationalizes $p$ and attains $\theta$, so $\theta\in\Theta(\mathcal M_{dominant},p)$.
\end{proof}

\section{Restrictions}  

Table \ref{tab:Restriction table} presents examples of restrictions from the literature. In the next few sections, we analyze some of these restrictions theoretically.

\begin{table}[h]
\vspace*{0cm}
\hspace*{0cm}%
\rotatebox{0}{
\begin{minipage}{0.9\textheight}
\tiny

\begin{tabular}{m{5cm} m{6cm} m{4cm}}
\toprule
\textbf{Name (Source)} & \textbf{Mathematical Formulation} & \textbf{Implication} \\
\midrule

IOR (individualized offer responses) &
$D_{ig}(z,z',d_{-i}) = D_{ig}(z,d_{-i})$ &
No response to the partner's offer; may respond to the partner's take-up. \\
\midrule

No interference in take-up (\citet{DITRAGLIA20231589}) &
$D_{ig}(z,z',d_{-i}) = D_{ig}(z)$ &
IOR $+$ SSN: no response to the partner's offer or take-up. \\
\midrule

Monotonicity (\citet{Imbens1994}) &
$D_{ig}(z,z',d_{-i}) = D_{ig}(z)$ and 
$D_{ig}(0) \leq D_{ig}(1)$ &
Identifies LATE; not ATE. \\
\midrule

Vazquez-Bare Monotonicity (\citet{Vazquez-Bare2020}) &
$D_{ig}(z,z',d_{-i}) = D_{ig}(z,z')$,

$D_{ig}(0,0) \leq D_{ig}(0,1) \leq D_{ig}(1,0) \leq D_{ig}(1,1)$ &
Direct and peer effects identified under one-sided noncompliance. \\
\midrule

Monotonicity and Restricted Interference \citep{ImaiJiangMalani2021} &
$D_{ig}(z,z',d_{-i}) = D_{ig}(z,z')$, 

$D_{ig}(1,z')\geq D_{ig}(0,z')$, 

$D_{ig}(1,z')=D_{ig}(0,z')\implies$ 
$Y_{ig}(D_{ig}(1,z'),D_{-i,g}(z',1))=Y_{ig}(D_{ig}(0,z'),D_{-i,g}(z',0))$
& Own instrument is weakly encouraging; for those it does not move, the partner's
induced take-up must not move the outcome either.\\ 
\midrule

One-sided Noncompliance (\citet{kormos2023}) &
$D_{ig}(z,z',d') = D_{ig}(z,z')$, $D=0$ if $z=0$ &
Restricts treatment access based on assignment. \\
\midrule

Strategic Substitution (\citet{HanBalat2023}) / Complementarity  &
$D_{ig}$ weakly monotonic in $d_{-i}$ &
Captures strategic interaction in take-up. \\
\midrule

Monotone Treatment Response (\citet{Manski1997}) &
$Y_{ig}(d_1,d_2) \geq Y_{ig}(d_1',d_2')$ if $(d_1,d_2) \geq (d_1',d_2')$ &
Outcome monotonicity. \\
\midrule

Monotone Treatment Selection (\citet{Manski2000, manski2009more}) &
For fixed $(d,d')$:
$\mathbb{P}(Y_{ig}(d,d')=1 \mid D_{ig}=1)
\geq
\mathbb{P}(Y_{ig}(d,d')=1 \mid D_{ig}=0)$ &
Positive selection on fixed potential outcomes by realized treatment status. \\
\midrule

Stochastic Dominance (Strata) (\citet{Imai2008,Blanco2013,Flores2013,Huber2017,Bartalotti2021}) &
Outcome distributions ordered across strata &
Selected strata have stochastically better outcomes. \\
\midrule

$\epsilon$-VB Monotonicity (related to \citet{deChaisemartin2017}) &
$\mathbb{P}(\text{violates VB}) \leq \epsilon$ &
Softens VB monotonicity; $\epsilon=0$ recovers original. \\
\midrule

$\epsilon$-Outcome Symmetry &
$\sum_{d_1,d_2}\mathbb P\big(Y_{1g}(d_1,d_2)\ne Y_{2g}(d_2,d_1)\big)\le\epsilon$ &
Allows outcome heterogeneity within pairs, up to $\epsilon$. \\
\midrule

$\epsilon$-Treatment Symmetry (\citet{kline2012bounds}) &
$\sum_{z,z',d}\mathbb P\big(D_{1g}(z,z',d)\ne D_{2g}(z,z',d)\big)\le\epsilon$ &
Allows treatment heterogeneity within pairs, up to $\epsilon$. \\

\midrule

$\epsilon$-Strategic Neutrality &
$\mathbb{P}(D_{ig}(\cdots,d) \neq D_{ig}(\cdots,1-d)) \leq \epsilon$ &
Bounds on strategic influence of partner. \\

\bottomrule
\end{tabular}
\end{minipage}
}
\caption{Examples of type restrictions}
\label{tab:Restriction table}
\end{table}

\subsection{Restricted Interference under Noncompliance}\label{sec:ijm}

\citet{ImaiJiangMalani2021} consider a related setting involving two-stage randomized experiments with interference and noncompliance. Their
identification result rests on two restrictions: a monotonicity
condition on take-up, and a condition limiting how a
unit's own assignment affects its outcome. Because both are restrictions on
the latent type space, they can be encoded and tested with the machinery of
Section \ref{sec:solvability_tests}. Their monotonicity assumption, in the spirit of \cite{Imbens1994}, requires that a unit's own assignment never lowers
its own take-up, whatever the rest of the cluster is assigned:
\begin{align}
	D_{ig}(1, \vec{z}_{-i}) \ \ge \ D_{ig}(0, \vec{z}_{-i})
	\qquad \text{for all } \vec{z}_{-i}.
	\label{eq:ijm_a4}
\end{align}
In \citet{ImaiJiangMalani2021} take-up is a function of the assignment vector alone, $D_{ig}(\vec{z})$. In our framework a latent type carries a best-response function in which the partner's \emph{action} is a separate argument, together with an equilibrium selection rule, so \eqref{eq:ijm_a4} admits two distinct encodings and we report both. The best-response reading imposes monotonicity
uniformly in the partner's action, $D_{mg}(1, z_{f}, d_{f}) \ge D_{mg}(0, z_{f}, d_{f})
	\quad\text{and}\quad
	D_{fg}(z_{m}, 1, d_{m}) \ge D_{fg}(z_{m}, 0, d_{m})$,
for all $(z_{f}, d_{f})$ and $(z_{m}, d_{m})$ respectively, which is eight
inequalities. The equilibrium reading is the literal translation of
$D_{ig}(\vec{z})$: it requires the \emph{realized} take-up, evaluated at the
equilibrium selected by the type, to be weakly increasing in own assignment,
which is four inequalities. Neither reading nests the other. Monotone
best responses need not produce monotone realized take-up, because the equilibrium selection mechanism can be affected by the offers. Conversely, realized take-up
can be monotone even when the underlying best responses are not. Their second condition restricts interference for noncompliers: if a unit's own
take-up does not respond to its own assignment, then it cannot affect the outcome,
\begin{align}
	\text{if } D_{ig}(1, \vec{z}_{-i}) = D_{ig}(0, \vec{z}_{-i}),
	\quad\text{then}\quad
	Y_{ig}\big(\vec{D}_{g}(1, \vec{z}_{-i})\big)
	= Y_{ig}\big(\vec{D}_{g}(0, \vec{z}_{-i})\big).
	\label{eq:ijm_a5}
\end{align}
Restriction \eqref{eq:ijm_a5} is implied by, but strictly
weaker than, IOR: a type satisfying IOR has a partner whose take-up cannot
respond to the focal unit's offer, which is the second of the three sufficient
scenarios \citet{ImaiJiangMalani2021} describe. In our type space all
$69{,}632$ IOR types satisfy \eqref{eq:ijm_a5}, while only $39{,}488$ of the
$65{,}536$ types with dominant best responses do: dominance removes the response
to the partner's action, but leaves intact the channel that \eqref{eq:ijm_a5}
closes. Table \ref{tab:ijm_solvability} reports solvability tests for the two
restrictions, separately and jointly, under both readings of \eqref{eq:ijm_a4}.
None is rejected. In every specification $\hat{T}_G = 0$ exactly, so the
restrictions rationalize $\hat{p}$; by
contrast IOR and monotone treatment selection for men both attain small but
strictly positive values ($0.008$ and $0.002$) in Tables
\ref{tab:solvability_tests} and \ref{tab:mts_tests}.

\begin{table}[h]
\footnotesize
    \centering
    \begin{tabular}{lrcccc}
        \hline
        Restriction & Types & $\sqrt{G}\hat{T}_G$ & CV (95\%) & $p$-value & Reject at 5\% \\
        \hline
        A4, best-response reading   & 1{,}683{,}712  & 0.000 & 2.140 & 1.000 & No \\
        A4, equilibrium reading     & 5{,}308{,}416  & 0.000 & 2.140 & 1.000 & No \\
        A5                          & 10{,}108{,}928 & 0.000 & 2.140 & 1.000 & No \\
        A4 (best-response) $+$ A5   & 1{,}161{,}472  & 0.000 & 2.140 & 1.000 & No \\
        A4 (equilibrium) $+$ A5     & 3{,}031{,}040  & 0.000 & 2.140 & 1.000 & No \\
        \hline
    \end{tabular}
    \caption{\footnotesize Solvability tests of \citet{goff2025inferencevaluelinearprogram} (Section 3.1)
    for Assumptions 4 and 5 of \citet{ImaiJiangMalani2021}, using round 6
    ($G = 411$). Both restrictions are imposed on both spouses. ``Types'' counts
    augmented types $(\vec{s}, e)$, that is, columns of $A$, out of
    $16{,}777{,}216$. $\hat T_G$ is the test statistic defined in Section \ref{sec:solvability_tests}. Critical values and $p$-values are from a nonparametric
    bootstrap with $B = 1000$ draws and tuning parameter $\sqrt{\log G} \approx 2.45$;
    under tuning $\log G \approx 6.02$ the critical value is $2.391$ and the
    $p$-values are unchanged.}
    \label{tab:ijm_solvability}
\end{table}

Table
\ref{tab:ijm_bounds} recomputes the direct and indirect effects of Table
\ref{tab:bounds_ade_ase} under each specification. The identified sets are
numerically identical to those obtained without the restrictions, for all four
parameters and under either reading of \eqref{eq:ijm_a4}. This holds even for
the joint restriction under the best-response reading, which retains
$1{,}161{,}472$ of $16{,}777{,}216$ augmented types and so discards roughly
$93$ percent of the latent space.

\begin{table}[h]
\footnotesize
    \centering
    \begin{tabular}{lccc}
        \hline
        Parameter & Lower Bound & Upper Bound & 95\% CI \\
        \hline

        $\text{ADE}_{0,m}
        := \mathbb{E}[Y_{mg}(1,0) - Y_{mg}(0,0)]$
        & 0.07 & 0.38 & $[-0.15,\ 0.71]$ \\[0.2cm]

        $\text{AIE}_{0,m}
        := \mathbb{E}[Y_{mg}(0,1) - Y_{mg}(0,0)]$
        & -0.04 & 0.31 & $[-0.23,\ 0.58]$ \\[0.2cm]

        $\text{ADE}_{0,f}
        := \mathbb{E}[Y_{fg}(1,0) - Y_{fg}(0,0)]$
        & -0.07 & 0.28 & $[-0.29,\ 0.61]$ \\[0.2cm]

        $\text{AIE}_{0,f}
        := \mathbb{E}[Y_{fg}(0,1) - Y_{fg}(0,0)]$
        & -0.10 & 0.22 & $[-0.32,\ 0.54]$ \\

        \hline
    \end{tabular}
    \caption{\footnotesize Bounds and 95\% confidence intervals for direct and
    indirect effects under Assumptions 4 and 5 of \citet{ImaiJiangMalani2021},
    using round 6. The bounds are identical across all five specifications of
    Table \ref{tab:ijm_solvability} and identical to those in Table
    \ref{tab:bounds_ade_ase}, so a single set is reported. $D$ follows
    normalization $\widehat\Lambda$ as in Appendix \ref{AppInference}; $B = 250$; tuning parameter $\sqrt{\log G}$;
    grid step $0.02$. Confidence intervals are within one grid step of those in Table
    \ref{tab:bounds_ade_ase}.}
    \label{tab:ijm_bounds}
\end{table}

Our irrelevance results do not cover \eqref{eq:ijm_a4}: they apply to strong strategic neutrality and to restrictions implied by it, and own-offer monotonicity is not implied by dominance. That \eqref{eq:ijm_a4} nonetheless fails to bind, under either reading, is an empirical finding. For the equilibrium reading, it reflects the fact that the observable distribution already pins down realized take-up. Restriction \eqref{eq:ijm_a5}
does constrain potential outcomes, but only for noncompliers and only along the
partner's take-up channel, which the one-sided noncompliance documented in Table
\ref{tab:observable_distribution} leaves largely unexercised: no spouse without a
voucher opens an account, so for most offer contrasts there is no partner
response for the restriction to discipline. The conditions that secure
identification of complier effects in \citet{ImaiJiangMalani2021} are therefore
consistent with these data, but they do not narrow the identified sets for the
fixed-allocation parameters we consider.

\subsection{Irrelevance of SSN under Equilibrium Invariance}
\label{subsec:ssn_equilibriuminvariance}

Theorem \ref{thm:equivalence_dominance_multiagent_distr} only evaluates what happens to the identified set when simultaneously imposing SSN and filtering out certain other types, but not the more general set of linear restrictions of the form $R \mu \le r$. However, even such restrictions like monotone treatment selection (MTS), may also satisfy an additional property known as ``equilibrium-invariance'' under which SSN is also irrelevant.
    
        \begin{definition}[Equilibrium-invariant restrictions]
A set of restrictions $\{ \mu \in \Delta(\mathcal{T}^*): R\mu\le r\}$ is \emph{equilibrium-invariant} if each column of $R$
depends on $t=(\vec s,e)$ only through the potential-outcome schedules
$\{Y(\cdot;s_i)\}_{i=1}^N$ and the selected take-up $\vec D^{e}(\cdot)$.
\end{definition}
\begin{corollary}[Irrelevance of SSN under equilibrium-invariant restrictions]
\label{cor:ssn_invariant_restrictions}
Maintain the conditions of Theorem \ref{thm:equivalence_dominance_multiagent_distr}, and let
$\mathcal M=\{\mu\in\Delta(\mathcal T^*):R\mu\le r\}$ for an equilibrium-invariant restriction.
Then $\Theta(\mathcal M,p)=\Theta(\mathcal M_{dominant},p)$ for all $p$. In particular, this holds
for monotone treatment selection (Lemma \ref{lem:mts_equivalence}), for one or both spouses.
\end{corollary}    
\begin{proof}[Proof of Corollary \ref{cor:ssn_invariant_restrictions}]

Follow the proof of Theorem \ref{thm:equivalence_dominance_multiagent_distr}. In the presence of additional restrictions, we need to verify that the proposed push-forward measure $\mu^{\varphi}$ also satisfies $R \mu^{\varphi} \le r$ in addition to the other restrictions. Let $R(j,t)$ denote the corresponding entries of the matrix for restriction $j$ and type $t$. When the restriction is equilibrium invariant, then $R(j,\varphi(t)) = R(j,t)$. Then
\begin{align*}
   r_j &\ge \sum_{t \in \mathcal{T}^*} R(j,t) \mu(t) = \sum_{t \in \mathcal{T}^*} R(j,\varphi(t)) \mu(t) \\
   &= \sum_{t'\in \varphi(\mathcal{T}^*)}\sum_{t: \varphi(t) = t'} R(j,\varphi(t)) \mu(t)  = \sum_{t'\in \varphi(\mathcal{T}^*)} R(j,t') \mu^{\varphi}(t').
\end{align*}
Next we prove that monotone treatment selection is equilibrium-invariant by focusing on the form of Lemma \ref{lem:mts_equivalence}. That lemma focuses on a particular ordering of the types $t_\ell$, for $\ell \in \{1,\ldots,|\mathcal{T}^*|\}$ (with $\mathcal{T}^* \subseteq \mathcal{T}$ in this case). For each restriction $j$, we assign values $(a_j,b_j)$. To match the notation, $R(j,t_\ell) = R^{a_j,b_j}_\ell$. The entries of $R$ depend on $t_\ell$ only
through $Y(a,b;s_{m\ell})$ and $D^{e}_m(\vec z)$, and $\kappa_m$ is determined by $p$ and
$\mathbb P(\vec Z_g=\vec z)$, so the restriction is equilibrium-invariant.
\end{proof}

\subsection{Irrelevance of Sequential Play and Network Dominance}
\label{sec:irrelevance_sequentialplay}

    \textbf{Encoding sequential play:} Our framework is flexible enough to accommodate restrictions on the timing in which players decide to take action. For example, Figure \ref{fig:networkrepresentation_types}.(a) shows a decision tree for a sequential game in which each player decides to take-up after the other. This can be encoded in a sequence of best-response functions where for agent 1 $D(z,z',d';s_i)$ only depends on $(z_1)$, for agent 2 it depends on $(z_2,z_1,d_1)$, and for agent 3 it depends on all information $(z_3,(z_1,z_2),(d_1,d_2))$. Figure \ref{fig:networkrepresentation_types}.(b) shows this as a dependency graph. Within the language of our framework, this can be encoded as restrictions over subsets of $\mathcal{T}$.

\begin{figure}[t]
	\centering
	\includegraphics[scale=0.7]{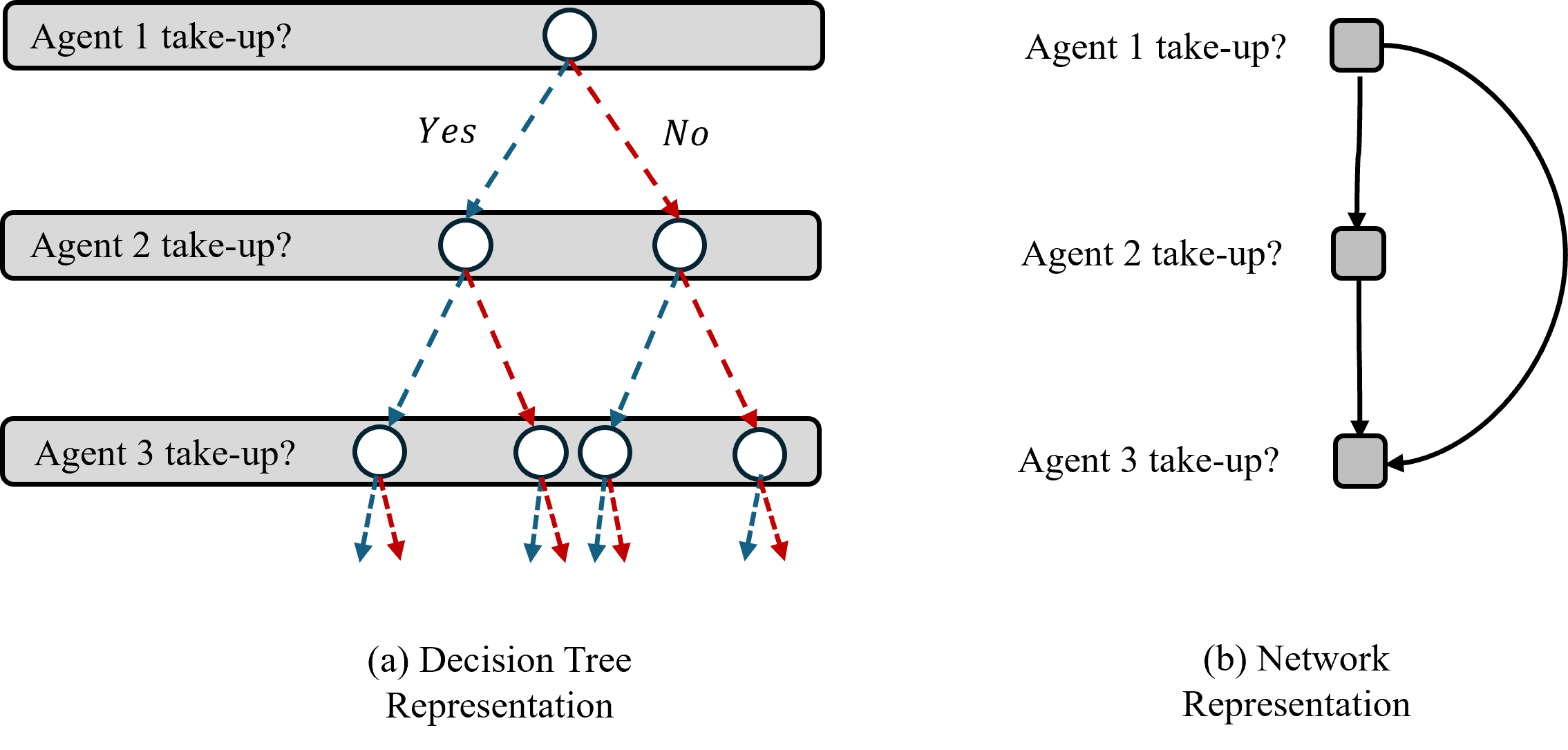}	
	\caption{\footnotesize A dynamic game where agents make a decision about whether to enroll sequentially. }
	\label{fig:networkrepresentation_types}		
\end{figure}

\textbf{Encoding network games:} The idea of a dependency graph is not only suitable for representing a sequential game, but can also be used for a more general class of network games, whenever network data is also available. An additional implication of the non-refutability of dominance is that, given an arbitrary observed distribution, we can always rationalize an arbitrary order of sequential play in the game. We can think of a particular sequential timing as a weak order\footnote{A weak order (a complete preorder) is a binary relation which is reflexive, transitive, and complete.} on the set of players $\{1, \cdots, N\}$ so that $i \succsim_O j$ if $i$ acts no later than $j$ in the sense of Figure \ref{fig:networkrepresentation_types}). More precisely:
	\begin{definition}
		Given a weak order $O$ on $\{1, \cdots n\}$, for any $i$ define $i_O^+= \{j: i \succ_O j\}$. Let: 
		\begin{align}
			\mathcal{M}_O := &\Bigg\{\mu \in \Delta(\mathcal{T}) \ \textrm{such that if} \ \mu(t) > 0, \ \textrm{for each player} \ i \ \textrm{and any} \ j \in i_O^+ \ \textrm{we have} \nonumber \\ &\  \ D(z,\vec{z}',\vec{d}';s_i) \ \textrm{does not depend on action} \ d_j \Bigg\}
		\end{align}
		That is, $D(z,\vec{z}',\vec{d}';s_i)$ cannot depend on $d_j$ for any $j$ that acts strictly after $i$.
	\end{definition}
	
	Because strategic orderings constrain only the potential take-up, leaving potential outcomes unrestricted, the typical $\mathcal{M}_O$ has the structure of a Cartesian product of the restricted take-up types with all potential outcome types. This is a polytope. Below, we characterize the set of dominant strategies as the ones which can rationalize any time ordering.

	\begin{theorem}[Dominant BRFs are BRFs consistent with all time orders] \text{ }\\
		\label{prop:dominance_characterization}
		Let $\mathcal{T}$ be the set of all types that satisfy Assumptions \ref{AsExclu}-\ref{assump:deterministic_equilibria} and let $\mathcal{M}_O \subseteq \Delta(\mathcal{T})$ be the set of distributions supported over pure-strategy BRFs consistent with ordering $O$ over players, and let $\mathcal{M}_{\mathrm{dom}} \subset \Delta(\mathcal{T})$ denote the set of profiles where each player's best response function is constant (i.e., corresponds to a dominant strategy). Then $\mathcal{M}_{\mathrm{dom}} = \bigcap_O \mathcal{M}_O.$
	\end{theorem}

    \begin{proof}[Proof of Theorem \ref{prop:dominance_characterization}] \text{ } \\
		\ We prove the two inclusions. To show the $(\subseteq)$ inclusion let $\mu \in \mathcal{M}_{\mathrm{dom}}$. Then $\mu(t)>0$ implies that $s_i$ satisfies $D(z,\vec{z'},\vec{d'};s_i) = D(z,\vec{z'},s_i)$ for all players $i$, for all $z$: the BRF is constant in $d_j$ for all $j \neq i \in \{1, \cdots, n \}$. In particular, for any timing $O$, $D(z,\vec{z'},\vec{d'};s_i) = D(z,\vec{z'},\vec{d'}_{-j},s_i)$ for all $j \in i_O^+$: the BRF is constant on the subset of indices $i_O^+$. Hence $\mu \in \mathcal{M}_O$. Because $O$ was arbitrary, $\mu \in \bigcap_O \mathcal{M}_O$ which implies $\mathcal{M}_{\mathrm{dom}} \subseteq \bigcap_O \mathcal{M}_O$. To show the $(\supseteq)$ inclusion suppose $\mu \not \in \mathcal{M}_{\mathrm{dom}}$. Then $\exists t$ with $\mu(t) > 0$ such that for some player $i$, $s_i$ is not constant: there exist index $j\neq i \in \{1, \cdots, N \}$ and vectors $d',d"$ with $j$-th coordinate $d'_{j} \neq d''_{j}$ such that $D(z,\vec{z'},\vec{d'};s_i) \ne D(z,\vec{z'},\vec{d"},s_i)$. Choose a strict timing $O$ in which player $i$ acts first. Then
        $i_O^+=\{1,\ldots,N\}\setminus\{i\}$. Since $s_i$ is not constant in the
        actions of subsequent players, $\mu\notin\mathcal{M}_O$, and therefore $\mu\notin\bigcap_O\mathcal{M}_O$. So $\bigcap_O \mathcal{M}_{O} \subseteq \mathcal{M}_{\mathrm{dom}}$. Combining the two inclusions, we conclude that $\mathcal{M}_{\mathrm{dom}} = \bigcap_O \mathcal{M}_O$.
	\end{proof} 

\subsection{Irrelevance and Non-refutability of Super/sub-modularity}
\label{sec:supermodularity}
    
 In a supermodular game, agents' payoffs from engaging in an action are weakly increasing in others' actions \citep{graham2023scenario,molinari2008identification,Lazzati2015,HanBalat2023}. In a submodular games the payoffs are weakly decreasing. These restrictions appear in empirical work on a variety of topics, including peer effects in the decision to smoke \citep{card2013peer,krauth2005peer,gaviria2001school}, technology adoption and network externalities \citep{goolsbee2002evidence,ackerberg2006quantifying}, and market entry \citep{ciliberto2009market,jia2008happens,bresnahan1991empirical}.
	Super/sub-modularity implies that the best-response take-up function is monotonic in others' take-up decisions. For example, for super-modularity: 
	\begin{align*}
		\vec{d'} &\le \vec{d'}_{*} \quad \implies \quad D(z,\vec{z'},\vec{d'};s_i) \le D(z,\vec{z'},\vec{d'}_{*};s_i) \qquad \text{(Supermodularity)}
	\end{align*}
	Constant best-response functions—which imply dominance—arise when super-modularity and sub-modularity are imposed simultaneously. Combining this observation with Theorem \ref{thm:equivalence_dominance_multiagent_distr}, we conclude that assuming either increasing or decreasing best-response behavior (both strictly weaker than dominance) has no additional identifying power for fixed-allocation parameters. Under the conditions of Theorem \ref{thm:equivalence_dominance_multiagent_distr}, 
	let $\mathcal{T}_{supermodular}^*$ and $\mathcal{T}_{submodular}^*$ denote the subsets where all individuals have either supermodular or submodular best-responses. Then $ \Theta(\Delta(\mathcal{T}_{supermodular}^*), p) = \Theta(\Delta(\mathcal{T}_{submodular}^*), p) =   \Theta(\Delta(\mathcal{T}_{dominant}\cap \mathcal{T}^*), p)$.

\subsection{Deterministic vs Stochastic Equilibrium Selection}\label{sec:det_vs_stoc}

In many situations, it is natural for multiple equilibria to emerge. In the main text, we assumed that households have a means of resolving this multiplicity by selecting a particular equilibrium. However, Assumption \ref{assump:deterministic_equilibria} in the main text rules out cases where indecisive households may decide to randomize between the different equilibria. Following the example in Section \ref{sec:ident_estim_infer}, these ``stochastic equilibria selection'' rules, place probabilities (adding up to one) on each Nash equilibrium.

	\noindent
	\begin{minipage}[t]{0.49\linewidth}
		\vspace{0pt}
		{\setlength{\mathindent}{0pt}\setlength{\abovedisplayskip}{0pt}\setlength{\abovedisplayshortskip}{0pt}\setlength{\belowdisplayskip}{2pt}%
			\[
			e(\vec z,\vec d)=
			\begin{cases}
				0.5,& d_1=0,\, d_2=0,\\
				0,  & d_1=0,\, d_2=1,\\
				0,  & d_1=1,\, d_2=0,\\
				0.5,& d_1=1,\, d_2=1.
			\end{cases}
			\]
		}
	\end{minipage}%
	\hfill%
	\begin{minipage}[t]{0.49\linewidth}
		\vspace{0pt}
		{\setlength{\mathindent}{0pt}\setlength{\abovedisplayskip}{0pt}\setlength{\abovedisplayshortskip}{0pt}\setlength{\belowdisplayskip}{2pt}%
			\[
			e^*(\vec z,\vec d)=
			\begin{cases}
				0,& d_1=0,\, d_2=0,\\
				0,& d_1=0,\, d_2=1,\\
				0,& d_1=1,\, d_2=0,\\
				1,& d_1=1,\, d_2=1.
			\end{cases}
			\]
		}
	\end{minipage}
 \text{ }  \\
 
Our goal will be to understand when modeling the types with either stochastic or deterministic equilibria has no impact on the identified set. The core intuition is that any household randomizing between two choices can simply be viewed as two separate households, each resolving the multiplicity one way or another based on an external signal. Both produce equivalent summary statistics of behavior, in the aggregate. Our finding is empirically important because it shows that we do not need to explore the infinite space of stochastic rules, and can, in many cases, adopt Assumption \ref{assump:deterministic_equilibria} as in the main paper, which has a tractable formulation as a finite linear program and more manageable computationally.

	\begin{assumption}[Closure over deterministic rules] \label{assump:closure_deterministic} \text{ } \\
		The sets $(\mathcal{Y},\mathcal{Z},\mathcal{D})$ are finite. Let $\mathcal{T}$ be the set of all types that satisfy Assumptions \ref{AsExclu}-\ref{AsNash}. We say that a subset $\mathcal{T}^* \subseteq \mathcal{T}$ is closed over deterministic rules if every type  $(\vec{s},e) \in \mathcal{T}^*$  can be coupled with another type $(\vec{s},e^*) \in \mathcal{T}^*$ that satisfies, 
            \begin{equation}
        e^*(\vec{z},\vec{d}; \vec{z}_*,\vec{d}_*,e) = \begin{cases} 1 & \text{ if }\quad \vec{z} = \vec{z}_*, \ \vec{d} = \vec{d}_*, \ e(\vec{z}_*,\vec{d}_*) > 0 \\
        0 & \text{ if }\quad  \vec{z} = \vec{z}_*, \ \vec{d} \ne \vec{d}_*, \ e(\vec{z}_*,\vec{d}_*) > 0  \\
        e(\vec{z},\vec{d}) & \text{otherwise},
        \end{cases}
        \label{eq:companion_rule_assump}
    \end{equation}
    for any choice of $(\vec{z}_*,\vec{d}_*) \in \mathcal{Z}^N \times \mathcal{D}^N$.
	\end{assumption}
    	\begin{theorem}[Invariance under stochastic rules] \text{ }\\
		\label{thm:deterministic_multiagent} Consider an estimand of the form $\theta(\mu) = \int_{\mathcal{T}} h(\vec{s}) \text{d}\mu(\vec{s},e)$ for some real-valued function $h$ and let $\Theta(\mathcal{M},p)$ be its identified set under model $\mathcal{M}$ and observed probability vector $p$. Let $\mathcal{T}$ be the set of all types that satisfy Assumptions \ref{AsExclu}-\ref{AsNash}, and let $\mathcal{T}^* \subseteq \mathcal{T}$ be any set of types that satisfies Assumption  \ref{assump:closure_deterministic}. Let $\mathcal{T}_{det}^* \subseteq \mathcal{T}^*$ be the subset of types for which $e(\vec{z},\vec{d}) \in \{0,1\}$ for all $(\vec{z},\vec{d})$. Then $\Theta(\Delta(\mathcal{T}_{det}^*),p) = \Theta(\Delta(\mathcal{T}^*),p)$ for all $p$.
	\end{theorem}
    
    The closure condition allows us to find a matching household type $(\vec{s},e^*)$ that is identical to $(\vec{s},e)$ in all regards, but that for a specific $\vec{z} = \vec{z}_*$, chooses an equilibrium deterministically rather than via a stochastic rule.\footnote{The superset $\mathcal{T}$ is trivially closed over deterministic types, since it contains all equilibrium selection rules, and so is the empty set. The closure condition can also be verified on a case-by-case basis once we impose additional restrictions. Assumption \ref{assump:closure_deterministic} states that if $\mathcal{T}^*$ includes rules supported over multiple Nash equilibria, it must also include rules supported over each of these equilibria separately, with otherwise equal features of the game type.} Theorem \ref{thm:deterministic_multiagent} shows that for any set of game types $\mathcal{T}^*$ that satisfies Assumption \ref{assump:closure_deterministic}, there is a subset of types with simpler equilibrium selection rules that rationalizes the same identified set for $\theta$. All of the estimands in Section \ref{subsec:parameterofinterest} satisfy this form, not just fixed-allocation estimands as in Theorem \ref{thm:equivalence_dominance_multiagent_distr}.

    	\begin{proof}[Proof of Theorem \ref{thm:deterministic_multiagent}]
	\ 
	By the definition of the identified set and the fact that $\mathcal{T}_{det}^* \subseteq \mathcal{T}^* \implies \Delta(\mathcal{T}^*_{det}) \subseteq \Delta(\mathcal{T}^*)$, we have $\Theta(\Delta(\mathcal{T}^*_{det})) \subseteq \Theta(\Delta(\mathcal{T}^*))$.
    
    It remains to show the reverse direction. 
	Let $\theta \in \Theta(\mathcal{T}^*)$ and let $\mu$ be a distribution achieving $\theta(\mu) = \theta$. Our general strategy can be summarized as follows: we iteratively loop over all $\vec{z}_* \in \mathcal{Z}^N$, create a new measure $\widetilde{\mu}$ that at every iteration (i) only puts weight on deterministic equilibria for that profile of offers, (ii) satisfies the feasibility constraints if and only if the original measure does, (iii) preserves the measures of $\vec{s}$ types thus $\theta(\mu) = \theta(\widetilde{\mu})$, and (iv) is supported over types that are a subset of $\mathcal{T}$. Since $\mathcal{Z}$ is finite, the process ends after a finite number of steps. The proof concludes by showing that the final set of supported types is a subset of those with deterministic equilibria. Therefore the final measure is feasible and can replicate each point in the identified set.
    
    Start by fixing an arbitrary offer profile $\vec{z}_* \in \mathcal{Z}^N$ and take-up profile $\vec{d}_* \in \mathcal{D}^N$. Our goal will be to pair each $(\vec{s},e) \in \mathcal{T}^*$ with a $(\vec{z}_*,\vec{d}_*)$-specific companion type $(\vec{s},e^*) \in \mathcal{T}^*$, whose properties we will describe shortly. To differentiate the fixed $(\vec{z}_*,\vec{d}_*)$-pair used to select the companion, from the counterfactual allocation induced by $e^*$, we write $e^*(\vec{z},\vec{d}) = e^*(\vec{z},\vec{d}; \vec{z}_*,\vec{d}_*,e)$, which is defined for any $(\vec{z},\vec{d})$ in the support. By Assumption \ref{assump:closure_deterministic}, we can always choose a companion such that 
    \begin{equation}
        e^*(\vec{z},\vec{d}; \vec{z}_*,\vec{d}_*,e) = \begin{cases} 1 & \text{ if }\quad \vec{z} = \vec{z}_*, \ \vec{d} = \vec{d}_*, \ e(\vec{z}_*,\vec{d}_*) > 0 \\
        0 & \text{ if }\quad  \vec{z} = \vec{z}_*, \ \vec{d} \ne \vec{d}_*, \ e(\vec{z}_*,\vec{d}_*) > 0  \\
        e(\vec{z},\vec{d}) & \text{otherwise}.
        \end{cases}
        \label{eq:companion_rule}
    \end{equation}
    The companion is chosen in such a way that it chooses a deterministic equilibrium under offer $\vec{z}_*$ (when it has positive support under the original $e$) and preserves the equilibrium selection rule for all other $\vec{z} \ne \vec{z}_*$. This is a valid equilibrium selection rule because $\sum_{\vec{d} \in \mathcal{D}^N} e^*(\vec{z},\vec{d}; \vec{z}_*,\vec{d}_*,e) = 1$ regardless of the inputs, and inherits the Nash-support property (Assumption \ref{AsNash}) from $e$. If we view the construction of the companion as a function $\varphi_{(\vec{z}_*,\vec{d}_*)} $, such that $(\vec{s},e^*) = \varphi_{(\vec{z}_*,\vec{d}_*)}(\vec{s},\vec{e})$, then we have effectively proven that $\varphi_{(\vec{z}_*,\vec{d}_*)}(\mathcal{T}^*) \subseteq \mathcal{T}^*$ and its union $\bigcup_{\vec{d}_* \in \mathcal{D}^N} \varphi_{(\vec{z}_*,\vec{d}_*)}(\mathcal{T}^*) \subseteq \mathcal{T}^*$.

    Our goal will be to construct a new measure supported on $\bigcup_{\vec{d}_* \in \mathcal{D}^N} \varphi_{(\vec{z}_*,\vec{d}_*)}(\mathcal{T}^*)$.  Suppose that conditional on $(\vec{s},e)$ we select each $\varphi_{(\vec{z}_*,\vec{d}_*)}(\vec{s},\vec{e})$ at random with probability $e(\vec{z}_*,\vec{d}_*)$. This can be interpreted as the share of households that resolve the equilibrium selection problem one way or another using an external signal. Then the resulting (push-forward) probability $\widetilde{\mu}$ takes the following form: for any $M \subseteq \mathcal{T}$,
    \begin{equation}
        \widetilde{\mu}(M) = \int_{\mathcal{T}} \sum_{\vec{d}_* \in \mathcal{D}^N} \mathbbm{1}\left\{ \varphi_{(\vec{z}_*,\vec{d}_*)}(\vec{s},\vec{e})  \in M \right\} e(\vec{z}_*,\vec{d}_*) \ \text{d}\mu(\vec{s},e),
        \label{eq:pushforward_measure}
    \end{equation}
    where $\mu$ is the original probability measure. The inner term is the conditional equilibrium selection term, which is equal to one when $M = \mathcal{T}^*$.

    Next, we show that if the original measure satisfies the feasibility constraint, then so does the new measure. Let $\vec{Y}(\vec{d};\vec{s})$ denote the vector of potential outcomes from allocation $\vec{d}$ and types $\vec{s}$. By the law of iterated expectation, then 
    \begin{align*}
        &\int_{\mathcal{T}} \sum_{\vec{d} \in \mathcal{D}^N} \mathbbm{1}\{ (\vec{Y}(\vec{d};\vec{s})) = y \}e^*(\vec{z},\vec{d}) \text{d}\mu(\vec{s},e^*) \\
        &= \int_{\mathcal{T}} \sum_{\vec{d}_* \in \mathcal{D}^N} \left[ \sum_{\vec{d} \in \mathcal{D}^N} \mathbbm{1}\{ (\vec{Y}(\vec{d};\vec{s})) = y \}e^*(\vec{z},\vec{d};\vec{z}_*,\vec{d}_*,e) \right] e(\vec{z}_*,\vec{d}_*) \text{d}\mu(\vec{s},e) \\
        &=  \int_{\mathcal{T}} \sum_{\vec{d} \in \mathcal{D}^N} \mathbbm{1}\{ (\vec{Y}(\vec{d};\vec{s})) = y \} \underbrace{\left[ \sum_{\vec{d}_* \in \mathcal{D}^N} e^*(\vec{z},\vec{d};\vec{z}_*,\vec{d}_*,e)  e(\vec{z}_*,\vec{d}_*)\right]}_{ = e(\vec{z},\vec{d}) } \text{d}\mu(\vec{s},e),
    \end{align*}
    where the identity in the second equality follows directly from the definition in \eqref{eq:companion_rule}. The last line exactly coincides with $\mathbb{P}((\vec{Y}_g,\vec{D}_g) = (\vec{y},\vec{d}) \mid \vec{Z}_g = \vec{z})$. Therefore the feasibility constraint is satisfied by the new measure if and only if the original measure is feasible.

    The new measure preserves the marginal of $\vec{S}$. Since the estimand only depends on $\vec{S}$ by assumption, then $\theta(\widetilde{\mu}) = \theta(\mu)$. We can repeat the above process iteratively, until we've exhausted all $\vec{z} \in \mathcal{Z}^N$. To conclude, the mapping $\varphi_{(\vec{z}_*,\vec{d}_*)}$ iteratively selects deterministic equilibria, so that by the last iteration only deterministic equilibria $e^*$ survive. This means that the measure $\widetilde{\mu}$ obtained in the last iteration is only supported over deterministic equilibria. This concludes the proof.
    \end{proof}

     \subsection{On the notions of Symmetry, Anonymity, and  Refutability}

    \label{sec:symmetry_refutability}

    Our type space allows agents in the same group to have distinct types, i.e. $s_i$ may differ across $i$. This generality reflects the need to flexibly  model the heterogeneous benefits and costs of take-up, and the game structure that disciplines it. This flexibility is particularly important in industrial organization contexts. Even models of firm entry that assume ex-ante identical payoff functions \citep{ciliberto2009market,berry1992estimation,mazzeo2002product}, typically allow agents to be heterogeneous based on individual characteristics or shocks. %

To compare the two individuals, write each best response with the offers in a fixed order, individual 1's offer first: $D_1(z_1,z_2,d') := D(z_1,z_2,d';s_1)$ and $D_2(z_1,z_2,d') := D(z_2,z_1,d';s_2)$, so that $D_i(z_1,z_2,d')$ is individual $i$'s take-up when the offers are $(z_1,z_2)$ and the other individual chooses $d'$. We consider two notions of symmetry.
\begin{align*}
& D_1(z_1,z_2,d') = D_2(z_1,z_2,d') && \text{(Individual symmetry)}\\
& D_1(z_1,z_2,d') = D_2(z_2,z_1,d') && \text{(Equivariant symmetry)},
\end{align*}
for all $(z_1,z_2,d')$. Individual symmetry says that, in the same game, both individuals respond in the same way to the other's take-up.\footnote{Although they may have the same response, in equilibrium they may take different actions. For example $D_i(z_1,z_2,d') = 1-d'$ satisfies both individual and equivariant symmetry but $(d_1,d_2) = (1,0)$ is a valid Nash equilibrium.} Equivariant symmetry says that if we swap the individuals and their offers, responses are unchanged; equivariance emerges when individuals optimize the same utility function of own and other's offer and take-up. %

Table \ref{tab:stategy_equilibria_table} illustrates the types and associated PSNE (for a given instrumental assignment) that satisfy, respectively, individual symmetry and dominance. Under individual symmetry, we focus on type pairs on the diagonal, whereas under dominance, we focus on a subset of rows/columns. Both restrictions narrow the space of admissible types in a similar way, while remaining sufficiently general to rationalize any observed equilibrium $(d_1,d_2)$. Imposing dominance and symmetry rules out certain observed equilibria, so that the combination of these two restrictions may be refutable.
	
	\begin{table}[t]
		\centering

		\footnotesize

		\begin{subtable}{0.45\textwidth}
			\captionsetup{labelformat=empty}
			\begin{tabular}{|c|c|c|c|c|} 
				\hline 
				& \multicolumn{4}{c|}{$\delta_s = (D(z,z',0;s), D(z,z',1;s))$} \\
				\hline
				$\delta_{s'} $ & \textbf{0,0} & \textbf{0,1} & \textbf{1,0} & \textbf{1,1} \\
				\hline 
				\textbf{0,0} & \cellcolor{lightgreen}(0,0) & (0,0) & (0,1) & (0,1) \\
				\hline 
				\textbf{0,1} & (0,0) & \cellcolor{lightgreen}(0,0); (1,1) & $\varnothing$ & (1,1) \\
				\hline 
				\textbf{1,0} & (1,0) & $\varnothing$ & \cellcolor{lightgreen}(1,0); (0,1) & (0,1) \\
				\hline
				\textbf{1,1} & (1,0) & (1,1) & (1,0) & \cellcolor{lightgreen}(1,1) \\
				\hline
			\end{tabular}
			\caption{Panel A: Symmetry}
		\end{subtable}
		\hspace{2em}
		\begin{subtable}{0.45\textwidth}
			\begin{tabular}{|c|c|c|c|c|}
				\hline 
				& \multicolumn{4}{c|}{$\delta_s = (D(z,z',0;s), D(z,z',1;s))$} \\
				\hline
				$\delta_{s'} $ & \textbf{0,0} & \textbf{0,1} & \textbf{1,0} & \textbf{1,1} \\
				\hline 
				\textbf{0,0} & \cellcolor{lightpurple}(0,0) & (0,0) & (0,1) & \cellcolor{lightpurple}(0,1) \\
				\hline 
				\textbf{0,1} & (0,0) & (0,0); (1,1) & $\varnothing$ & (1,1) \\
				\hline 
				\textbf{1,0} & (1,0) & $\varnothing$ & (1,0); (0,1) & (0,1) \\
				\hline
				\textbf{1,1} & \cellcolor{lightpurple}(1,0) & (1,1) & (1,0) & \cellcolor{lightpurple}(1,1) \\
				\hline
			\end{tabular}
			\caption{Panel B: Dominance}
		\end{subtable}
		
		\caption{\footnotesize Take-up Equilibria. Cells represent the set of possible PSNE induced by the joint BRFs, given a fixed set of offers, e.g. $(z,z') = (0,0)$. For example, the cell in the first column/second row of panel B, represents the set of Nash equilibria for a type pair $(s,s')$ featuring $\delta_s = (D(z,z',0;s),D(z,z',1;s)) = (0,0)$ and $\delta_{s'} = (D(z',z,0;s'),D(z',z,1;s')) = (0,1)$. Rows and columns represent the take-up type of each spouse, respectively. Panels (A) and (B) show the restrictions on equilibria imposed by individual symmetry and dominance. The pairs of BRFs and set of equilibria could vary with $(z,z')$, so there could be up to four different versions of these tables corresponding to each offer pair.}
		\label{tab:stategy_equilibria_table}
	\end{table}

These notions extend to $N > 2$ but require more notation. Let $D_i(\vec{z},d') = D(z_i,\vec{z}_{-i},d';s_i)$ for $d' \in \mathcal{D}^{N-1}$. We say that $\pi$ is a permutation map if $\pi:\{1,\ldots,N-1\} \to \{1,\ldots,N-1\}$ and it is bijective. For vectors $d' \in \mathcal{D}^{N-1}$ and $z'\in \mathcal{Z}^{N-1}$, let $(d'_{\pi},z'_{\pi})$ denote their respective reorderings, e.g. $d'_{\pi} = (d_{\pi(1)}',\ldots,d_{\pi(N-1)}')$ and $z_{\pi}' = (z_{\pi(1)}',\ldots,z_{\pi(N-1)}')$. We generalize the preceding notions of symmetry as follows.
    \begin{definition}[Individual Symmetry]
    \label{Def:Individual symmetry}
    We say that $\vec{s}$ has individually symmetric best-responses if for every distinct pair $i,j \in \{1,\ldots,N\}$, $i \ne j$, and $\vec{z} \in \mathcal{Z}^N$, there exists a permutation map $\pi_{ij}(\vec{z})$, such that $\forall d' \in \mathcal{D}^{N-1}$, $D_{i}(\vec{z},d')=D_j(\vec{z},d_{\pi_{ij}(\vec{z})}')$.
    \end{definition}
    \begin{definition}[Equivariant symmetry]
    \label{Def:Equivariant symmetry}
    We say that $\vec{s}$ has equivariantly symmetric best-responses if for every distinct pair $i,j \in \{1,\ldots,N\}$, $i \ne j$, there exists a permutation map $\pi_{ij}$, such that $\forall (z,z',d') \in \mathcal{Z}\times\mathcal{Z}^{N-1} \times \mathcal{D}^{N-1}$, $D(z,z',d';s_i)=D(z,z_{\pi_{ij}}',d_{\pi_{ij}}';s_j)$.\footnote{This coincides with the bivariate notion since  $D(z_1,z_2,d';s_1) = D_1(z_1,z_2,d') = D_2(z_2,z_1,d') = D(z_1,z_2,d';s_2) $.}
    \end{definition} 

     Finally, note that both of the notions of symmetry are distinct from take-up anonymity.
    \begin{definition}[Take-up Anonymity] 
    We say that $\vec{s}$ is take-up anonymous if for all $i \in \{1,\ldots,N\}$, $(z,z',d') \in \mathcal{Z}\times\mathcal{Z}^{N-1} \times \mathcal{D}^{N-1}$, and all permutation maps $\pi$, $D(z,z',d';s_i) = D(z,z',d_{\pi}';s_i)$.
    \end{definition} 
    Take-up anonymity requires that each individual's best response depends on the other individuals' choices only through a symmetric function of $d'$. For example when the action space is $\mathcal{D}=\{0,1\}$, some natural choices for take-up anonymity may be the total sum of treated $\sum_{j\neq i} d_j$ (which is also the maximal invariant under anonymity) or whether anyone at all has been treated $\max_{j \in [N] \setminus \{i\}} d_j$. Anonymity does not impose any restriction when $N=2$, because there is only one other agent and the only permutation is the identity.\footnote{To understand how types grow exponentially with the number of players (and how anonymity affects this scale), with $N$ players there are $|\mathcal{D}|^{N-1}$ values for the argument of $D_{i}(\vec z,\cdot)$ leading to $|\mathcal{D}|^{|\mathcal{D}|^{N-1}}$ different potential functions $D_{i}(\vec z,\cdot)$, and, with $N$ players, $|\mathcal{D}|^{N \cdot |\mathcal{D}|^{N-1}}$ distinct combinations of BRFs overall. In comparison, under take-up anonymity, the number of distinct $D_{i}(\vec z,\cdot)$ functions for a single player is $|\mathcal{D}|^{\binom{N+|\mathcal{D}|-2}{|\mathcal{D}|-1}}$ and, with $N$ players, $|\mathcal{D}|^{N \cdot \binom{N+|\mathcal{D}|-2}{|\mathcal{D}|-1}}$ distinct combinations of BRFs overall. For example, with $\mathcal{D} =\{0,1\}$ and 3 players we have 4096 distinct joint BRFs without take-up anonymity, against only 512 under anonymity.}

  Building on the insight of Table \ref{tab:stategy_equilibria_table}, the following theorem establishes that (i) individual symmetry is not refutable, and (ii) it is irrelevant for fixed-allocation parameters. For ease of exposition, we focus on combining symmetry with the equilibrium-invariant restrictions from Appendix \ref{subsec:ssn_equilibriuminvariance}, which hold trivially when $(R,r) = (\mathbf{0}_{1 \times |\mathcal{T}|},0)$ (the unrestricted model).

       \begin{theorem}
    \label{thm:symmetry}
    Let $\mathcal{T}$ be the set of all types satisfying Assumptions \ref{AsExclu}--\ref{assump:deterministic_equilibria}, and let $\mathcal{T}^* \subseteq \mathcal{T}$ denote the subset of types with individually symmetric best-response functions. (i)  If $p$ can be rationalized by some $\mu \in \Delta(\mathcal{T})$, then there exists $\mu^* \in \Delta(\mathcal{T}^*)$ that rationalizes $p$. (ii) Moreover, consider estimands $\theta=\mathbb E[\theta(T_g)]$ where $\theta(t)$ depends on $t$ only through the potential-outcome profile $\{Y(\cdot;s_i)\}_{i=1}^N$, and let $\mathcal M=\{\mu\in\Delta(\mathcal T):R\mu\le r\}$ for an equilibrium-invariant restriction. Then $\Theta(\mathcal M,p)=\Theta(\mathcal M\cap\Delta(\mathcal T^*),p)$ for all $p$.
    \end{theorem}

    \begin{proof}[Proof of Theorem \ref{thm:symmetry}]
    \textit{Part (i).} We will show that there exists a distribution entirely supported on types satisfying individual symmetry.  By assumption, there exists a distribution $\mu$ in $\Delta(\mathcal{T})$ that is compatible with $p$. Now, for any type $t = (\vec{s},e)$ with $\mu(t) > 0$ we can construct a new type $t^* := (\vec{s}^*, e^*) = (s_1^*, s_2^*, \cdots, s_N^*, e^*)$ satisfying individual symmetry. Consider the reduced-form take-up decisions induced by $e$, $\vec{D}^{e}(\vec z) = (d_1(\vec{z}), d_2(\vec{z}), \cdots, d_N(\vec{z}))$. To construct symmetric types $s_i^*$ with best-responses $D_i^*(\cdot)$, assign to each agent the BRF:
    \begin{equation}
    D_i^*(\vec z,d')
    =
    \begin{cases}
    d_j(\vec z),
    & \text{if } d'=d_{-j}(\vec z)
    \text{ for some } j\in\{1,\ldots,N\},\\
    d_0,
    & \text{otherwise,}
    \end{cases}
    \end{equation}
    for any fixed $d_0 \in \mathcal{D}$. This definition is single-valued. If, for some $j\neq k$,
    $d_{-j}(\vec z)=d_{-k}(\vec z)$, then, since the two vectors differ only
    in which component is omitted, equality of the vectors implies
    $d_j(\vec z)=d_k(\vec z)$. Hence, whenever more than one $j$ satisfies
    the first case in the definition of $D_i^*(\vec{z},d')$, all such $j$ prescribe the
    same value. Thus, $D_i^*(\vec{z},d')$ is a well-defined common BRF (it does not depend on $i$ on the right-hand side) and therefore satisfies
    individual symmetry with $\pi_{ij}(\vec{z})$ equal to the identity. For each $\vec z$, the selected equilibrium imposes restrictions on the common best-response function only at the $N$ profiles $d_{-i}(\vec z)$, while its values at all remaining action profiles are unrestricted and can therefore be chosen arbitrarily. By construction, $D_i^*(\vec z,d_{-i}(\vec z)) = D(z_i,\vec{z}_{-i},d_{-i}(\vec z);s_i^*)=d_i(\vec z)$ for every $i\in\{1,\ldots,N\}$,
    so $\vec{D}^{e}(\vec z)$ remains a Nash equilibrium under $\vec{s}^{\,*}$. For potential outcomes, assign $Y(\vec d;s_i^*) = Y(\vec d;s_i)$. Because this construction may extend the equilibrium set, choose $e^* = e$. Then $t^*$ is a type that satisfies individual symmetry. Applying this construction to every type in the support of $\mu$, let
    $\mu^*$ be the pushforward of $\mu$ under the mapping $t\mapsto t^*$.
    Since $\vec{D}^{e^*}(\vec z)=\vec{D}^{e}(\vec z)$ and $Y(\vec d; s_i^*)=Y(\vec d;s_i)$, the
    distribution $\mu^* \in \Delta(\mathcal{T}^*)$ generates the same observable distribution $p$ as
    $\mu$.

\textit{Part (ii).} Since $\mathcal M\cap\Delta(\mathcal T^*)\subseteq\mathcal M$, we have $\Theta(\mathcal M\cap\Delta(\mathcal T^*),p)\subseteq\Theta(\mathcal M,p)$. Conversely, let $\theta\in\Theta(\mathcal M,p)$ be attained by some $\mu\in\mathcal M$ that rationalizes $p$, and let $\mu^* \in \Delta(\mathcal{T}^*)$ be the pushforward of $\mu$ under $t\mapsto t^*$ constructed in part (i). Write the $j$-th restriction in $R\mu\le r$ as $\sum_{t}\rho_j(t)\mu(t)\le r_j$, where $\rho_j(t)$ is the entry of $R$ in row $j$ and the column of type $t$. The map $t\mapsto t^*$ preserves the potential-outcome schedules and the selected take-up $\vec D^{e}(\vec z)$, so $\theta(t^*)=\theta(t)$ and, by equilibrium invariance, $\rho_j(t^*)=\rho_j(t)$ for every $j$. Therefore $r_j\ \ge\ \sum_{t}\rho_j(t)\mu(t)=\sum_{t}\rho_j(t^*)\mu(t)=\sum_{t'}\rho_j(t')\mu^*(t')$
   for every $j$ and $\theta=\sum_{t}\theta(t)\mu(t)=\sum_{t}\theta(t^*)\mu(t)=\sum_{t'}\theta(t')\mu^*(t')$. Hence $\mu^*\in\mathcal M\cap\Delta(\mathcal T^*)$ attains $\theta$, and $\theta\in\Theta(\mathcal M\cap\Delta(\mathcal T^*),p)$.
   \end{proof}  

\section{Connection to \citet{kline2012bounds}}
\label{app:Kline_Tamer_2012}

\citet{kline2012bounds} (KT) study a two-player entry game with binary actions, no
instruments and no outcomes, and target the joint distribution of best responses under assumptions on the payoffs in \eqref{eq:defn_bestresponse}. This is a special case of our framework with $|\mathcal{Z}|=1$ and no outcome. Without instruments, each player's type reduces to a best-response pair
$\delta_{s_i}=(D(0;s_i),D(1;s_i))\in\{0,1\}^2$: the take-up decision when the other
player stays out and when the other player enters. There are four such pairs:
never enter ($00$), always enter ($11$), match the other player ($01$), and do the
opposite of the other player ($10$). A household type $t=(s_1,s_2,e)$ also records
the equilibrium $e$ the household plays in case of multiplicity.
The data are $p_{d_1d_2}=\mathbb{P}(D_{1g}=d_1,D_{2g}=d_2)$. In our notation, the KT
target is the joint distribution of best responses, $\theta=\mathbb{E}[\theta(T_g)]$ for each $(\delta_1,\delta_2)$, where
$\theta(t)=\mathbbm{1}\{(\delta_{s_1},\delta_{s_2})=(\delta_1,\delta_2)\}$. 

KT normalize the payoff from not entering to zero and assume that payoffs from entering are nonzero (KT Assumption 2.1), which implies that the best-response function $D(d';s_i)$ is uniquely defined. Their remaining assumptions imply restrictions over types.  Assumption 2.2 implies $D(0;s_i)\ge D(1;s_i)$, i.e., submodularity
(Section \ref{sec:supermodularity}); Assumption 2.3 is
$\delta_{s_1}=\delta_{s_2}$, i.e., individual symmetry (Definition
\ref{Def:Individual symmetry}), which coincides with equivariant symmetry when there
is no instrument; and Assumption 2.4 (pure strategies) is Assumption
\ref{AsNash}. Because several payoff orderings induce the same best response,
only these BRF-level restrictions have identifying content.

KT's target $\theta$ depends on types through best responses rather than potential outcomes, so
some of our previous irrelevance results do not apply: restrictions
that are irrelevant for fixed-allocation effects can become identifying here. To illustrate, impose KT Assumption 2.3. Under symmetry, only the diagonal of Panel A of Table
\ref{tab:stategy_equilibria_table} remains. Splitting the two pairs with multiple
equilibria by the equilibrium played gives the six types in Table
\ref{tab:kt_types}, with $\mu_\ell=\mathbb{P}(T_g=t_\ell)$ as in Section
\ref{sec:ident_estim_infer}.

\begin{table}[h]
\centering
\footnotesize
\begin{tabular}{clccc}
\toprule
Type & Both players' best response & $\delta_{s_1}=\delta_{s_2}$ & Nash equilibria & Equilibrium played \\
\midrule
$t_1$ & never enter                 & $00$ & $(0,0)$          & $(0,0)$ \\
$t_2$ & match the other player      & $01$ & $(0,0)$, $(1,1)$ & $(0,0)$ \\
$t_3$ & match the other player      & $01$ & $(0,0)$, $(1,1)$ & $(1,1)$ \\
$t_4$ & do the opposite             & $10$ & $(1,0)$, $(0,1)$ & $(1,0)$ \\
$t_5$ & do the opposite             & $10$ & $(1,0)$, $(0,1)$ & $(0,1)$ \\
$t_6$ & always enter                & $11$ & $(1,1)$          & $(1,1)$ \\
\bottomrule
\end{tabular}
\caption{Household types under symmetry in the setting of \citet{kline2012bounds}.}
\label{tab:kt_types}
\end{table}
The system $A\mu=p$ is then $p_{00}=\mu_1+\mu_2,$ $p_{10}=\mu_4,$ $p_{11}=\mu_3+\mu_6$, and $p_{01}=\mu_5$.
Symmetry is never refutable: any $p$ can be matched (a special case of Theorem \ref{thm:symmetry}).
It does not, however, point-identify $\mu$, because households where neither player
enters may never enter ($t_1$) or may match each other and coordinate on staying
out ($t_2$), and similarly for $p_{11}$. Submodularity rules out matching
behavior, which removes $t_2$ and $t_3$. Each observed cell then corresponds to
exactly one type, so the distribution of best responses and the equilibrium played
are point identified (KT, Corollary 2.4). By contrast, symmetry combined with
supermodularity or with dominance leaves only $t_1,t_2,t_3,t_6$ (or $t_1,t_6$),
which generate only $(0,0)$ and $(1,1)$, and is therefore refuted whenever
$p_{01}+p_{10}>0$.

\section{Test of Unitary Households}
\label{sec:unitary_households}
The unitary household is a workhorse model for fertility, consumption, and investment decisions, occupying a central role in the household literature. Decisions in a unitary household are often represented as maximizing a household welfare function $W(d_1,d_2,z_1,z_2)$. To embed this model in our game-theoretic framework, consider a game in which each household member has the same payoff function, $U_{ig}(d_1,d_2,z_1,z_2)=W(d_1,d_2,z_1,z_2)$ for $i=1,2$. The resulting common-payoff game is a special case of an exact potential game, with potential function $W$, and therefore admits at least one pure-strategy Nash equilibrium (PSNE) \citep{monderer1996potential}. Although exact potential games need not have identical payoff functions, they induce the same best-response functions (BRFs) as a common-payoff game in which both players maximize the potential function \citep{morris2004best}. Thus, while common-payoff games form a subset of exact potential games, the two classes generate the same BRFs. Moreover, any global maximizer of $W$ is a PSNE of this game (in Table \ref{fig:Unitary_households}, every blue cell contains a Nash equilibrium). Then, pairing the game with an equilibrium selection rule selecting a global maximizer embeds the unitary household model within our framework. In the two-player, two-action case, for a fixed $(z,z')$, the class of joint BRFs compatible with a unitary household coincides with the class of joint BRFs admitting at least one PSNE. Hence, absent further restrictions on how $(z_i,z_{-i})$ shift the ranking of preferences in $W(d_1,d_2,z_1,z_2)$, testing compatibility with the unitary model reduces to testing for the existence of a PSNE.
If $W(d_1,d_2,z_1,z_2)$ is completely unrestricted, there is no empirical content. In principle, \citet{dupas2019effect} use a (re-parametrized) regression of $D$ on a constant, $Z_1$ $Z_2$ and their interaction, using a pooled sample of men and women. This amounts to testing the empirical counterpart of the joint restriction $a = P(D=1|Z=(1,0))-P(D=1|Z=(0,1))=0$ and $b = P(D=1|Z=(1,1))-P(D=1|Z=(0,1))=0.$ Sufficient conditions that imply the validity of the test are $W(d_1,d_2,z_1,z_2) = W(d_1,d_2,z_1+z_2)       \implies  a=0$ and $W(d_1,d_2,z_1,z_2) = W(d_1,d_2,\max(z_1,z_2)) \implies a=b=0$.

These structural invariances at the household level are sufficient to imply the aggregate empirical restrictions. While they are not strictly necessary, since aggregate probabilities could theoretically balance out to zero due to unobserved population heterogeneity, they represent the natural structural conditions underlying the empirical test. Unlike in the 2 player, 2 action case, with three actions the class of joint BRFs compatible with a unitary household is a strict subset of the class admitting at least one PSNE. In Table \ref{fig:Unitary_households}, the light-purple cells correspond to joint BRFs that admit a PSNE but exhibit cycles of best-response adjustments, which are incompatible with a common potential function. 
\begin{table}[h]
\centering
\tiny
\setlength{\tabcolsep}{0pt}
\renewcommand{\arraystretch}{1}
\begin{tabular}{|>{\centering\arraybackslash}m{8mm}|*{27}{>{\centering\arraybackslash}m{4mm}|}}
\hline
\rule{0pt}{8mm}$D_1$ \textbackslash $D_2$ & \rotatebox{90}{$\mathbf{0,0,0}$} & \rotatebox{90}{$\mathbf{0,0,1}$} & \rotatebox{90}{$\mathbf{0,0,2}$} & \rotatebox{90}{$\mathbf{0,1,0}$} & \rotatebox{90}{$\mathbf{0,1,1}$} & \rotatebox{90}{$\mathbf{0,1,2}$} & \rotatebox{90}{$\mathbf{0,2,0}$} & \rotatebox{90}{$\mathbf{0,2,1}$} & \rotatebox{90}{$\mathbf{0,2,2}$} & \rotatebox{90}{$\mathbf{1,0,0}$} & \rotatebox{90}{$\mathbf{1,0,1}$} & \rotatebox{90}{$\mathbf{1,0,2}$} & \rotatebox{90}{$\mathbf{1,1,0}$} & \rotatebox{90}{$\mathbf{1,1,1}$} & \rotatebox{90}{$\mathbf{1,1,2}$} & \rotatebox{90}{$\mathbf{1,2,0}$} & \rotatebox{90}{$\mathbf{1,2,1}$} & \rotatebox{90}{$\mathbf{1,2,2}$} & \rotatebox{90}{$\mathbf{2,0,0}$} & \rotatebox{90}{$\mathbf{2,0,1}$} & \rotatebox{90}{$\mathbf{2,0,2}$} & \rotatebox{90}{$\mathbf{2,1,0}$} & \rotatebox{90}{$\mathbf{2,1,1}$} & \rotatebox{90}{$\mathbf{2,1,2}$} & \rotatebox{90}{$\mathbf{2,2,0}$} & \rotatebox{90}{$\mathbf{2,2,1}$} & \rotatebox{90}{$\mathbf{2,2,2}$} \\ \hline
\rule{0pt}{4mm}$\mathbf{0,0,0}$ & \cellcolor{lightblue} & \cellcolor{lightblue} & \cellcolor{lightblue} & \cellcolor{lightblue} & \cellcolor{lightblue} & \cellcolor{lightblue} & \cellcolor{lightblue} & \cellcolor{lightblue} & \cellcolor{lightblue} & \cellcolor{lightblue} & \cellcolor{lightblue} & \cellcolor{lightblue} & \cellcolor{lightblue} & \cellcolor{lightblue} & \cellcolor{lightblue} & \cellcolor{lightblue} & \cellcolor{lightblue} & \cellcolor{lightblue} & \cellcolor{lightblue} & \cellcolor{lightblue} & \cellcolor{lightblue} & \cellcolor{lightblue} & \cellcolor{lightblue} & \cellcolor{lightblue} & \cellcolor{lightblue} & \cellcolor{lightblue} & \cellcolor{lightblue} \\ \hline
\rule{0pt}{4mm}$\mathbf{0,0,1}$ & \cellcolor{lightblue} & \cellcolor{lightblue} & \cellcolor{lightblue} & \cellcolor{lightblue} & \cellcolor{lightblue} & \cellcolor{lightblue} & \cellcolor{lightblue} & \cellcolor{lightblue} & \cellcolor{lightblue} & \cellcolor{lightblue} & \cellcolor{lightblue} & \cellcolor{lightblue} & \cellcolor{lightblue} & \cellcolor{lightblue} & \cellcolor{lightblue} & \cellcolor{lightblue} & \cellcolor{lightblue} & \cellcolor{lightblue} &  &  &  &  &  &  & \cellcolor{lightblue} & \cellcolor{lightblue} & \cellcolor{lightblue} \\ \hline
\rule{0pt}{4mm}$\mathbf{0,0,2}$ & \cellcolor{lightblue} & \cellcolor{lightblue} & \cellcolor{lightblue} & \cellcolor{lightblue} & \cellcolor{lightblue} & \cellcolor{lightblue} & \cellcolor{lightblue} & \cellcolor{lightblue} & \cellcolor{lightblue} & \cellcolor{lightblue} & \cellcolor{lightblue} & \cellcolor{lightblue} & \cellcolor{lightblue} & \cellcolor{lightblue} & \cellcolor{lightblue} & \cellcolor{lightblue} & \cellcolor{lightblue} & \cellcolor{lightblue} &  &  & \cellcolor{lightblue} &  &  & \cellcolor{lightblue} &  &  & \cellcolor{lightblue} \\ \hline
\rule{0pt}{4mm}$\mathbf{0,1,0}$ & \cellcolor{lightblue} & \cellcolor{lightblue} & \cellcolor{lightblue} & \cellcolor{lightblue} & \cellcolor{lightblue} & \cellcolor{lightblue} & \cellcolor{lightblue} & \cellcolor{lightblue} & \cellcolor{lightblue} &  &  &  & \cellcolor{lightblue} & \cellcolor{lightblue} & \cellcolor{lightblue} &  &  &  & \cellcolor{lightblue} & \cellcolor{lightblue} & \cellcolor{lightblue} & \cellcolor{lightblue} & \cellcolor{lightblue} & \cellcolor{lightblue} & \cellcolor{lightblue} & \cellcolor{lightblue} & \cellcolor{lightblue} \\ \hline
\rule{0pt}{4mm}$\mathbf{0,1,1}$ & \cellcolor{lightblue} & \cellcolor{lightblue} & \cellcolor{lightblue} & \cellcolor{lightblue} & \cellcolor{lightblue} & \cellcolor{lightblue} & \cellcolor{lightblue} & \cellcolor{lightblue} & \cellcolor{lightblue} &  &  &  & \cellcolor{lightblue} & \cellcolor{lightblue} & \cellcolor{lightblue} & \cellcolor{lightblue} & \cellcolor{lightblue} & \cellcolor{lightblue} &  &  &  & \cellcolor{lightblue} & \cellcolor{lightblue} & \cellcolor{lightblue} & \cellcolor{lightblue} & \cellcolor{lightblue} & \cellcolor{lightblue} \\ \hline
\rule{0pt}{4mm}$\mathbf{0,1,2}$ & \cellcolor{lightblue} & \cellcolor{lightblue} & \cellcolor{lightblue} & \cellcolor{lightblue} & \cellcolor{lightblue} & \cellcolor{lightblue} & \cellcolor{lightblue} & \cellcolor{lightpurple} & \cellcolor{lightblue} &  &  & \cellcolor{lightpurple} & \cellcolor{lightblue} & \cellcolor{lightblue} & \cellcolor{lightblue} &  &  & \cellcolor{lightblue} &  &  & \cellcolor{lightblue} & \cellcolor{lightpurple} & \cellcolor{lightblue} & \cellcolor{lightblue} &  &  & \cellcolor{lightblue} \\ \hline
\rule{0pt}{4mm}$\mathbf{0,2,0}$ & \cellcolor{lightblue} & \cellcolor{lightblue} & \cellcolor{lightblue} & \cellcolor{lightblue} & \cellcolor{lightblue} & \cellcolor{lightblue} & \cellcolor{lightblue} & \cellcolor{lightblue} & \cellcolor{lightblue} &  & \cellcolor{lightblue} &  &  & \cellcolor{lightblue} &  &  & \cellcolor{lightblue} &  & \cellcolor{lightblue} & \cellcolor{lightblue} & \cellcolor{lightblue} & \cellcolor{lightblue} & \cellcolor{lightblue} & \cellcolor{lightblue} & \cellcolor{lightblue} & \cellcolor{lightblue} & \cellcolor{lightblue} \\ \hline
\rule{0pt}{4mm}$\mathbf{0,2,1}$ & \cellcolor{lightblue} & \cellcolor{lightblue} & \cellcolor{lightblue} & \cellcolor{lightblue} & \cellcolor{lightblue} & \cellcolor{lightpurple} & \cellcolor{lightblue} & \cellcolor{lightblue} & \cellcolor{lightblue} &  & \cellcolor{lightblue} &  &  & \cellcolor{lightblue} &  & \cellcolor{lightpurple} & \cellcolor{lightblue} & \cellcolor{lightblue} &  & \cellcolor{lightpurple} &  &  & \cellcolor{lightblue} &  & \cellcolor{lightblue} & \cellcolor{lightblue} & \cellcolor{lightblue} \\ \hline
\rule{0pt}{4mm}$\mathbf{0,2,2}$ & \cellcolor{lightblue} & \cellcolor{lightblue} & \cellcolor{lightblue} & \cellcolor{lightblue} & \cellcolor{lightblue} & \cellcolor{lightblue} & \cellcolor{lightblue} & \cellcolor{lightblue} & \cellcolor{lightblue} &  & \cellcolor{lightblue} & \cellcolor{lightblue} &  & \cellcolor{lightblue} & \cellcolor{lightblue} &  & \cellcolor{lightblue} & \cellcolor{lightblue} &  & \cellcolor{lightblue} & \cellcolor{lightblue} &  & \cellcolor{lightblue} & \cellcolor{lightblue} &  & \cellcolor{lightblue} & \cellcolor{lightblue} \\ \hline
\rule{0pt}{4mm}$\mathbf{1,0,0}$ & \cellcolor{lightblue} & \cellcolor{lightblue} & \cellcolor{lightblue} &  &  &  &  &  &  & \cellcolor{lightblue} & \cellcolor{lightblue} & \cellcolor{lightblue} & \cellcolor{lightblue} & \cellcolor{lightblue} & \cellcolor{lightblue} & \cellcolor{lightblue} & \cellcolor{lightblue} & \cellcolor{lightblue} & \cellcolor{lightblue} & \cellcolor{lightblue} & \cellcolor{lightblue} & \cellcolor{lightblue} & \cellcolor{lightblue} & \cellcolor{lightblue} & \cellcolor{lightblue} & \cellcolor{lightblue} & \cellcolor{lightblue} \\ \hline
\rule{0pt}{4mm}$\mathbf{1,0,1}$ & \cellcolor{lightblue} & \cellcolor{lightblue} & \cellcolor{lightblue} &  &  &  & \cellcolor{lightblue} & \cellcolor{lightblue} & \cellcolor{lightblue} & \cellcolor{lightblue} & \cellcolor{lightblue} & \cellcolor{lightblue} & \cellcolor{lightblue} & \cellcolor{lightblue} & \cellcolor{lightblue} & \cellcolor{lightblue} & \cellcolor{lightblue} & \cellcolor{lightblue} & \cellcolor{lightblue} & \cellcolor{lightblue} & \cellcolor{lightblue} &  &  &  & \cellcolor{lightblue} & \cellcolor{lightblue} & \cellcolor{lightblue} \\ \hline
\rule{0pt}{4mm}$\mathbf{1,0,2}$ & \cellcolor{lightblue} & \cellcolor{lightblue} & \cellcolor{lightblue} &  &  & \cellcolor{lightpurple} &  &  & \cellcolor{lightblue} & \cellcolor{lightblue} & \cellcolor{lightblue} & \cellcolor{lightblue} & \cellcolor{lightblue} & \cellcolor{lightblue} & \cellcolor{lightblue} & \cellcolor{lightpurple} & \cellcolor{lightblue} & \cellcolor{lightblue} & \cellcolor{lightblue} & \cellcolor{lightpurple} & \cellcolor{lightblue} &  &  & \cellcolor{lightblue} &  &  & \cellcolor{lightblue} \\ \hline
\rule{0pt}{4mm}$\mathbf{1,1,0}$ & \cellcolor{lightblue} & \cellcolor{lightblue} & \cellcolor{lightblue} & \cellcolor{lightblue} & \cellcolor{lightblue} & \cellcolor{lightblue} &  &  &  & \cellcolor{lightblue} & \cellcolor{lightblue} & \cellcolor{lightblue} & \cellcolor{lightblue} & \cellcolor{lightblue} & \cellcolor{lightblue} &  &  &  & \cellcolor{lightblue} & \cellcolor{lightblue} & \cellcolor{lightblue} & \cellcolor{lightblue} & \cellcolor{lightblue} & \cellcolor{lightblue} & \cellcolor{lightblue} & \cellcolor{lightblue} & \cellcolor{lightblue} \\ \hline
\rule{0pt}{4mm}$\mathbf{1,1,1}$ & \cellcolor{lightblue} & \cellcolor{lightblue} & \cellcolor{lightblue} & \cellcolor{lightblue} & \cellcolor{lightblue} & \cellcolor{lightblue} & \cellcolor{lightblue} & \cellcolor{lightblue} & \cellcolor{lightblue} & \cellcolor{lightblue} & \cellcolor{lightblue} & \cellcolor{lightblue} & \cellcolor{lightblue} & \cellcolor{lightblue} & \cellcolor{lightblue} & \cellcolor{lightblue} & \cellcolor{lightblue} & \cellcolor{lightblue} & \cellcolor{lightblue} & \cellcolor{lightblue} & \cellcolor{lightblue} & \cellcolor{lightblue} & \cellcolor{lightblue} & \cellcolor{lightblue} & \cellcolor{lightblue} & \cellcolor{lightblue} & \cellcolor{lightblue} \\ \hline
\rule{0pt}{4mm}$\mathbf{1,1,2}$ & \cellcolor{lightblue} & \cellcolor{lightblue} & \cellcolor{lightblue} & \cellcolor{lightblue} & \cellcolor{lightblue} & \cellcolor{lightblue} &  &  & \cellcolor{lightblue} & \cellcolor{lightblue} & \cellcolor{lightblue} & \cellcolor{lightblue} & \cellcolor{lightblue} & \cellcolor{lightblue} & \cellcolor{lightblue} &  &  & \cellcolor{lightblue} & \cellcolor{lightblue} & \cellcolor{lightblue} & \cellcolor{lightblue} & \cellcolor{lightblue} & \cellcolor{lightblue} & \cellcolor{lightblue} &  &  & \cellcolor{lightblue} \\ \hline
\rule{0pt}{4mm}$\mathbf{1,2,0}$ & \cellcolor{lightblue} & \cellcolor{lightblue} & \cellcolor{lightblue} &  & \cellcolor{lightblue} &  &  & \cellcolor{lightpurple} &  & \cellcolor{lightblue} & \cellcolor{lightblue} & \cellcolor{lightpurple} &  & \cellcolor{lightblue} &  &  & \cellcolor{lightblue} &  & \cellcolor{lightblue} & \cellcolor{lightblue} & \cellcolor{lightblue} & \cellcolor{lightpurple} & \cellcolor{lightblue} & \cellcolor{lightblue} & \cellcolor{lightblue} & \cellcolor{lightblue} & \cellcolor{lightblue} \\ \hline
\rule{0pt}{4mm}$\mathbf{1,2,1}$ & \cellcolor{lightblue} & \cellcolor{lightblue} & \cellcolor{lightblue} &  & \cellcolor{lightblue} &  & \cellcolor{lightblue} & \cellcolor{lightblue} & \cellcolor{lightblue} & \cellcolor{lightblue} & \cellcolor{lightblue} & \cellcolor{lightblue} &  & \cellcolor{lightblue} &  & \cellcolor{lightblue} & \cellcolor{lightblue} & \cellcolor{lightblue} & \cellcolor{lightblue} & \cellcolor{lightblue} & \cellcolor{lightblue} &  & \cellcolor{lightblue} &  & \cellcolor{lightblue} & \cellcolor{lightblue} & \cellcolor{lightblue} \\ \hline
\rule{0pt}{4mm}$\mathbf{1,2,2}$ & \cellcolor{lightblue} & \cellcolor{lightblue} & \cellcolor{lightblue} &  & \cellcolor{lightblue} & \cellcolor{lightblue} &  & \cellcolor{lightblue} & \cellcolor{lightblue} & \cellcolor{lightblue} & \cellcolor{lightblue} & \cellcolor{lightblue} &  & \cellcolor{lightblue} & \cellcolor{lightblue} &  & \cellcolor{lightblue} & \cellcolor{lightblue} & \cellcolor{lightblue} & \cellcolor{lightblue} & \cellcolor{lightblue} &  & \cellcolor{lightblue} & \cellcolor{lightblue} &  & \cellcolor{lightblue} & \cellcolor{lightblue} \\ \hline
\rule{0pt}{4mm}$\mathbf{2,0,0}$ & \cellcolor{lightblue} &  &  & \cellcolor{lightblue} &  &  & \cellcolor{lightblue} &  &  & \cellcolor{lightblue} & \cellcolor{lightblue} & \cellcolor{lightblue} & \cellcolor{lightblue} & \cellcolor{lightblue} & \cellcolor{lightblue} & \cellcolor{lightblue} & \cellcolor{lightblue} & \cellcolor{lightblue} & \cellcolor{lightblue} & \cellcolor{lightblue} & \cellcolor{lightblue} & \cellcolor{lightblue} & \cellcolor{lightblue} & \cellcolor{lightblue} & \cellcolor{lightblue} & \cellcolor{lightblue} & \cellcolor{lightblue} \\ \hline
\rule{0pt}{4mm}$\mathbf{2,0,1}$ & \cellcolor{lightblue} &  &  & \cellcolor{lightblue} &  &  & \cellcolor{lightblue} & \cellcolor{lightpurple} & \cellcolor{lightblue} & \cellcolor{lightblue} & \cellcolor{lightblue} & \cellcolor{lightpurple} & \cellcolor{lightblue} & \cellcolor{lightblue} & \cellcolor{lightblue} & \cellcolor{lightblue} & \cellcolor{lightblue} & \cellcolor{lightblue} & \cellcolor{lightblue} &  &  & \cellcolor{lightpurple} &  &  & \cellcolor{lightblue} & \cellcolor{lightblue} & \cellcolor{lightblue} \\ \hline
\rule{0pt}{4mm}$\mathbf{2,0,2}$ & \cellcolor{lightblue} &  & \cellcolor{lightblue} & \cellcolor{lightblue} &  & \cellcolor{lightblue} & \cellcolor{lightblue} &  & \cellcolor{lightblue} & \cellcolor{lightblue} & \cellcolor{lightblue} & \cellcolor{lightblue} & \cellcolor{lightblue} & \cellcolor{lightblue} & \cellcolor{lightblue} & \cellcolor{lightblue} & \cellcolor{lightblue} & \cellcolor{lightblue} & \cellcolor{lightblue} &  & \cellcolor{lightblue} & \cellcolor{lightblue} &  & \cellcolor{lightblue} & \cellcolor{lightblue} &  & \cellcolor{lightblue} \\ \hline
\rule{0pt}{4mm}$\mathbf{2,1,0}$ & \cellcolor{lightblue} &  &  & \cellcolor{lightblue} & \cellcolor{lightblue} & \cellcolor{lightpurple} & \cellcolor{lightblue} &  &  & \cellcolor{lightblue} &  &  & \cellcolor{lightblue} & \cellcolor{lightblue} & \cellcolor{lightblue} & \cellcolor{lightpurple} &  &  & \cellcolor{lightblue} & \cellcolor{lightpurple} & \cellcolor{lightblue} & \cellcolor{lightblue} & \cellcolor{lightblue} & \cellcolor{lightblue} & \cellcolor{lightblue} & \cellcolor{lightblue} & \cellcolor{lightblue} \\ \hline
\rule{0pt}{4mm}$\mathbf{2,1,1}$ & \cellcolor{lightblue} &  &  & \cellcolor{lightblue} & \cellcolor{lightblue} & \cellcolor{lightblue} & \cellcolor{lightblue} & \cellcolor{lightblue} & \cellcolor{lightblue} & \cellcolor{lightblue} &  &  & \cellcolor{lightblue} & \cellcolor{lightblue} & \cellcolor{lightblue} & \cellcolor{lightblue} & \cellcolor{lightblue} & \cellcolor{lightblue} & \cellcolor{lightblue} &  &  & \cellcolor{lightblue} & \cellcolor{lightblue} & \cellcolor{lightblue} & \cellcolor{lightblue} & \cellcolor{lightblue} & \cellcolor{lightblue} \\ \hline
\rule{0pt}{4mm}$\mathbf{2,1,2}$ & \cellcolor{lightblue} &  & \cellcolor{lightblue} & \cellcolor{lightblue} & \cellcolor{lightblue} & \cellcolor{lightblue} & \cellcolor{lightblue} &  & \cellcolor{lightblue} & \cellcolor{lightblue} &  & \cellcolor{lightblue} & \cellcolor{lightblue} & \cellcolor{lightblue} & \cellcolor{lightblue} & \cellcolor{lightblue} &  & \cellcolor{lightblue} & \cellcolor{lightblue} &  & \cellcolor{lightblue} & \cellcolor{lightblue} & \cellcolor{lightblue} & \cellcolor{lightblue} & \cellcolor{lightblue} &  & \cellcolor{lightblue} \\ \hline
\rule{0pt}{4mm}$\mathbf{2,2,0}$ & \cellcolor{lightblue} & \cellcolor{lightblue} &  & \cellcolor{lightblue} & \cellcolor{lightblue} &  & \cellcolor{lightblue} & \cellcolor{lightblue} &  & \cellcolor{lightblue} & \cellcolor{lightblue} &  & \cellcolor{lightblue} & \cellcolor{lightblue} &  & \cellcolor{lightblue} & \cellcolor{lightblue} &  & \cellcolor{lightblue} & \cellcolor{lightblue} & \cellcolor{lightblue} & \cellcolor{lightblue} & \cellcolor{lightblue} & \cellcolor{lightblue} & \cellcolor{lightblue} & \cellcolor{lightblue} & \cellcolor{lightblue} \\ \hline
\rule{0pt}{4mm}$\mathbf{2,2,1}$ & \cellcolor{lightblue} & \cellcolor{lightblue} &  & \cellcolor{lightblue} & \cellcolor{lightblue} &  & \cellcolor{lightblue} & \cellcolor{lightblue} & \cellcolor{lightblue} & \cellcolor{lightblue} & \cellcolor{lightblue} &  & \cellcolor{lightblue} & \cellcolor{lightblue} &  & \cellcolor{lightblue} & \cellcolor{lightblue} & \cellcolor{lightblue} & \cellcolor{lightblue} & \cellcolor{lightblue} &  & \cellcolor{lightblue} & \cellcolor{lightblue} &  & \cellcolor{lightblue} & \cellcolor{lightblue} & \cellcolor{lightblue} \\ \hline
\rule{0pt}{4mm}$\mathbf{2,2,2}$ & \cellcolor{lightblue} & \cellcolor{lightblue} & \cellcolor{lightblue} & \cellcolor{lightblue} & \cellcolor{lightblue} & \cellcolor{lightblue} & \cellcolor{lightblue} & \cellcolor{lightblue} & \cellcolor{lightblue} & \cellcolor{lightblue} & \cellcolor{lightblue} & \cellcolor{lightblue} & \cellcolor{lightblue} & \cellcolor{lightblue} & \cellcolor{lightblue} & \cellcolor{lightblue} & \cellcolor{lightblue} & \cellcolor{lightblue} & \cellcolor{lightblue} & \cellcolor{lightblue} & \cellcolor{lightblue} & \cellcolor{lightblue} & \cellcolor{lightblue} & \cellcolor{lightblue} & \cellcolor{lightblue} & \cellcolor{lightblue} & \cellcolor{lightblue} \\ \hline
\end{tabular}   
\caption{\footnotesize Joint BRF diagram with two players and three actions, $d\in\{0,1,2\}$. White cells admit no PSNE. Light-blue cells are compatible with a unitary household, whereas light-purple cells admit at least one PSNE but are not compatible with a unitary household.}
\label{fig:Unitary_households}
\end{table}

\section{Empirical Analysis of Female-headed Households}\label{appendix:female_headed}
We note that in \cite{dupas2019effect} the vouchers were also offered to single female headed households. In this case, there are no spillovers the 2SLS analysis is based on $Y_{ ig} = \beta_{0} + \beta_{1} D_{ ig} + U_{ ig},$ and $D_{ ig} = \pi_{0} + \pi_{1} Z_{ ig}+ V_{ ig}.$ Under the no-defiers assumption, the 2SLS estimator can be interpreted as a LATE. Table  \ref{tab:2sls_female_headed_by_round} shows the results. Since this case fits exactly the case analyzed in \cite{balke1997bounds} we can bound the ATE, which yields $[-0.078,0.277]$. 

\begin{table}[H]
\scriptsize
\centering
\begin{tabular}{lcccc} \hline
    & (1) & (2) & (3) & (4) \\
    & Round 3 & Round 4 & Round 5 & Round 6 \\ \hline
    $D_{ig}$ 
    & -0.012 & 0.087 & 0.008 & 0.091 \\
    & (0.069) & (0.067) & (0.068) & (0.065) \\

    Constant 
    & 0.266*** & 0.215*** & 0.239*** & 0.198*** \\
    & (0.032) & (0.030) & (0.032) & (0.030) \\

    &  &  &  &  \\

    Observations 
    & 382 & 379 & 368 & 362 \\ \hline

    \multicolumn{5}{c}{ } \\
    \multicolumn{5}{c}{ *** p$<$0.01, ** p$<$0.05, * p$<$0.1 } \\
\end{tabular}
\caption{\footnotesize 2SLS estimates for female-headed households: Own Enterprise outcome by round. Heteroskedastic-robust standard errors in parentheses.}
\label{tab:2sls_female_headed_by_round}
\end{table}

\normalsize
\section{Inference} \label{AppInference} For inference, we construct confidence sets by inverting a test of the solvability of the corresponding system of linear equalities and inequalities, following \cite{goff2025inferencevaluelinearprogram}. Let $\widetilde R := [-I, R]'; \ \widetilde r := [0, r]'$. 
The inequality \( \widetilde R\mu\leq\widetilde r \) therefore collects the nonnegativity restrictions \( \mu\geq0 \) and the behavioral restrictions \( R\mu\leq r \). Define the population feasible set as $\mathcal P(A,p,\widetilde R,\widetilde r) := \left\{ \mu\in\mathbb R^{|\mathcal T|}: A\mu=p,\; 1^\top\mu=1,\; \widetilde R\mu\leq\widetilde r \right\}$. 
Here, $\mu$ is the vector of probabilities assigned to the latent response types in $\mathcal T$. The equality $A\mu=p$ links these latent-type probabilities to the distribution of the observed outcomes and treatments, while $1^\top\mu=1$ requires the latent-type probabilities to sum to one. For a given objective vector $c$, the population identified set for the scalar parameter $\theta=c^\top\mu$ is $\Theta_I := \left\{ c^\top\mu: \mu\in \mathcal P(A,p,\widetilde R,\widetilde r) \right\}$.
Because the feasible set is convex and the objective is linear, $\Theta_I$ is an interval whenever the feasible set is nonempty. Its endpoints are $\underline\theta := \min_{\mu\in\mathcal P(A,p,\widetilde R,\widetilde r)} c^\top\mu, \qquad \overline\theta := \max_{\mu\in\mathcal P(A,p,\widetilde R,\widetilde r)} c^\top\mu$
so that $\Theta_I = \left[ \underline\theta,\overline\theta \right]$. 
The endpoints are scalar values of linear programs, whereas their optimizers are latent-type probability vectors. 

\subsection{Inference by Inversion of a Solvability Test} \label{AppInferenceInversion} 
Let $\theta^\star$ denote the true scalar parameter. Partial identification implies the maintained restriction $\theta^\star\in\Theta_I$. For each candidate value $\theta_0$, consider the augmented system \begin{equation} A\mu=p, \qquad 1^\top\mu=1, \qquad \widetilde R\mu\leq\widetilde r, \qquad c^\top\mu=\theta_0. \label{eq:augmented-system-theta} \end{equation} Membership of $\theta_0$ in the population identified set is equivalent to the solvability of this system. In particular, $\theta_0\in\Theta_I \quad\Longleftrightarrow\quad \exists\mu\in\mathbb R^{|\mathcal T|}$ satisfying \eqref{eq:augmented-system-theta}. We therefore test \begin{equation} H_0(\theta_0): \exists\mu\in\mathbb R^{|\mathcal T|} \text{ such that } A\mu=p,\; 1^\top\mu=1,\; \widetilde R\mu\leq\widetilde r,\; c^\top\mu=\theta_0. \label{eq:solvability-null} \end{equation} The alternative is that no such vector $\mu$ exists. Let $\mathcal G_\theta = \left\{ \theta_{\min}, \theta_{\min}+\Delta_\theta, \ldots, \theta_{\max} \right\}$ denote the prespecified grid of candidate values, where $\Delta_\theta>0$ is the grid spacing. The limits $\theta_{\min}$ and $\theta_{\max}$ are selected so that the grid contains the economically admissible range of the parameter. Let $\varphi_G(\theta_0)\in\{0,1\}$
denote the decision rule for the solvability test, where $\varphi_G(\theta_0)=1$ indicates rejection of $H_0(\theta_0)$ at significance level $\alpha$. The grid-inverted confidence set is $\widehat{\mathcal C}_{1-\alpha} := \left\{ \theta_0\in\mathcal G_\theta: \varphi_G(\theta_0)=0 \right\}$.
Thus, $\widehat{\mathcal C}_{1-\alpha}$ contains the candidate values whose associated linear systems are not rejected as unsolvable. The inferential target is the true scalar value $\theta^\star$, rather than the entire identified interval. In the empirical specifications reported in the paper, the nonrejected grid points form a contiguous subset of $\mathcal G_\theta$. We therefore report 
$\left[ \min\widehat{\mathcal C}_{1-\alpha}, \; \max\widehat{\mathcal C}_{1-\alpha} \right]$. 
The endpoints reported in the tables inherit the numerical resolution of the grid. 

\subsection{Sampling Uncertainty and Estimated Restrictions} \label{AppInferenceSampling} 

Sampling uncertainty is indexed by a low-dimensional vector of primitive inputs $\eta := (p,\psi), \sqrt G \left( \widehat\eta-\eta \right) \rightsquigarrow \mathcal N(0,\Sigma_\eta)$, where $G$ is the number of sampled households. The vector $\psi$ collects any estimated coefficients entering sample-dependent constraints. The arrays in the solvability problem are deterministic functions of these primitive inputs: $A=A(\psi), R=R(\psi), c=c(\psi)$. In the baseline specifications, the coefficient matrices and objective vector are fixed, and sampling uncertainty enters through the estimated probability vector $\widehat p$. In those specifications, $\psi$ is absent. When a behavioral restriction contains estimated coefficients, the corresponding rows of the constraint system are reestimated jointly with $p$. For example, under monotone treatment selection, some rows of the restriction matrix depend on $\kappa := \frac{\mathbb P(D_{ig}=0)} {\mathbb P(D_{ig}=1)}$ with sample analogue $\widehat\kappa := \frac{\widehat{\mathbb P}(D_{ig}=0)} {\widehat{\mathbb P}(D_{ig}=1)}$, and the corresponding restriction matrix is $\widehat{R} = R(\widehat\kappa)$ and $\widetilde{R}' = [-I,\widehat{R}']$. Thus, in MTS specifications, the primitive vector includes $\kappa$, and the bootstrap reproduces the joint sampling variation in $\widehat p$ and $\widehat\kappa$. For household bootstrap draw $b$, define $\zeta_{\eta,G}^{*(b)} := \sqrt G \left( \widehat\eta^{*(b)}-\widehat\eta \right)$. Each draw re-estimates all sample-dependent primitive inputs and induces the corresponding bootstrap perturbation of the equality and inequality systems.  The solvability statistic and its bootstrap critical value are computed for each candidate value $\theta_0\in\mathcal G_\theta$ following the solvability-testing procedure in \cite{goff2025inferencevaluelinearprogram}. Resampling is performed at the household level, which is the sampling unit in the empirical analysis. For a candidate value $\theta_0$, write the equalities in \eqref{eq:solvability-null} as
$\mathbf A\mu=\widehat b(\theta_0)$, where $\mathbf A$ stacks $A$ and $c^\top$ and
$\widehat b(\theta_0)=(\widehat p^\top,\theta_0)^\top$; the restriction $1^\top\mu=1$ is implied by
$A\mu=p$, because the cells for each offer vector sum to one. The inequalities are $\mu\ge0$ and
$R\mu\le r$. By Farkas' lemma, the system is solvable if and only if the value of the linear program
\begin{equation}
\widehat T_G(\theta_0)
:=\max_{\lambda,\ \nu\ge0}\Big\{\lambda^\top\widehat b(\theta_0)-\nu^\top r:\
\mathbf A^\top\lambda-R^\top\nu\le 0,\ \|\widehat\Lambda\lambda\|_1\le1\Big\}
\label{eq:gm-statistic}
\end{equation}
equals zero. Here $\lambda$ is the multiplier on the equalities and $\nu$ the multiplier on the
behavioral inequalities; when no behavioral inequalities are imposed, $\nu$ and $R$ drop out.
The normalization is $\widehat\Lambda=\operatorname{diag}(\widehat\sigma_1,\ldots,\widehat\sigma_K,0)$,
with $\widehat\sigma_k^2=G^{-1}\sum_{g}\big(m_k(W_g)-\widehat p_k\big)^2$,
$m_k(W_g)=1_{\{W_g=w_k\}}/\widehat\pi_{z_k}$, and $\widehat\pi_z$ the share of households with
offer vector $z$. This is the estimated asymptotic standard deviation of $\sqrt G(\widehat p-p)$, as in
\citet[eq.~(3.2)]{goff2025inferencevaluelinearprogram}. The row for $\theta_0$ and the rows of
$R\mu\le r$ receive zero weight because they are not estimated. We reject $H_0(\theta_0)$ when
$\sqrt G\,\widehat T_G(\theta_0)$ exceeds the $(1-\alpha)$ quantile of the bootstrap statistic
\begin{align}
v_G^{*(b)}(\theta_0)
&:=\max_{(\lambda,\nu)\in\widehat\Delta_G(\theta_0)}
\lambda^\top\sqrt G\big(\widehat b^{*(b)}(\theta_0)-\widehat b(\theta_0)\big),
\label{eq:gm-bootstrap} \\
\widehat\Delta_G(\theta_0)
&:= \begin{aligned}[t]
\Big\{(\lambda,\nu):\ \nu\ge0,\ \mathbf A^\top\lambda-R^\top\nu\le0,\ \|\widehat\Lambda\lambda\|_1\le1,\notag\\
\lambda^\top\widehat b(\theta_0)-\nu^\top r\ \ge\ \widehat T_G(\theta_0)-\kappa_G/\sqrt G\Big\}.
\end{aligned}
\label{eq:gm-near-optimal}
\end{align}
Since $r$ is not estimated, only the equality block is perturbed in the bootstrap. This is the
fixed-coefficient procedure of \citet[Section 3.1]{goff2025inferencevaluelinearprogram}.

When $R$ contains estimated coefficients, as under monotone treatment selection, where
$\widehat R=R(\widehat\kappa_m)$ and $r=0$, the statistic is \eqref{eq:gm-statistic} with $R$ replaced
by $\widehat R$, and the normalization also weights each row of $\widehat R\mu\le0$ by the largest
estimated standard deviation of its coefficients \citep[eq.~(3.5)]{goff2025inferencevaluelinearprogram}.
Critical values then follow the estimated-coefficient procedure in their Section 3.2, which perturbs
$\widehat R$ as well as $\widehat b$ in each bootstrap draw.

The tuning parameter $\kappa_G$ governs the set $\widehat\Delta_G(\theta_0)$ of near-optimal dual
solutions. The (asymptotic) conditions required are $\kappa_G\to\infty$ and $\kappa_G/\sqrt G\to0$,
and, following \citet[Remark~4.21]{goff2025inferencevaluelinearprogram}, we use
$\kappa_G=\sqrt{\log G}\approx 2.45$. The estimated-coefficient version involves three such sequences,
one relaxing the dual constraint, one defining near-optimality as above, and one in the auxiliary
quadratic program; we set all three equal to $\kappa_G$. The solvability results in Tables
\ref{tab:solvability_tests} and \ref{tab:mts_tests} are unchanged under the alternative
$\kappa_G=\log G\approx6.02$.

\subsection{Implementation} \label{AppGMInf}
Before the test, we impose hard behavioral restrictions, such as strong strategic neutrality, by
deleting the columns of $A$ and $c$ that correspond to excluded types, so that they do not enter $R$.
For the offer-mediated effects, we also merge types whose columns coincide in both $A$ and $c$.
Neither operation changes whether \eqref{eq:solvability-null} has a solution. With these deletions,
the confidence intervals in the paper involve no behavioral inequalities, so $\nu$ and $R$ are absent
from \eqref{eq:gm-statistic} and the fixed-coefficient procedure applies. The only inequalities with
estimated coefficients are the MTS rows, which are used only in the solvability tests of
Table \ref{tab:mts_tests}, with the estimated-coefficient procedure. Within each household bootstrap
draw, $\widehat p$ (and $\widehat\kappa_m$ in the MTS tests) is re-estimated, while the normalization
is held at its full-sample value. The same system, bootstrap design, tuning sequence, and
significance level are used at every point of $\mathcal G_\theta$.